\documentclass[11pt]{article}
\usepackage[margin=1.25in]{geometry}
\pdfoutput=1
\usepackage{amssymb, amsmath, amsthm, verbatim, mathrsfs}
\usepackage{thmtools, thm-restate}
\usepackage{enumerate}
\usepackage{framed}
\usepackage{bbm}
\usepackage{soul}

\usepackage{graphicx}
\usepackage{graphpap}
\usepackage{subfig}

\usepackage{fullpage}
\usepackage[round]{natbib}

\usepackage{tikz}
\usetikzlibrary{automata, positioning}

\usepackage{todonotes}

\newtheorem{theorem}{Theorem}
\newtheorem{definition}{Definition}

\newtheorem{proposition}{Proposition}
\newtheorem{observation}{Observation}
\newtheorem{lemma}{Lemma}
\newtheorem*{remark}{Remark}
\newtheorem{corollary}{Corollary}

\newtheorem{fact}{Fact}

\usepackage{soul}
\def\P{{\bf P}}
\def\E{{\bf E}}
\def\R{\mathbb{R}}
\def\N{\mathbb{N}}

\def\BH{{\mathcal{B}_H}}
\def\BE{{\mathcal{B}_E}}
\def\Cd{{\mathcal{C}(d)}}

\def\tCd{{\hat{\mathcal{C}}(d)}}
\def\tCdm{{\widetilde{\mathcal{C}}(d, \mu)}}
\def\tBE{{\widetilde{\mathcal{B}_E}}}
\def\cP{\mathcal{P}}

\def\I{{\mathbb{I}}}

\newtheoremstyle{named}{}{}{\itshape}{}{\bfseries}{.}{.5em}{\thmnote{#3  }#1}
\theoremstyle{named}
\usepackage{setspace}
\usepackage[pdftex]{hyperref}

\begin{document}
\author{Itai Ashlagi, Maximilien Burq, Patrick Jaillet\thanks{Jaillet acknowledges the research support of the National Science Foundation grant OR-1029603 and the Office of Naval Research grants {N00014-12-1-0033} and {N00014-15-1-2083}.}, Vahideh Manshadi}
\title{On Matching and Thickness in Heterogeneous Dynamic Markets\thanks{First Draft: Feb 2015.}}

\maketitle

\begin{abstract}
We study dynamic matching in an infinite-horizon stochastic market. While all agents are potentially compatible with each other, some  are hard-to-match and others are easy-to-match. Agents prefer to be matched as soon as possible and matches are formed either bilaterally or indirectly through chains.
We adopt an asymptotic approach and compute tight bounds on the limit of waiting time of agents under myopic policies that differ in matching technology and prioritization. 

We find that the market composition is a key factor in the desired matching technology and prioritization level.
When hard-to-match agents arrive less frequently than easy-to-match ones (i) bilateral matching is almost as efficient as chains (waiting times scale similarly under both, though chains always outperform bilateral  matching by a constant factor), and (ii) assigning priorities to hard-to-match agents
improves their waiting times. When hard-to-match agents arrive more frequently, chains are much more efficient than bilateral matching and prioritization has no impact.


We further conduct comparative statics on arrival rates. Somewhat surprisingly, we find that in a heterogeneous market and  under bilateral matching, increasing arrival rate has a non-monotone effect on waiting times, due to the fact that, under some market compositions, there is an adverse  effect of competition. Our comparative statics shed light on the impact of merging markets and attracting altruistic agents (that initiate chains) or easy-to-match agents.

This work uncovers fundamental differences between heterogeneous and homogeneous dynamic markets, and potentially helps policy makers to generate insights on the operations of matching markets such as kidney exchange programs.

\end{abstract}


\section{Introduction}
\label{sec:intro}

This paper is concerned with the problem of matching in a  dynamic marketplace, where  heterogeneous  agents arrive over time to the market looking to exchange an indivisible  item for another compatible item.  A key feature of the market is its exogenous thickness, as determined by the types of agents and their arrival rates to the marketplace. For example, in kidney exchange some patient-donor pairs are very hard-to-match while others are very easy-to-match. In online labor markets, employers have different qualification requirements and workers have different skills.

Efficiency is determined by  the matching policy and the matching technology. The former determines which  exchanges to be implemented and when, and in particular which priorities, to assign to different types of agents.
The  latter determines the forms of matches that can take place; For example, while kidney exchanges were first conducted through bilateral exchanges (2-way cycles) (\cite{RothKidneyJET}), multi-hospital platforms are now facilitating many of their transplants through chains initiated by altruistic donors (\cite{anderson2015kidney}). In many matching markets, such as dating, only bilateral matches take place.

We are interested  in  the  behavior of  simple  {\it myopic policies}  under different matching technologies  and different thickness levels of the market.
Myopic policies form matches as soon as they become available, but may vary with respect to how  they prioritize agents in the events of ties.
Our framework will  allow to discuss  policy questions such as: What is the effect of prioritizing different types of agents? How does disproportional change in arrival of different types influence market efficiency? What is the impact of merging matching marketplaces with different thickness levels on different types?

Two comments are in place. First, restricting attention to myopic  policies is motivated by current practices in kidney exchange platforms in the United States.
\cite{frequency} uses simulations based on empirical data from multiple exchange programs to show  matching myopically is nearly harmless.
Moreover, a similar conclusion is also arrived in theoretical work \cite{AndersonDynamic}.
\footnote{See Subsection \ref{sec:review} for further details.} While they consider a stylized model with  homogeneous agents, their result can be generalized to our heterogeneous model.\footnote{This is not the focus of this paper, but for completeness, we show this in Appendix \ref{sec:anyAlg}.}
Second, the literature on dynamic matching in sparse environments has focused on
homogeneous agents (\cite{AndersonDynamic,Akbarpour}). The motivation for this paper stems from the  heterogeneity of agents in the marketplace.

For our purposes we propose a  simple  infinite-horizon model with two types of  agents, easy-to-match ($E$)  and hard-to-match ($H$).  Agents of each type $T$ arrive to the marketplace according  to an independent  Poisson process with rate $\lambda_T$. Each agent arrives with an indivisible item that she wishes to exchange. We assume a stochastic demand structure, where each agent of type $T$ finds the item of any other agent compatible independently with probability $p_T$.  A key feature of the  model is that $p_H$ is significantly smaller than $p_E$.  Agents are indifferent between compatible items but prefer to be matched as early as possible. Moreover, agents in our model depart the market after being matched. We therefore adopt the average waiting time of agents in steady-state as a measure for efficiency.\footnote{{More precisely, we focus on the average waiting time of $H$ agents, because the waiting time of $E$ agents is negligible compared to that of $H$ agents. For  a more detailed discussion, see Section \ref{model section}.}} 
While our model is highly stylized, it captures some important features observed in kidney exchange pools (see Section  \ref{sec:kidney}, where we provide a brief background that further motivates this study).

Two settings are considered, distinguished by how matchings are formed (feasible exchanges): bilateral (2-way cycles), and chains. Our main findings are the following. First, we find that market thickness plays a crucial role on the desired matching technology; when easy-to-match agents arrive more frequently to the market than hard-to-match ones, the average waiting time of $H$ agents  scales similarly under chains and bilateral matchings. But there is a sharp increase in the average waiting time of $H$ agents as soon as hard-to-match agents arrive more frequently,  highlighting the importance of chains  in marketplaces with a  majority of hard-to-match agents. Second, we find that, under bilateral matching, increasing arrival rates of hard-to-match agents may negatively affect hard-to-match agents by increasing their waiting times. Under chains, however, increasing arrival rates always shortens waiting times. Third, impact of prioritization in bilateral matching also depends on the market composition; when hard-to-match agents are the minority type, assigning them priority improves their waiting times. {\footnote{To be precise, theoretically we are only able to prove that prioritizing $H$ agents leads to shorter or equal waiting times (of $H$ agents), however, numerically we confirm
that such prioritization indeed leads to strictly shorter   waiting times. }}
However, when they are in majority, such prioritization has no significant impact.

Next we describe our results more formally under the bilateral and chain settings. In our analysis we compute the average waiting time of $H$ agents  under various myopic policies as $p_H \rightarrow 0$.

\noindent{\bf Bilateral matching.} Two  myopic policies are  considered for bilateral matching, differing in the type of agents they prioritize; While it appears natural to prioritize hard-to-match agents,  it is also interesting to consider the prioritization of easy-to-match agents as  these may have better outside options.\footnote{\label{note1} In reality, agents may leave the market without being matched due to various reasons such as finding outside options. Under a stochastic departure model, shorter waiting times correspond to fewer departures because both quantities are proportional with the market size. We leave the rigorous treatment of a model with departure as an open question. }
We find that regardless of how agents are prioritized,  when  $\lambda_H < \lambda_E$ waiting time  scales with $1/p_H$, and when $\lambda_H > \lambda_E$  waiting time scales with $1/p_H^2$. When easy-to-match agents arrive more frequently, prioritizing $H$ agents results in  shorter waiting times than prioritizing $E$ agents. However,  when $\lambda_H > \lambda_E$, the average waiting time in the limit is identical under both types of priorities.

We further provide comparative statics for the case in which $H$ agents are prioritized. Increasing $\lambda_E$ always decreases waiting times. However, the average waiting time is non-monotone when increasing $\lambda_H$; It has an increasing trend up to a certain threshold, which depends on  $\lambda_E$, and then it decreases (note that in a homogeneous model with only $H$ agents, Little's law implies that increasing $\lambda_H$ always decreases waiting times). These findings have two main implications: (i) thickening the market by increasing arrival rates of hard-to-match agents can result in longer waiting times depending on the existing  arrival rates, (ii) merging two marketplaces with different compositions, i.e. different ratios between the two arrival rates,  may not be beneficial for both.

\noindent{\bf Chain matching.} Under the chain setting,  we consider   policies termed \textit{{ChainMatch(d)}} for markets endowed with $d$ altruistic donors  who initiate chains that continue indefinitely. In a chain, each agent is matched by (receives an item from) some agent, and matches another. Whenever the last agent of a chain can match a new arriving  agent, the policy forms a new  {\it chain-segment}, which is a maximal sequence of matches resulting from a \textit{local search}, in which the next matched agent is selected randomly while breaking ties in favor of $H$ agents (so the policy does not always identify the longest possible chain-segment, which requires a global search and may be computationally hard).
We prove an upper-bound on the average waiting time that scales with $1/p_H$ for  all positive arrival rates. We also find that even in the regime $\lambda_E>\lambda_H$ where the waiting time scales similarly under both matching technologies,  chains result in lower waiting times than bilateral matching.

We  provide comparative statics  over the arrival rates of both types. We show (analytically for $p_E=1$ and numerically for $p_E<1$), that the average waiting time decreases when the arrival rate of either type increases. When $p_E=1$, we further find that the average waiting time is independent of the constant $d$. Similar patterns hold numerically when $p_E<1$.  Finally, we are able to compute the average chain-segment (which plays an important operational role for example in kidney exchange). An increase in  $\lambda_E$ or $d$, decreases the average length of a chain-segment. In contrast, increasing $\lambda_H$ has the opposite effect.

Next we provide brief intuition for some of the main findings, beginning with why the market composition and the desired matching technology are tightly connected.
Under the bilateral setting when easy-to-match agents arrive more frequently, almost all hard-to-match agents will be matched with easy-to-match ones resulting in a scaling of $1/p_H$; on the other hand, when hard-to-match agents arrive more frequently,  many of them will have to match with each other resulting in a scaling of $1/p_H^2$ which is the inverse of the probability that two $H$ agents can match each other.
In contrast, matching through chains does not require such ``coincidence of wants'' between pairs of $H$ agents even when $H$ agents are the majority. This results in a waiting time that scales with  $1/p_H$ regardless of the composition.
We further find that the heterogeneity in the marketplace may lead to non-trivial effects when increasing participation; The intuition for why, in the bilateral setting, $H$ agents may be harmed when attracting more $H$ agents to the market is that this leads to harsher  competition among $H$ agents for matching with $E$ agents (even though $H$ agents can potentially match with each other).\footnote{A similar effect happens in kidney exchange where  O-A patient-donor pairs that cannot match with each other  compete to match with scarce pairs with blood-type O donors. Note, however, that in our setting all agents can potentially match with each other; in  particular this effect extends to sets of pairs that are blood type compatible with each other, like O-O  pairs, some of which are much harder-to-match than others.} We elaborate and provide intuition  for other results throughout the paper.

Understanding the impact of market composition by providing comparative statics requires us to not only compute the scaling of asymptotic  behavior of average waiting time but also to characterize the exact limits. Such exact characterization in a heterogeneous model is particularly challenging as we need to analyze $2$-dimensional Markov chains.  For bilateral matching polices, we directly analyze the underlying $2$-dimensional spatially non-homogeneous random walks.
{One of the main challenges in our analysis is the need to jointly bound the distribution in both dimensions, because applying methods such as Lyapunov functions or analyzing  marginal probability distributions would not result in tight bounds.}
In doing so, we prove two auxiliary lemmas on concentration bounds for a general class of $2$-dimensional random walks that can be of interest for studying similar random walks that may arise in other
applications. For chain policies, we first couple the underlying Markov process with a $1$-dimensional process where no $E$ agent joins the market. Analysis of the resulting $1$-dimensional Markov chain presents new challenges as transitions between non-neighboring states happen due to the possibility of forming arbitrarily long chain-segments.
However, we show that the chain-segment formation process exhibits a memoryless property, which proves helpful in computing the waiting time limits.

\subsection{Related work}
\label{sec:review}

A close stream of related papers study dynamic matching in  models, in which agents' preferences are based on compatibility, i.e, agents are indifferent between whom they match with (\cite{Utku,AndersonDynamic,Akbarpour}).

The impact of the matching technology is addressed in  markets comprised of only easy-to-match agents (\cite{Utku}) (with multiple coarse types) or only hard-to-match ones (\cite{AndersonDynamic}). \cite{Utku} finds that short cycles are sufficient for efficiency.\footnote{The findings by \cite{Utku} thus provide a rationale for the static large market results (see, e.g. \cite{RothKidneyAER}).} \cite{AndersonDynamic} consider markets, in which all agents are ex ante symmetric and hard-to-match. They study  the waiting-time scaling  behavior of myopic policies that attempt to match each agent upon arrival in three settings of exchanges, $2$-ways, $2$ and $3$-ways, and chains, and find that moving  from 2-ways or 3-ways to chains significantly reduces  the average waiting time.\footnote{See  also \cite{ProcacciaSandholm} that  demonstrate the benefit of chains using simulations in dynamic kidney exchange pools.}
Our paper  bridges the gap by looking at a model with both hard- and easy-to-match agents and thus allowing  for different levels of thickness in the market.
\footnote{\cite{ding2015non} study a similar two-type model in a static setting and quantifies the effectiveness of matching through chains taking a novel random walk approach.}

The papers above also find that, by and large, myopic policies are near-optimal: \cite{Utku} analyzes a kidney exchange model with different types and deterministic compatibility structure across types  and finds that matching upon arrival is near optimal, even though some waiting with certain types to facilitate three-way exchanges adds some benefits.
\footnote{See also \cite{gurvich2014dynamic}, who study a similar compatibility-based inventory control model.}
\cite{AndersonDynamic} consider a homogeneous model  without departures (similar to our model with $\lambda_E = 0$) and finds that  there is little benefit from waiting before matching under both  matching technologies of short cycles and chains.\footnote{The waiting-time scales with the same factor with or without waiting before matching.}
\cite{Akbarpour}  consider  a homogeneous model with departures and finds that the optimality gap  of the policy that matches without waiting remains   constant as the match probability decreases. Moreover, using data-driven simulations, \cite{frequency} study the impact of match-run frequency, and show that among polices that match periodically (e.g., every week or every day), high matching frequencies perform best. \footnote{{Non-myopic policies have also been studied, for example} \cite{ProcacciaSandholm2} study forward-looking polices by casting the dynamic matching problem as a high-dimensional dynamic program, and develop a heuristic to overcome the curse of dimensionality.} This paper builds on these findings, and only analyzes myopic policies that search for a match upon arrival of a new agent.

We  elaborate  on the relation to  \cite{AndersonDynamic}, which is closest to our paper. Studying myopic policies under a homogeneous setting resulted valuable insights.  Some insights, however,  do not carry over to heterogeneous settings like kidney exchange (See Subsection \ref{sec:kidney}).
For instance, merging markets is often sought as a solution to improve efficiency. A homogeneous model  predicts that  increasing arrival rates (or merging markets) will always decrease waiting times.
In contrast, we find that  merging heterogeneous markets may not decrease waiting times for both markets. 
The homogeneous model by \cite{AndersonDynamic} predicts very  infrequent but very long chain-segments. Our model predicts shorter chain-segments, which fits better empirical evidence (chain-segments  typically  consist of only a few pairs). 
Further, we remark that some questions cannot be addressed in a homogeneous setting; for instance kidney exchange  programs attempt to attract easy-to-match pairs (\cite{AshlagiRothIR}); but the impact of such an increase cannot be investigated in a homogeneous model. As another example, exchange programs usually assign high priority to hard-to-match pairs; effect of such prioritization cannot be studied in a homogeneous model.
Overall it is natural and important to study richer models in order to address relevant policy questions.

Another stream  of related research considers models of  agents' preferences that do not depend only on compatibility. These papers find that policies that match without waiting are inefficient (\cite{baccara2015optimal,fershtmanre,doval2014theory,kadam2014multi}) since some waiting can improve the quality of matches.\footnote{See also related results in queueing models \cite{Leshno,bloch2014dynamic}.}



Our work is also  related to the problem of matching multi-class customers to multi-class servers studied in queueing literature (e.g., \cite{InfiniteBMatching,Weiss2}).
In our model, an agent can be thought as a pair of customer-server, and the compatibility between any two agents is probabilistic, thus we will not have a finite number of queues.

Finally  our work is related to the online  matching  literature that study  online matching in which the underlying graph is  bipartite and agents on one side of the graph are all present in the market and only agents on the other side arrive over time (\cite{kvv,aryanak_randominput,aryanak_stmatching,mos,STMatchingPatrick}).


%

\subsection{Organization}
In Section \ref{model section} we  introduce the model,   polices, and the underlying stochastic processes. In Subsection \ref{sec:kidney} we provide a brief background on kidney exchange further  motivating our framework and study.  In Section \ref{sec:results} we present the main theoretical results and Section \ref{sec:policy} complements the results with numerical experiments. Section \ref{sec:proofs} outlines  the main  proof ideas and techniques along with the details of Markov chains induced by each policy. Section \ref{sec:conclusion} concludes. 
For the sake of brevity, we only include proofs  of selected results in the main text. The detailed proofs of the rest of the statements are deferred to clearly marked appendices.


\section{Model}
\label{model section}

We study an infinite-horizon dynamic matching market, where  each  arriving agent is endowed  with a single item she wants to exchange for  another item she finds compatible. Agents are indifferent between compatible items and wish to exchange as early as possible, their cost of waiting being proportional to the waiting time. 

There are two types of agents, $H$ and $E$,  referred  by hard-to-match and easy-to-match, respectively. Beginning at  time $t=0$,  agents of type  $T\in\{H,E\}$  arrive to the market according to an independent Poisson process with rate $\lambda_T >0$. 

Any agent of type $H$ ($E$) finds the item of any other agent compatible independently with  probability $p_H$ ($p_E$). Our analysis  is asymptotic in $p_H \rightarrow 0$, while $p_E$ is a fixed constant.  So, on average, an $H$ agent finds significantly  fewer items compatible than an  $E$ agent.   
We  say that an agent $j$ is {\em matched} by agent $i$, if agent $j$ receives agent $i$'s item. An agent leaves the market only when she is {\it matched}, i.e., she receives a compatible item.

We study matching policies in two different settings, distinguished by how agents can exchange items. In the first setting two agents can exchange items bilaterally in a cyclic fashion. In the second setting agents exchange items through chains; at time $t=0$, there  are $d$ special agents called {\it altruistic agents} who are willing to give an item without getting anything in return (all other agents that will arrive to the market are regular agents who want to exchange their item for another item).\footnote{Having altruistic agents is an intrinsic property of the market in the sense that some markets do not have access to such agents.}
Each agent in a chain receives a compatible item from one agent and gives to the next.
At any given time, there are exactly $d \geq 1$  agents who are either altruistic or  received an item but have not given their item. The latter are called {\it bridge agents}.  We  sometimes refer to altruistic agents also by bridge agents. The transactions between two bridge agents in a given chain is called a {\it chain-segment}. We assume that matches in a chain-segment are conducted instantaneously.
A {\it policy} is a mapping from the history of
exchanges and the state of the marketplace to a set of feasible exchanges involving non-overlapping sets of agents.


We adopt the average waiting time in steady-state as the measure of the efficiency of a policy (the waiting of an agent is the difference between her departure time and her arrival time).  In our model, the average waiting time of one type of agents is equivalent to  the average number of agents of that type in the marketplace divided by the arrival rate of that type since these two quantities are proportional to each other by Little's law.

It is convenient to think about the state of the marketplace at any time in terms of a {\it compatibility graph}, which is  a directed graph with each agent represented by a node, and a directed edge from $i$ to $j$ means that agent $j$ finds agent $i$'s item compatible. Let $\mathcal{G}_t=(\mathcal{V}_t, \mathcal{E}_t)$ denote the (observed) compatibility graph at  time $t$. When a new agent arrives directed edges are formed in each direction independently and with probabilities corresponding to the agents' types, between the arriving agent and each agent in the marketplace. A bilateral exchange is a directed cycle of length two in the compatibility graph  and a chain-segment is a  directed path in this graph  starting from a bridge or altruistic agent.




We study the following myopic  policies, which attempt to match agents upon arrival.

\begin{definition}[\textit{BilateralMatch(T)} for $T\in \{H,E\}$]\emph{Upon arrival of a new agent, if a cycle of length $2$ can be formed with the newly arrived agent, it is removed. If more than one such cycle exists, priority is assigned to cycles with agents of type $T$. Further ties are broken uniformly at random.}
\end{definition}

%
%
%

\begin{definition}[\textit{ChainMatch(d)}]
\emph{
There are $d$ bridge or altruistic agents in the market at any given time. We describe first the policy for $d=1$.
Consider a new arriving agent $i_1$. If $i_1$ does not have an incoming edge from the bridge agent then no matches happen. Otherwise, a chain-segment begins with matching $i_1$ by the bridge agent and advances as follows. If there are edges from $i_1$ to other agents (not already in the chain-segment) one is selected, say $i_2$,  who is matched by $i_1$. Ties are broken randomly while favoring $H$ agents.  This  disjoint directed path  continues as long as possible and instantaneously. All agents in the chain-segment leave the market and the last agent becomes a bridge  agent.}

\emph{When there are  $d>1$ altruistic/bridge  agents, if there is at least one directed edge from one of them to the newly arrived agent, one of such edges is selected uniformly at random.  As the process moves forward, each altruistic agent eventually gives her item to an arriving agent and starts a chain.} 
\end{definition}

Under the \textit{ChainMatch(d)} policy,  upon arrival of a new agent a maximal chain-segment  (path) is identified through  local search originating from a bridge agent.
\footnote{{Our local search chain-segment formation process bears similarity to  Phase 1 of the  two-phase clearing procedure of \cite{ding2015non}.
}}
Note that the chain-segment has a positive length if and only if at least one bridge/altruistic agent has a directed edge  to the new agent.

For brevity we often refer to \textit{BilateralMatch(E)},  \textit{BilateralMatch(H)}, and \textit{ChainMatch(d)}, by  $\BH$, $\BE$, and $\Cd$, respectively.
All the policies above are Markov policies, and thus define a continuous-time Markov chain (CTMC). 
The following observation will allow us to ignore the edges within the market when analyzing the underlying stochastic processes.

\begin{observation}
\label{obs:onlynodes}
For each policy $\BH$, $\BE$, and $\Cd$, in order to analyze the average waiting time, it is sufficient to keep track  of only the number of agents of each type in the market.
\end{observation}
The observation is immediate for the bilateral policies $\BH$ and $\BE$; due to their myopic behavior there are no $2$-length cycles in the market except with a new arriving agent, implying that the corresponding Markov chains can be fully specified using only the set of vertices. For the $\Cd$ policy, the observation is more subtle. Note that under this policy there is no outgoing edge from a bridge agent to any  waiting agent, again due to the myopic behavior of the policy. The first time we  examine whether there is an edge from $i$ to $j$, we  effectively flip a bias coin with probability $p_H$ ($p_E$) if the  agent  $j$ is of type $H$ ($E$). Importantly, we examine at most once whether a directed  edge from $i$ to $j$ exists by the definition of the policy since $i$ either leaves the market or becomes a bridge agent, in which case it will never match to $j$. Since  both the edge formation and the matching policies do not depend on agents' identities (rather only on their types) we can merely keep track of the number of agents of each type.

In the remainder of the paper, for any policy $\cP$, we focus on the state space  $\left\{[H_t^{\cP},  E_t^{\cP}]; t \geq 0\right\}$, which captures the number  of hard- and easy-to-match agents at any time $t$,  and we  denote the corresponding transition rate matrix by $Q^{\cP}$. 

Given the {\em self-regulating} dynamic undergoing each matching process, one would expect that all three (irreducible) CTMC's reach steady-state. A rigorous statement and proof is provided in Appendix \ref{app:existence_proofs}. Hereafter, we are  concerned only with steady-state analysis;
For policy $\cP$, we denote its steady-state distribution by $\pi^{\cP}$. The random vector  $[H^{\cP}, E^{\cP}]$ is the random number of $H$ and $E$ agents in steady-state, i.e., the vector is distributed according to  distribution  $\pi^{\cP}$. Finally we define $w_H^{\cP}$ ($w_E^{\cP}$) to be the average waiting time of type $H$ ($E$) agents under policy $\cP$.
Little's law implies that

\begin{align}
\label{eq:littles:law}
w_H^{\cP} = \frac{\E[H^{\cP}]}{\lambda_H} \text{, and}  \quad w_E^{\cP} = \frac{\E[E^{\cP}]}{\lambda_E}.
\end{align}

Since in our model $p_H \rightarrow 0$ while $p_E$ is kept constant, and all policies are myopic, one would expect that $w_E^{\cP}$ is negligible compared to  $w_H^{\cP}$.
We verify this claim below using numerical simulations and analytical proofs (see Figure \ref{fig:all_policies} and Lemmas \ref{lem:upper_E} and \ref{lem:upper_E_prioE}).
Therefore we  focus on analyzing the average waiting time of $H$ agents under different policies.


In Section~\ref{sec:results}, we derive asymptotic results ($p_H \rightarrow 0$) for $w_H^{\cP}$ for different set of parameters $\lambda_H$, $\lambda_E$ and $p_E$. We note that  $w_H^{\cP}$ is indeed a function of four parameters, and a more precise notation would be $w_H^{\cP}(\lambda_H, \lambda_E, p_H, p_E)$, but we drop these parameters for  the sake of brevity.

\subsection{Motivating application: kidney exchange}
\label{sec:kidney}

\noindent{\bf Background.} 
There is a large shortage of kidneys for transplants\footnote{As for 2017, the average  waiting time is  between 3-5 years in the U.S.} and many live donors are incompatible with their intended recipients. Kidney exchange allows such patient-donor pairs to swap donors so that each patient can receive a kidney from a compatible donor. There have been efforts to create large platforms to increase  opportunities for kidney exchanges  \cite{RothKidneyQJE,nikzad2017financing}.

Exchanges are conducted through cycles or chains.\footnote{See \cite{sonmez2017market} for a detailed description of kidney exchange.} Typically  pairs do not give a kidney prior to  receiving one. This creates logistical barriers requiring  cycles to be limited to $2$ or $3$ pairs. Chains, however, can be organized non-simultaneously, thus can  be longer \cite{RothAJT,ReesNEJM}.
For a transplant to take place, the patient needs to be both blood-type and tissue-type compatible with a donor. 
The common measure of patient sensitivity is the  Panel Reactive Antibody (PRA), which captures the likelihood the patient is tissue-type incompatible with a donor chosen at random in the population, based on her antibodies.

Numerous kidney exchange platforms operate in the U.S., varying in size, composition, and  policies. Some   are national platforms (with many participating hospitals) like the Alliance for Paired Donation (APD) and the National Kidney Registry (NKR). Others  are regional  or even single center programs like Methodist Hospital in San Antonio (MSA).

\noindent{\bf  Data.} Next we provide some  figures about the pool composition. Kidney exchange platforms are selected to have a large fraction of  highly sensitized patients \cite{AshalgiGamarnikRoth}.
Figure \ref{fig:PRA}(left) plots the PRA distributions of   patients enrolled at the  NKR, APD and MSA. Most patients are either highly sensitized (PRA above 95) or low sensitized (below 5 PRA). 
Note that blood-type compatibility is not incorporated in this aggregate PRA distribution.
Figure \ref{fig:PRA}(right) provides the same distributions for patients belonging to blood-type compatible pairs (e.g., O-O patient-donor pairs), who can match with each other if they are tissue-type compatible. These distributions can be roughly viewed as bimodal; note that  among blood-type compatible pairs  there are more highly sensitized patients than low sensitized ones.

\begin{figure}[ht!]
  \centering
\includegraphics[scale=0.35]{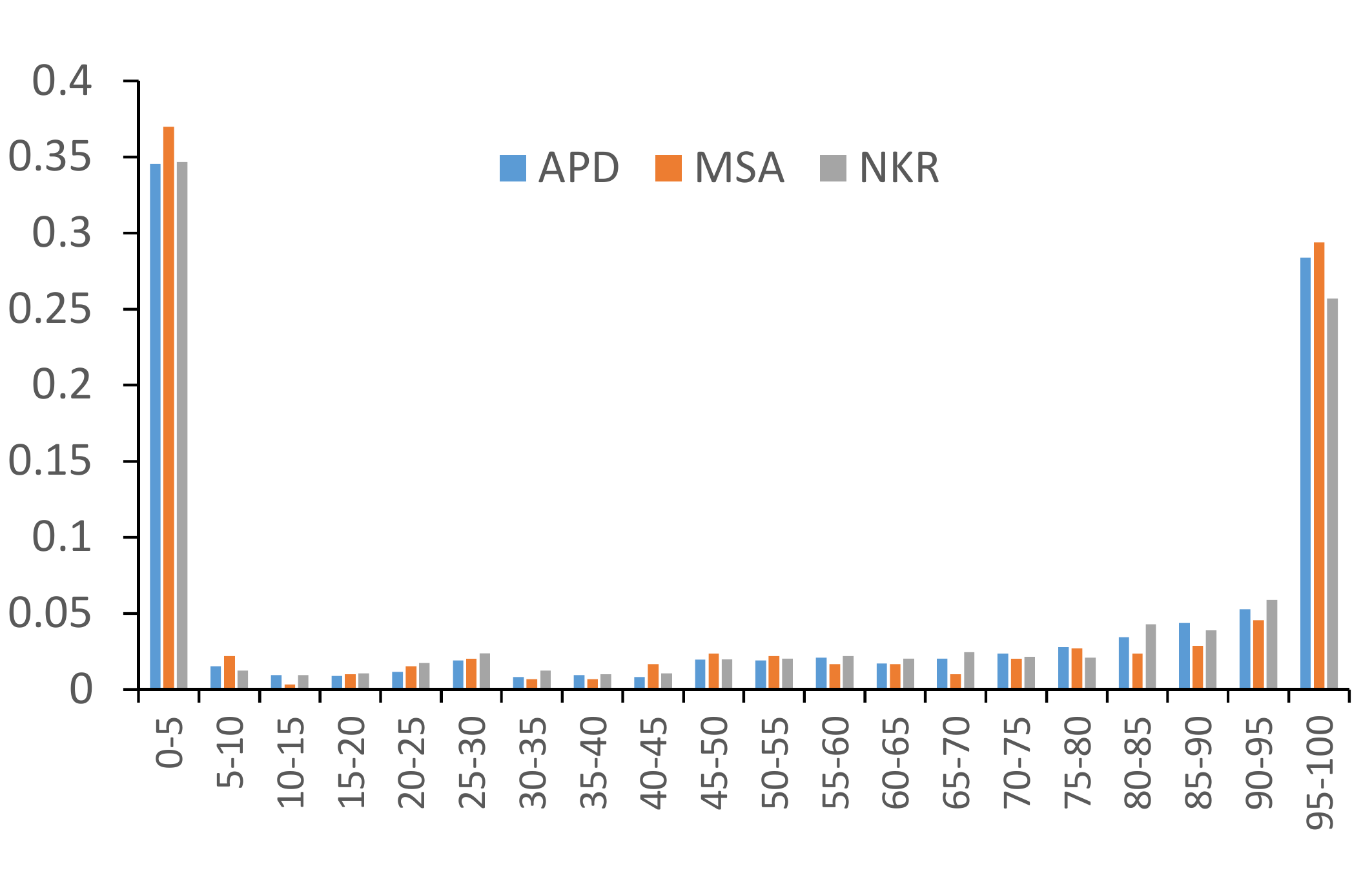}
\hspace{0.2cm}
\includegraphics[scale=0.35]{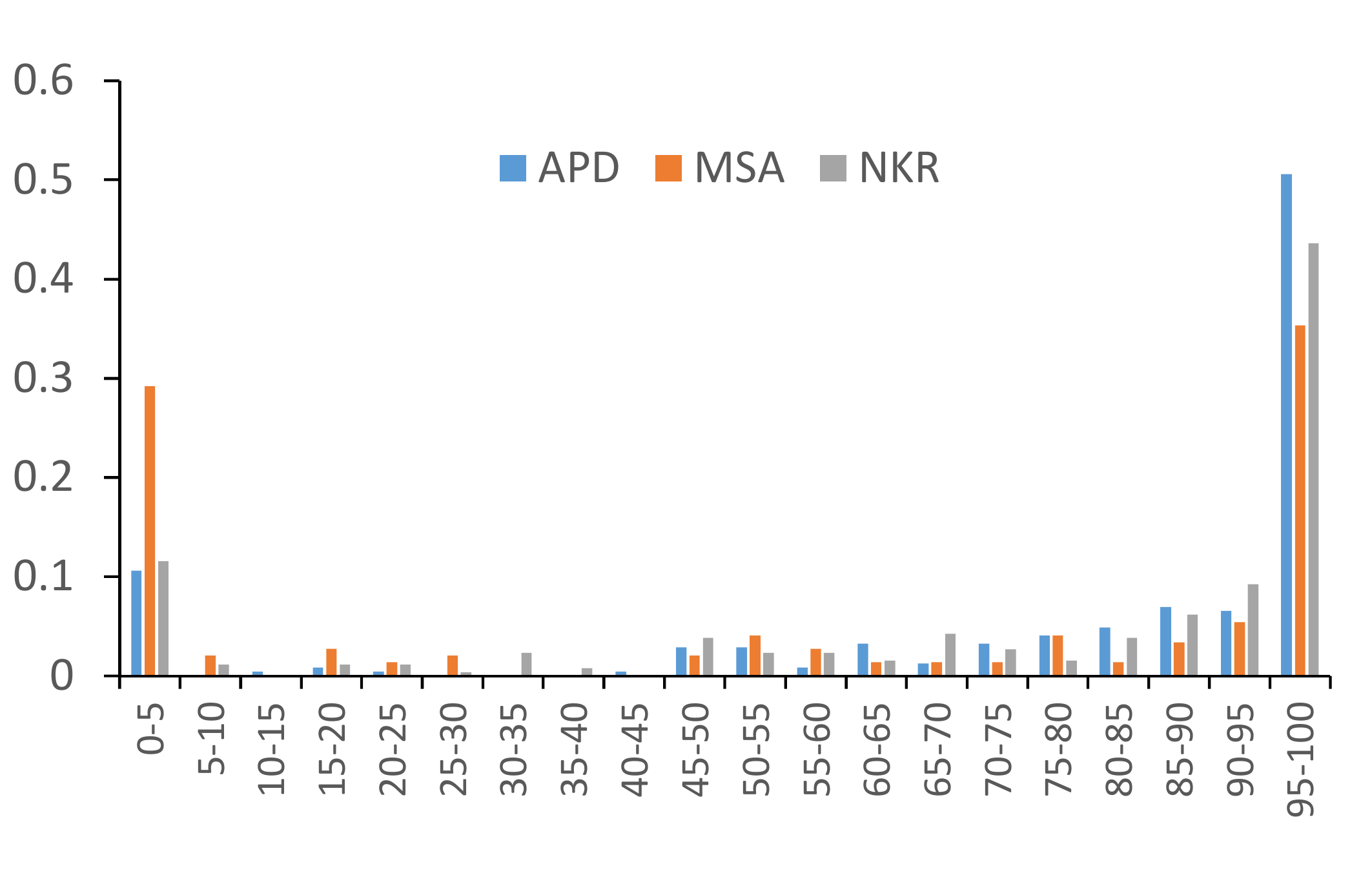}
\caption{PRA distributions of patients enrolled at NKR (1/2012-12/2014), APD  (1/2007-8/2016), and MSA (7/2013-2/2017). Left: all patients. Right: patients belonging to blood-type compatible pairs.}
\label{fig:PRA}
\end{figure}

Percentage of high PRA patients also varies across programs outside the U.S.
In Australia, $42\%$ of registered candidates have a PRA greater than $90$ (\cite{ferrari2012high}) and in the UK, $46\%$ of patients have a PRA greater than $85$ (\cite{johnson2008early}), while in Canada, only $36\%$ of pairs have a PRA of $80$ or more (\cite{MalikCole}).  In the Netherlands, \cite{glorie2014kidney} estimate that $30\%$ of patients have a PRA above $97$.

Similar to PRA distribution, the pool compositions also vary with respect to blood-type distributions of patient-donor pairs.
\cite{frequency} report that O-O pairs make $26\%$ at the MSA pool but only  $20.81\%$  of the APD pool; the percentage of pairs that contain an O donor in the APD and MSA pools are $34.4\%$ and $50.3\%$, respectively.

These platforms also differ in  size; during the period of the data,   MSA and APD had  an enrollment rate of roughly $180$ pairs per year, while the NKR had an enrollment rate of about $360$ pairs per year. Access to altruistic donors also varies, with roughly $1$, $8$ and $50$ altruistic donors per year at the MSA, APD and NKR, respectively.

\noindent{\bf Matching.}  While more than $80\%$ of the transplants at the  NKR and the APD have been conducted through chains (\cite{anderson2015kidney,NKRreport}), some platforms (such as MSA, Belgium, Czech Republic) match their pairs mostly through cycles due to short access to altruistic donors. In  countries like France, Poland and Portugal, chains are infeasible since altruistic donations are not permitted (\cite{eurokpd}).

Exchange platforms in the U.S.  adopt typically  myopic-like matching policies that periodically search for matches.\footnote{Based on personal communication with numerous platforms.}\footnote{Generally speaking,  a myopic policy is one that upon matching does not explicitly account for possible future matches.} The  APD, MSA, and NKR search for exchanges on a daily basis  and UNOS searches for exchanges bi-weekly.\footnote{There is some concern that this behavior is inefficient (and arguably a result of competition). However,  numerical simulations  by \cite{frequency} suggest that in steady-state there is essentially  no harm from frequent matching (though having multiple small platforms does harm efficiency). Moreover, MSA is not facing any  competition.} However, some countries,  such as Canada, United Kingdom, the Netherlands and Australia, search  for exchanges every $3$ or $4$ months (\cite{ferrari2014kidney}).

Matching policies at most platforms assign high weights to highly sensitized patients (easy-to-match patients match quickly (\cite{frequency,NKRreport})). We note, however,   that MSA and  NKR assign high priority to compatible pairs, which are very easy-to-match.\footnote{Such pairs could choose to go through a direct transplant if they are not matched quickly.} Platforms typically have  multiple desiderata. However,  implicit first order related goals are to reduce waiting times and facilitate many transplants (\cite{NKRreport}).

\noindent{\bf Policy.}  Various challenges  arise due to  variation across kidney exchange pools with respect to their compositions and even operational issues: What priorities to assign to different types of patients? What is the impact of attracting more easy-to-match pairs  and even compatible pairs?\footnote{See for example \cite{marketfailure} and  \cite{sonmez2017incentivized} for incentive schemes towards thickening the pool with such pairs.} How important is it to   incorporate chains and attract altruistic donors?

There are also several initiatives  to merge kidney exchange platforms in order to  increase efficiency and  matching opportunities for highly sensitized patients (for example, see \cite{bohmig2017czech} for  merging the  Austrian and the Czech Republic programs, \cite{israelcyprus} for Israel and Cyprus; further, \cite{nikzad2017financing} look at augmenting national programs through global kidney exchange.). It is natural to study what is the impact of merging programs on different types of patients. 

This paper does not intend to model  the details in  kidney exchange. However, our stylized model  does capture some important features in kidney exchange and will hopefully generate some useful insights.


\section{Main results}
\label{sec:results}

We analyze the average waiting time under the myopic policies defined in Section \ref{model section}.
For bilateral matching polices, we identify a stark threshold in the scaling of  waiting time when moving from the regime where a  majority of arrivals are hard-to-match  agents to the regime where the majority of arrivals are easy-to-match. Such a contrast does not exist when agents are matched through chains.  We further study the impact of arrival rates of the two types on the market performance under the three polices.

\subsection{Bilateral matching}

This section considers the setting, in which agents match only through bilateral exchanges, i.e. through 2-way cycles.

\begin{restatable}{theorem}{bilatprioH}
\label{th:bilat_prioH}
Under the \textit{BilateralMatch(H)} policy and in steady-state, the average waiting time $w_H^\BH$ satisfies the following.
\begin{itemize}
\item[-] If $\lambda_H < \lambda_E$, then $\lim_{p_H \rightarrow 0} p_H  w_H^\BH = \frac{\ln\left(\frac{\lambda_E}{\lambda_E - \lambda_H}\right)}{p_E  \lambda_H }$.
\item[-] If  $\lambda_H > \lambda_E$, then $\lim_{p_H \rightarrow 0} p_H^2 w_H^\BH = \frac{\ln \left( \frac{2 \lambda_H}{\lambda_H + \lambda_E} \right)}{\lambda_H }$.
\end{itemize}
\end{restatable}

Theorem \ref{th:bilat_prioH} provides  not only the scaling laws on $w_H^\BH$ but also the associated constants. The following  corollaries provide comparative statics with respect to $\lambda_H$.
\begin{corollary}
\label{cor:monotone}
Consider the \textit{BilateralMatch(H)} policy and fix  $\lambda_E$. The limiting average waiting time $w_H^\BH$  increases with $\lambda_H$ in the interval  $\lambda_H < \lambda_E$.
\end{corollary}

\begin{corollary}
\label{cor:nonmonotone}
Consider the \textit{BilateralMatch(H)} policy and fix  $\lambda_E$. The limiting average waiting time $w_H^\BH$   increases with $\lambda_H$ in the interval $\lambda_E < \lambda_H < x^* \lambda_E $, and decreases in the interval $\lambda_H > x^* \lambda_E$, where $x^* \approx 2.18$ is the unique solution of
\begin{align}
\label{eq:fixPoint}
(x+1) \ln(2-2/(x+1)) = 1.
\end{align}
\end{corollary}

The above theorem and corollaries provide  several  messages on the impact of thickness on the performance of bilateral matching.
First, the main factor in the  asymptotic behavior of $w_H^\BH$ is which type of agents has a larger arrival rate. Some intuition for the scaling factors is the following.
Agents'  average waiting time is inversely proportional to the probability of a bilateral match to occur. Under a myopic bilateral policy, no existing pair of agents in the market can match with each other. For an arriving $H$ agent, the probability of forming a bilateral match with an existing $E$ agent is $p_E p_H$,  and with an existing $H$ agent is $p_H^2$. When $\lambda_H < \lambda_E$,  almost all $H$ agents are matched with $E$ agents resulting in an average waiting time that scales with $1/p_E p_H$. When $H$ agents arrive more frequently than $E$ agents, there are simply not enough $E$ agents to match with $H$. So a  non-negligible fraction of  $H$ agents  match with each other and thus the scaling of the average waiting time increases to  $1/p_H^2$.

Second, the arrival rates affect the average waiting times directly and not necessarily monotonically. Increasing the arrival rate of $E$ agents always decreases the average waiting time.
But this is not the case with $H$ agents.
When $\lambda_H < \lambda_E$, the average waiting time of $H$ agents increases with $\lambda_H$. So even though  the market {\it thickens} and all $H$ agents match with $E$ agents, increasing  $\lambda_H$ creates more competition among $H$ agents.
When $\lambda_H > \lambda_E$, there is a non-monotone behavior of the waiting time when increasing $\lambda_H$.
Here too, this increase escalates competition among $H$ to match with $E$ agents and is the dominant effect as long as $\lambda_H$ is not ``much'' larger than $\lambda_E$.
After a certain threshold,  the positive effect from having more  $H$ agents (increasing the possibility of forming bilateral matches between $H$ agents), dominates the negative impact that results from the competition to match with $E$ agents.

The key insight from the above discussion is that in a heterogeneous market, increasing the arrival rate does not always result in improving the waiting time due to the adverse effect of competition for certain market compositions. This cannot be captured in a homogeneous model with only hard-to-match agents (the model studied in \cite{AndersonDynamic}). 

Finally, we comment on the impact of $p_E$ on the waiting time.  When $\lambda_H < \lambda_E$,  $w_H^\BH$ is decreasing in $p_E$. On the other hand, when $\lambda_H > \lambda_E$,  $w_H^\BH$ is independent of $p_E$. The intuition is that in the former, all $H$ agents match with $E$ agents and in the latter the dominant factor in the average waiting time is due to $2$-ways between $H$ agents, which is independent of  $p_E$. 

The proof of Theorem \ref{th:bilat_prioH} amounts to analyzing the underlying $2$-dimensional continuous-time spatially non-homogeneous random walk. The  description of the random walk is presented in Subsection \ref{subsec:BH} (Figure \ref{fig:2d:walk}), along with a heuristic that helps us guess the right constants, and build intuition on the behavior of the  random walk. The main idea behind the proof is establishing concentration results for a $2$-dimensional CTMC where the steady-state distribution decays geometrically when moving away from the expectation. These concentration results allow us to establish matching lower and upper bounds on  $w_H^\BH$ (the proof is outlined in Subsection \ref{subsec:BH} with details deferred to Appendix \ref{app:bilat_prioH}).
We note that one of the main challenges in our analysis is the need to jointly bound the distribution in both dimensions, because  analyzing  marginal probability distributions  would not result in tight bounds.
As a byproduct of our analysis, in Subsection \ref{subsec:lemmas}, we state two auxiliary lemmas on concentration bounds for a general class of $2$-dimensional random walks.
The  corollaries follow from basic analysis of the corresponding constants (as a function of $\lambda_H$). Both corollaries are proved in  Appendix \ref{subsec:corrs}.

\begin{restatable}{theorem}{bilatprioE}
\label{th:bilat_prioE}
Under the \textit{BilateralMatch(E)} policy and in steady-state, the average waiting time $w_H^\BE$ satisfies the following.
\begin{itemize}
\item[-] If $\lambda_H < \lambda_E$, then $ \frac{\ln\left(\frac{\lambda_E}{\lambda_E - \lambda_H}\right)}{p_E \lambda_H} \leq \lim_{p_H \rightarrow 0} p_H w_H^\BE   \leq \frac{\ln\left(\frac{2\lambda_E}{\lambda_E - \lambda_H}\right)}{p_E \lambda_H}$.
\item[-] If $\lambda_H > \lambda_E$, then $\lim_{p_H \rightarrow 0} p_H^2 w_H^\BE  = \frac{\ln \left( \frac{2 \lambda_H}{\lambda_H + \lambda_E} \right)}{\lambda_H}$.
\end{itemize}
\end{restatable}


Comparing results of Theorems \ref{th:bilat_prioH} and \ref{th:bilat_prioE}, we observe that when $\lambda_H < \lambda_E$, the average waiting time of $H$ agents is  larger {or equal} when  prioritizing $E$ agents rather then $H$ agents (numerical simulations presented in  Subsection \ref{subsec:sim_priorities} suggest that prioritizing $E$ agents results in a strictly larger average waiting time). Nevertheless, the scaling remains the same.
However, when $\lambda_H > \lambda_E$  prioritizing  $E$ agents does not impact the waiting time of $H$ agents. The intuition is as follows. When  $\lambda_H > \lambda_E$, the number of $H$ agents waiting in the market scales as $1/p_H^2$, suggesting that the chance that an $E$ agent does not match immediately upon arrival vanishes. Therefore assigning priority to $E$ agents is redundant.


The proof of Theorem \ref{th:bilat_prioE} also requires  analysis of  the underlying $2$-dimensional continuous-time spatially non-homogeneous random walk, and, in most parts, follows a similar structure to the proof of Theorem \ref{th:bilat_prioH}. A detailed description of the random walk is presented in Subsection \ref{subsec:BE}. 
The proof of the  upper and lower bounds is presented in Appendix \ref{sec:bilatE:proof}, where establishing the upper bound requires new ideas beyond the concentration results: we couple the Markov process underlying policy $\BE$ with another process in which an $E$ agent that cannot form a match upon arrival turns into an $H$ agent.\footnote{In Subsection \ref{subsec:BE} we  provide a rough intuition on why we cannot close the gap between our upper and lower bounds on $w_H^\BE$ for the regime $\lambda_H < \lambda_E$.}
In Subsection \ref{subsec:BE}, we also provide a heuristic argument that leads us to guess {that the exact limit is $\frac{\ln \left(\frac{\lambda_E + \lambda_H}{\lambda_E - \lambda_H}\right)}{p_E p_H}$} (See Figure \ref{wtimeHeuristicBE} in Subsection \ref{sec:app:tightness}).

\subsection{Chain matching}
\label{subsec:chains}

In this section we analyze  the \textit{ChainMatch(d)}  policy, under which agents match myopically  through chains.

\subsubsection{Waiting time behavior}
\begin{restatable}{theorem}{chainsTh}
\label{th:chain}
Let $d \geq 1$ be a constant (independent of $p_H$). Under the \textit{{ChainMatch(d)}} policy and in steady-state, the average waiting time $w_H^\Cd$ satisfies
$$\lim_{p_H \rightarrow 0} p_H w_H^\Cd   \leq \frac{\ln \left(\frac{\lambda_H }{\lambda_E (1 - (1 - p_E)^d)} + 1\right)}{\lambda_H}. $$
\end{restatable}

\noindent A stronger result is obtained for the special case, in which  $p_E = 1$: 
\begin{proposition}
\label{prop:chain:p1}
Let $p_E = 1$ and $d\geq 1$ be a constant (independent of $p_H$).  Then
$$\lim_{p_H \rightarrow 0} p_H w_H^\Cd  = \frac{\ln \left(\frac{\lambda_H}{\lambda_E} + 1 \right)}{\lambda_H}.$$
Consequently, $\lim_{p_H \rightarrow 0} p_H w_H^\Cd $ decreases with $\lambda_E$ and  $\lambda_H$.
\end{proposition}

First we discuss the intuition behind Proposition \ref{prop:chain:p1}, which states that when $p_E = 1$, any constant number of altruistic agents  will result in the same behavior of $w_H^\Cd$.
The positive impact of having $d>1$ altruistic agents stems from the increase in probability of starting a new chain-segment. When an  $H$ agent arrives, the probability that she finds one of the bridge agents acceptable is $1 - (1-p_H)^d$ which  vanishes as $p_H \rightarrow 0$.
When an  $E$ agent arrives she will always be matched by one of the bridge agents and proceed to advance the chain-segment, and thus there is no advantage in having more than one bridge agent.

While we are not able to pin down the exact behavior when $p_E < 1$,   some intuition  suggests that the behaviour for $p_E < 1$ is similar to the case in which $p_E=1$: even though having less  bridge agents decreases the likelihood of starting a new chain-segment upon arrival of an $E$ agent, this by itself does not result in longer waiting times for $H$ agents. Suppose an arriving $E$ agent $i$ cannot be matched by one of the bridge agents, and therefore joins the market. We argue that the presence of $i$ in the  market helps matching more $H$ agents in the (near) future chain-segments. Consider the first time a chain-segment is being formed after $i$ joins the market. Because $H$ agents have priority, the chain-segment tries to advance through $H$  agents until it gets ``stuck'' (i.e, cannot find an $H$ agent to add). At this point, with a constant probability agent $i$ can match the last  agent in the chain-segment and therefore can progress the chain-segment through more $H$ agents.
In Subsection \ref{subsec:sim:ndd_E_increase}, we study numerically the impact of arrival rates and number of altruistic agents on the waiting time for the case $p_E<1$. Our numerical results are qualitatively in agreement with the predictions of Proposition \ref{prop:chain:p1}.


When $p_E<1$,
in a heuristic argument (in Appendix \ref{app:chains_heuristic}), we analyze a related  $3$-dimensional random walk by artificially assuming that chain-segments advance according to an independent  Poisson process with a very high  rate $\mu$  (recall that under $\Cd$ policy, chain-segments are formed and executed instantaneously upon arrivals). The heuristic provides an estimated waiting time that scales as $\ln \left( \frac{\lambda_H + \lambda_E}{\lambda_H (1 - (1 - p_H)^d) + \lambda_E} \right)/p_H$. In the limit when $p_H$ approaches zero, the constant converges to $\ln \left( \frac{\lambda_H + \lambda_E}{\lambda_E} \right)$ which is consistent with Proposition \ref{prop:chain:p1}. 
Numerical simulations that are aligned with the result of the heuristic argument are presented in Subsection \ref{sec:app:tightness} (see Figure \ref{wtimeHeuristicChains})

Finally we comment on the chain-segment formation process; \textit{{ChainMatch(d)}} policy forms chain-segments employing a local search process and indeed our analysis relies on such chain-segment formation process. This begs the question of how much the waiting time improves if we employed a global search (that searches for the longest possible chain-segment). A precise comparison is beyond the scope of our work, however, we make the following remarks: (1) In Figure \ref{fig:max_vs_local_chains} of Subsection \ref{subsec:sim:ndd_E_increase}, we numerically study this questions, and we
see that advancing chains  locally  results in a small loss in comparison to policies that search globally for the longest possible chain-segment.
(2) The lower-bound on the waiting time of any anonymous Markovian policy (See \cite{AndersonDynamic} and Proposition \ref{prop:lowBound} in Appendix \ref{sec:anyAlg}) implies that the scaling of $H$-agent waiting time  cannot be smaller than $1/p_H$ (unless the policy makes $E$ agents wait for a very long time, i.e., proportional to $1/p_H$); Theorem \ref{th:chain} shows that the local-search method already achieves such a scaling.


Under the \textit{{ChainMatch(d)}} policy, the length of a  chain-segment trigged by a newly arrived agent is unrestricted. As a result the underlying CTMC is significantly more complicated to analyze than those  that arise from bilateral policies and we need other techniques to prove Theorem \ref{th:chain}.
In order to  bound $w_H^\Cd$,  we  couple the underlying Markov chain with a $1$-dimensional chain, in which  $E$ agents that are not matched upon arrival leave the market immediately (Lemma \ref{cl:chains_coupling}).
A key property used in the analysis of the coupled $1$-dimensional chain is that chain-segment formation exhibits a memoryless property.\footnote{This is different from the Markov property of the overall CTMC under $\Cd$.} This is  due to the local search process used to advance a chain-segment, which randomly selects the next agent among all possible agents (favoring $H$  agents). 
The proof is presented in \ref{sec:proof_chains}. Finally, we note that for the special case $p_E = 1$, the original CTMC is a $1$-dimensional chain for which we can prove matching upper and lower bounds on the limit of  $w_H^\Cd$. 

Theorems \ref{th:bilat_prioH} and \ref{th:chain} together highlight  the importance of having altruistic agents that can initiate chains.
In the regime $\lambda_H > \lambda_E$ comparing $w_H^\BH$ and $w_H^\Cd$ is straightforward as the former scales as $1/p_H^2$
but the latter only scales as $1/p_H$.  The following corollary (proven in appendix \ref{app:chain_coupling}) states that in the regime $\lambda_H < \lambda_E$
where both $w_H^\BH$ and $w_H^\Cd$ scale as $1/p_H$, \textit{{ChainMatch(d)}} performs better:

\begin{restatable}{corollary}{corBilatChains}
\label{cor:comp}
For any $\lambda_H$, $\lambda_E$, $p_E$, and $d$, if $\lambda_H < \lambda_E$  then $\lim_{p_H \rightarrow 0} p_H w_H^\Cd < \lim_{p_H \rightarrow 0} p_H w_H^\BH$.
\end{restatable}

In Subsection \ref{subsec:altruistic}  we further compare    \textit{BilateralMatch(H)} to  \textit{{ChainMatch(d)}} in order to understand the importance of attracting easy-to-match agents in markets that have limited access to altruistic agents.



\subsubsection{Chain-segment length}
\label{sec:chain-segment}

We analyze here the  expected length of chain-segments formed under the \textit{{ChainMatch(d)}} policy.

While we focus on the  average waiting time to measure
efficiency,  length of chain-segments also play a significant role on the operational efficiency of the market. In kidney exchange for example, executing a chain-segment takes  time and bears the risk of match failures.\footnote{In this stylized model, we abstract away from both of these effects.} These practical considerations motivate extending the analysis to the limiting behavior of chain-segments.

First we  define the chain-segment length. Let $[H_k^{\Cd},  E_k^{\Cd}]$ denote the (discrete-time) Markov chain embedded in the CTMC $[H_t^{\Cd},  E_t^{\Cd}]$ resulting from observing the system at arrival epochs.\footnote{Note that every time an agent arrives, the Markov chain advances in discrete time from $k$ to $k+1$.} Define: 
$$L_k = H_k + E_{k} - H_{k+1} - E_{k + 1} + 1,$$
and let $L$  be its corresponding random variable in steady-state;  if the arriving agent cannot be matched by the bridge agent, she will join the market, and therefore $L_k = 0$; otherwise, a chain-segment of length $L_k \geq 1$ will be formed. The following proposition characterizes the chain-segment length in the limit:
\begin{restatable}{proposition}{chainLengths}
\label{chL}
Under the \textit{{ChainMatch(d)}} policy and in steady-state, $$\lim_{p_H \rightarrow 0} \E[L \mid L \geq 1] = \frac{\lambda_H + \lambda_E(1-p_E)^d }{\lambda_E (1 - (1 - p_E)^d) } + 1.$$
\end{restatable}

\noindent The proof is presented in Appendix \ref{app:chain_coupling}.
We note that the expected chain length is decreasing in both $\lambda_E$ and $d$, but  increasing in $\lambda_H$; intuitively  with more $E$ agents or more bridge agents, chain-segments will be formed at a higher rate and thus be shorter (for a fixed $\lambda_H$). However, increasing $\lambda_H$ does not significantly impact the frequency of chain-segment formation, but given  that more $H$ agents join the market within two consecutive chain-segments, the length of the chain-segment grows.

\section{Numerical studies}
\label{sec:policy}

In this section, we present a set of numerical simulations that complement the theoretical results of the previous section. In Subsection \ref{subsec:merge} we look at how merging markets with different compositions affect each market. Subsection \ref{subsec:sim_priorities}  explores the  impact of giving priorities when using the bilateral matching policy. Subsection \ref{subsec:sim:ndd_E_increase} presents comparative statics for chain matching when $p_E <1$, and subsection \ref{subsec:altruistic}  highlights the advantage of having chains. Finally, Subsection \ref{sec:app:tightness} compares our theoretical bounds (for cases for which we do not have matching upper and lower bounds) to heuristics guesses  and  simulations.

All simulations in this section are conducted by first computing the average number of agents in the market; then applying Little's law \eqref{eq:littles:law}. In order to compute the number of  agents, we simulate the discrete-time Markov chain embedded in the corresponding CTMC resulting from observing the system at arrival epochs.
We denote $T$ the number of  arrivals (not counting the $d$ initial altruistic agents in the case of $\Cd$).
In order to remove the transient behavior, the numbers reported correspond to the time average over the second half of the simulation.

\subsection{Merging markets}
\label{subsec:merge}

We consider here the effects from merging two markets, with arrival rates $(\lambda_{H,1},\lambda_{E,1})$ and $(\lambda_{H,2},\lambda_{E,2})$ under bilateral exchanges using the {\it BilateralMatch(H)} policy.
This expands Theorem \ref{th:bilat_prioH}, which provides comparative statics in the limit when $p_H$ tends to zero.

We consider two numerical examples to illustrate these effects. In both examples the arrival rates to the first market are  kept fixed while the arrivals rates to the second  market  vary. For any pair of arrivals  we compare the waiting time $w_{H,1}$ of $H$ agents in the first market  with the average waiting time $w_{H, 1-2}$ in the  merged market. The results are plotted in Figure \ref{fig:merging}.
Consistent with our prediction, merging can  result in one of the markets being worse off. Note that this can happen even if the majority type is the same for both markets (e.g., when $\lambda_{H,1} > \lambda_{E,1}$ and $\lambda_{H,2} > \lambda_{E,2}$). This highlights the effect of arrival rates beyond their impact on the scaling factor. \footnote{{We note that the constants computed in Theorem \ref{th:bilat_prioH} allow us to
determine whether market one is better off or worse off for any $(\lambda_{H,2},\lambda_{E,2})$, and to compute the  boundary separating the two regions, in the limit $p_H \rightarrow 0$. }}

\label{subsec:merging}
\begin{figure}[h!]
  \centering
 \includegraphics[width=8cm]{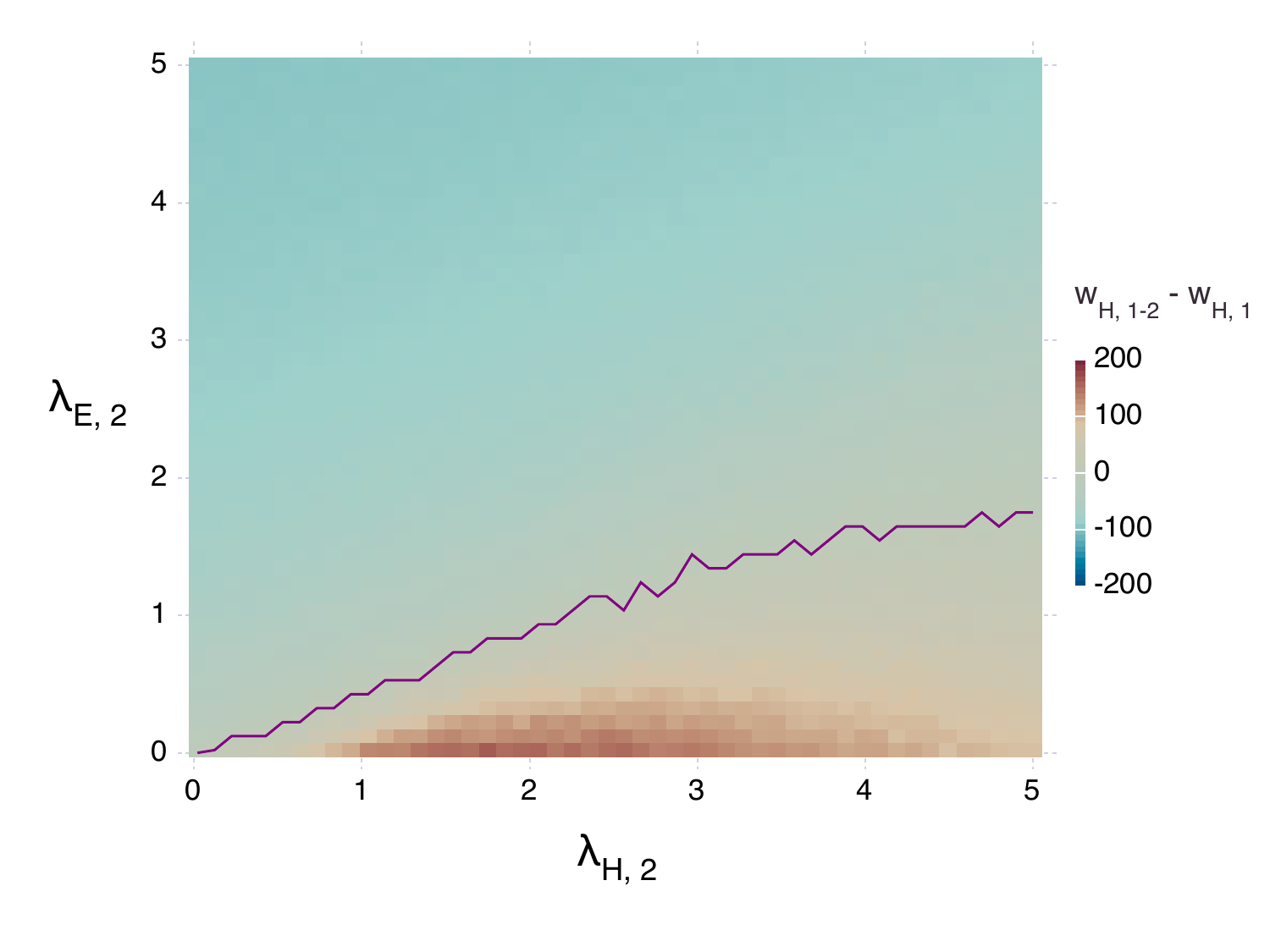}
  \includegraphics[width=8cm]{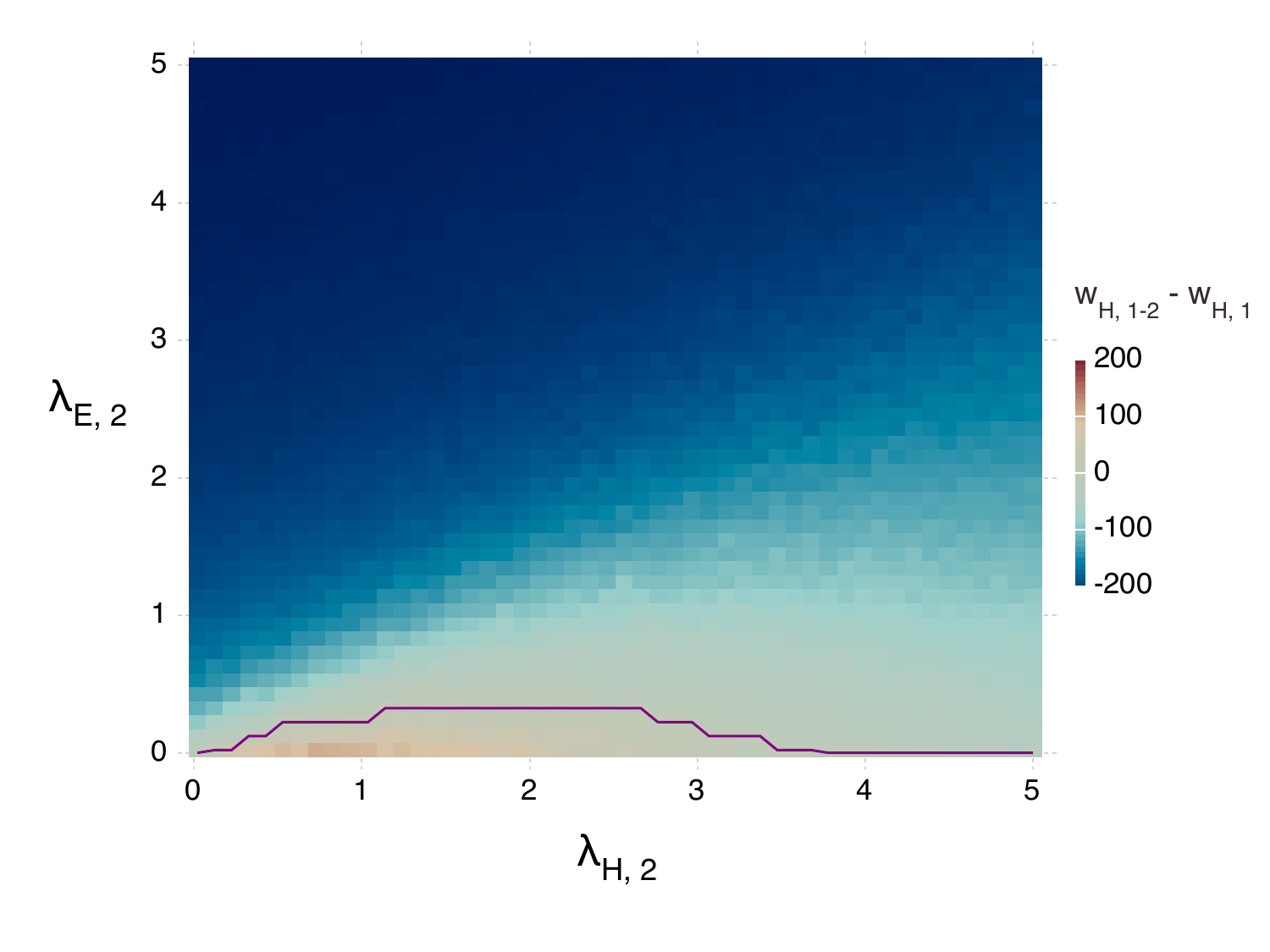}
\caption{Change in the waiting time for $H$ agents of  the first market: $w_{H, 1-2} - w_{H,1}$, as a function of $(\lambda_{H,2}, \lambda_{E,2})$, for $p_E = 0.5$, $p_H = 0.02$, after $T = {10^5}$ iterations. Left subplot corresponds to $\lambda_{H,1} = 1, \lambda_{E,1} = {1.3}$, right subplot corresponds to $\lambda_{H,1} = {1.3}, \lambda_{E,1} = 1$. The purple line separates the region where the waiting time increases after merging (below the line) and the region where it decreases (above the line).} 
\label{fig:merging}
\end{figure}



\subsection{Impact of priorities in bilateral matching}
\label{subsec:sim_priorities}

We compare here the average waiting time of $H$ agents under the {\it BilateralMatch(H)} and {\it BilateralMatch(E)} policies. From Theorems \ref{th:bilat_prioH} and  \ref{th:bilat_prioE} it follows that (i) when  $\lambda_H > \lambda_E$, asymptotically, the average  waiting time of $H$  agents is the same under both policies, but (ii) when $\lambda_H < \lambda_E$, the average waiting time of $H$ agents under $\BH$ is at  most the average waiting time under $\BE$. However, numerical simulations suggest that the average waiting time of $H$ agents is indeed strictly smaller under $\BH$ than under $\BE$ (Figure \ref{fig:all_policies} left).
{For instance, in simulation setting of Figure \ref{fig:all_policies}, when  $\lambda_H = 4$ and $\lambda_E = 5$ , we have $w_H^{\BE} = 534$ while $w_H^{\BH} = 388$. }
The  average waiting times of $E$ agents are plotted in  Figure \ref{fig:all_policies}(right).
\begin{figure}[ht!]
  \centering
 \includegraphics[width=8cm]{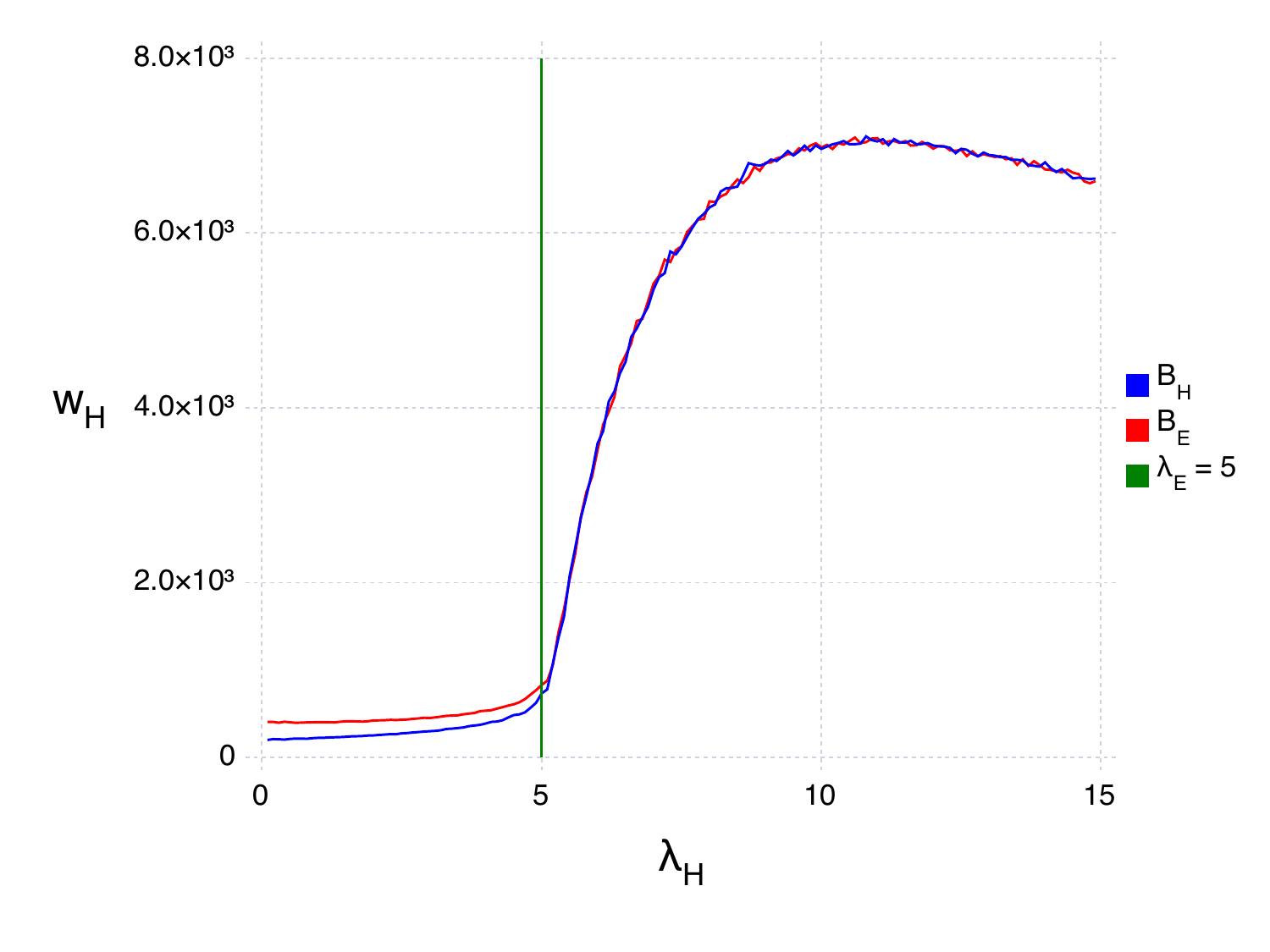}
 \includegraphics[width=8cm]{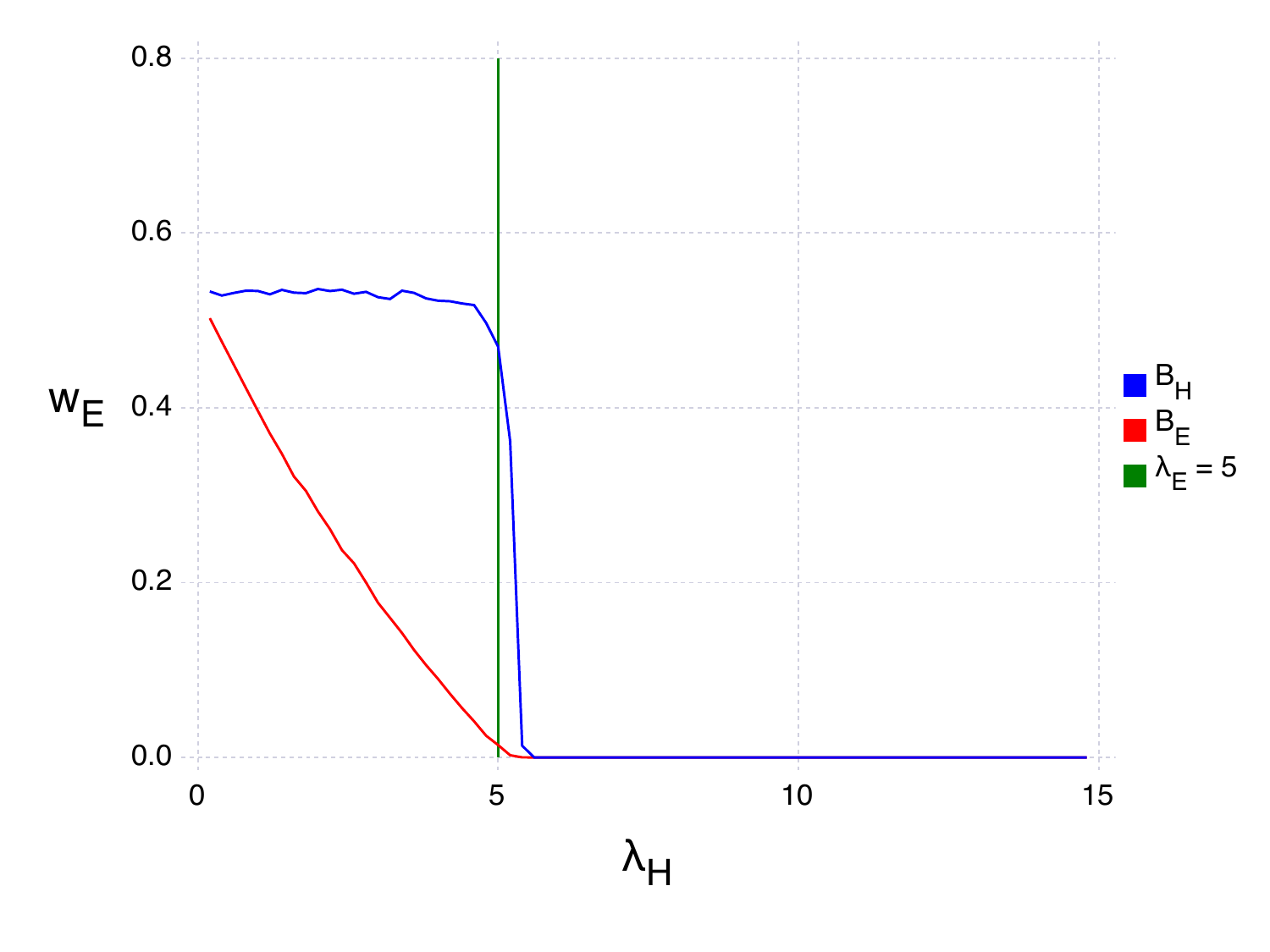}
\caption{Comparison of $w_H$ (left) and $w_E$ (right) for $\BH$, $\BE$, as a function of $\lambda_H$, for a fixed $\lambda_E = 5$, $T = 2\cdot 10^6$, $p_E = 0.5$ and $p_H = 0.002$.
}
\label{fig:all_policies}
\end{figure}

The main insight is that the benefit from assigning priority to hard-to-match agents varies based on the composition of the market. Further, our qualitative insights can be useful in understanding the tradeoffs that may arise in markets where easy-to-match agents have outside options. For example, when $\lambda_H>\lambda_E$, there is no tradeoff from prioritizing $E$ agents. 
This issue arises in kidney exchange, where very easy-to-match patient-donor pairs (such as compatible pairs) may choose to get transplanted elsewhere.

\subsection{Comparative statics in chain matching with \texorpdfstring{$p_E <1$}{pE < 1}}
\label{subsec:sim:ndd_E_increase}

We run simulations using {\it ChainMatch(d)} to numerically explore the effects  varying $d$, $\lambda_E$, and $\lambda_H$  have on   $w_H^{\Cd}$.

We find that $w_H^{\Cd}$ decreases as the arrival rate of either types increases (Figure \ref{wtimeChains} left and middle). Moreover, the value of an additional altruistic agent also  diminishes with increasing $\lambda_E$, $\lambda_H$ or $d$.

Further, as $p_H$ decreases, the impact of $d$ vanishes (Figure \ref{wtimeChains} right).
Observe that the simulation results (in which $p_E<1$) are qualitatively aligned with the predictions of Proposition \ref{prop:chain:p1}  even for the case $p_E < 1$.

\begin{figure}[ht!]
  \centering
 \includegraphics[width=5.4cm]{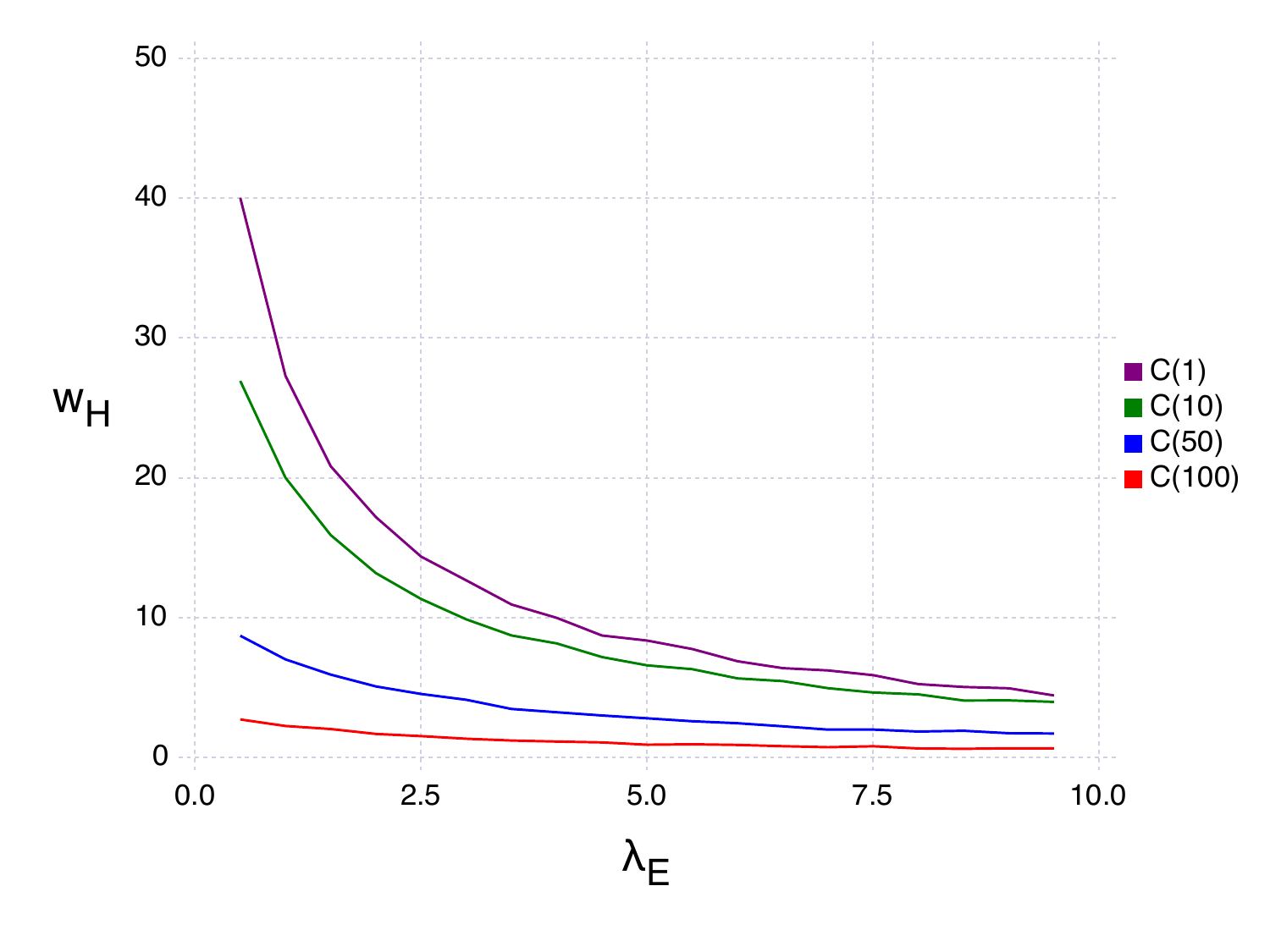}
  \includegraphics[width=5.4cm]{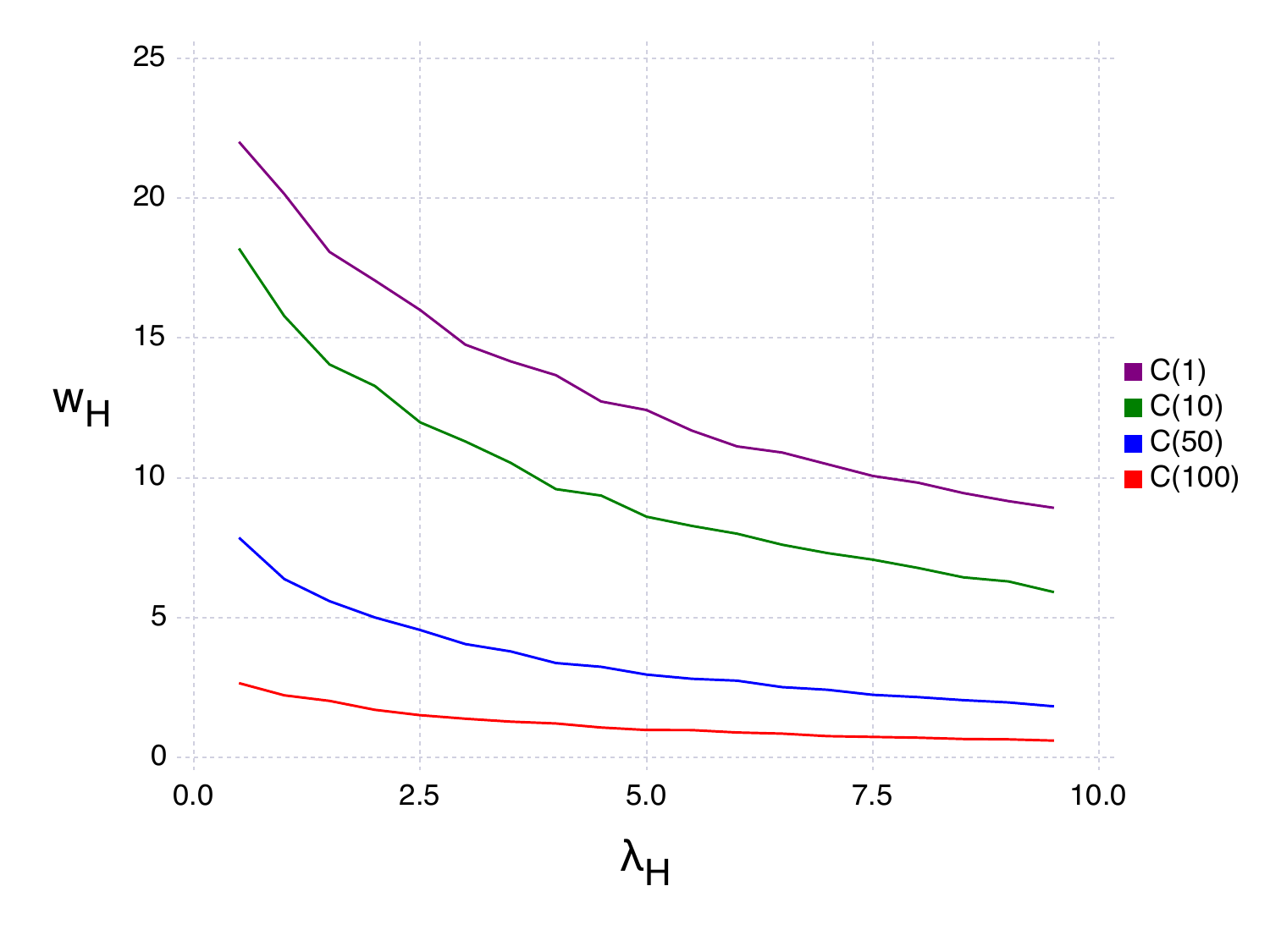}
  \includegraphics[width=5.4cm]{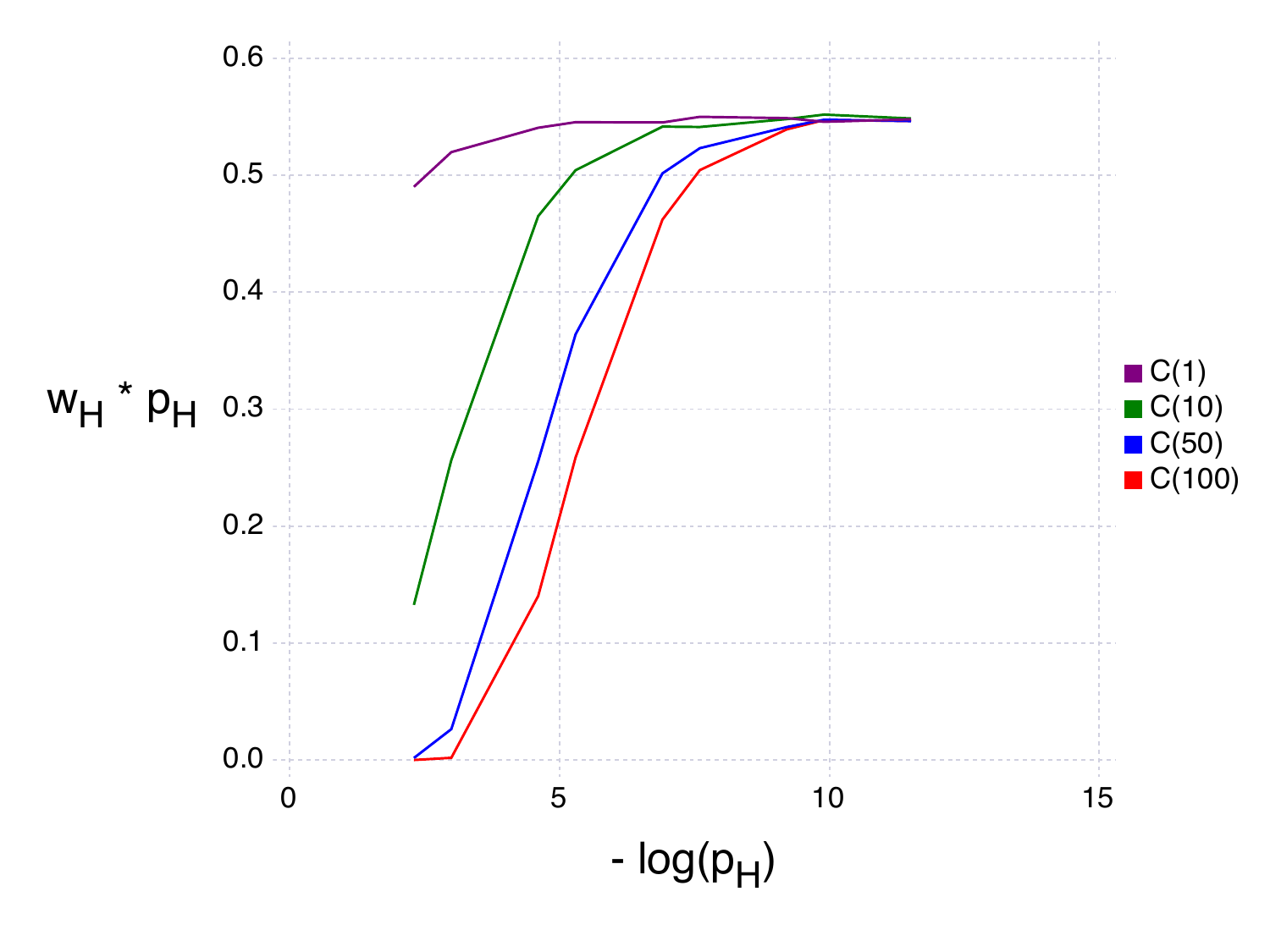}
\caption{Left: $w_H^{\Cd}$ for varying values of $d$ as a function of $\lambda_E$,  for fixed $\lambda_H = 2$, $T = 10^5$, $p_E = 0.5$, $p_H = 0.02$.
Middle: $w_H^{\Cd}$ for varying values of $d$ as a function of $\lambda_H$,  for fixed $\lambda_E = 2$, $T = 10^5$, $p_E = 0.5$, $p_H = 0.02$.
Right: Normalized waiting times  (i.e., $p_H w_H^{\Cd}$) in the case of chains for varying values of $d$ as a function of $p_H$, for fixed $\lambda_H = 2$, $\lambda_E = 1$, $T = 10^5$, $p_E = 0.5$.}
\label{wtimeChains}
\end{figure}

Next we study the loss from employing a local search for forming chain-segments rather than looking for the maximum-length path at each chain-segment-formation. For this, we define a new policy {\it Max-Chains} that upon starting a chain-segment searches for the chain-segment that maximizes lexicographically the number of $H$ agents matched, while breaking ties over matching more agents over all.

\begin{figure}[ht!]
  \centering
 \includegraphics[width=11cm]{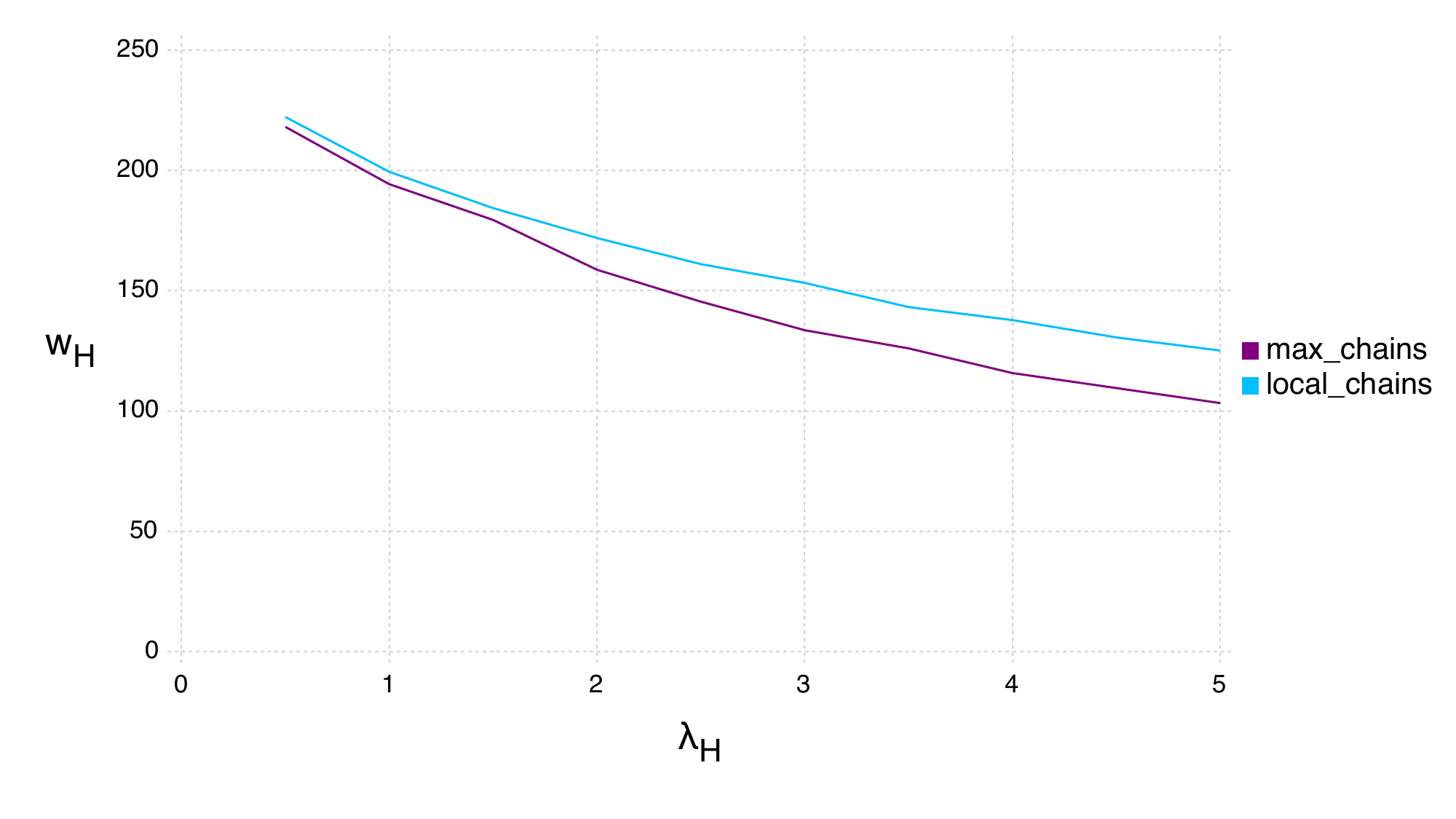}
\caption{Waiting times $w_H$ for chains conducted with local-search ($\Cd$) and {\it Max-Chains} as a function of $\lambda_H$,  for fixed $\lambda_E = 2$, $T = 10^5$, $p_E = 0.5$, $p_H = 0.002$.}
\label{fig:max_vs_local_chains}
\end{figure}
We observe that the benefit of using {\it Max-Chains } is small when $\lambda_H$ is small compared to $\lambda_E$, and it increases as $\lambda_H$ increases. If we consider $\lambda_E / 2 \leq \lambda_H \leq 2 \lambda_E$ as the practical range relevant to the kidney exchange programs, our simulations suggest that the loss ranges between $5$ to $15$\%.

\subsection{Impact of the matching technology: bilateral vs. chain matching}
\label{subsec:altruistic}
Theorems \ref{th:bilat_prioH} and \ref{th:chain} imply that  for {\em any} arrival rates $(\lambda_H, \lambda_E)$   matching through chains even with only one initial altruistic agent (i.e., under {\it ChainMatch(1)}) results in shorter average waiting time for $H$  agents. The theoretical gap is significant  when $\lambda_H>\lambda_E$.
We run  numerical simulations for a variety of  parameters  to examine these differences (see Figure \ref{fig:chains_vs_BH}).

\begin{figure}[h!]
  \centering
 \includegraphics[width=8cm]{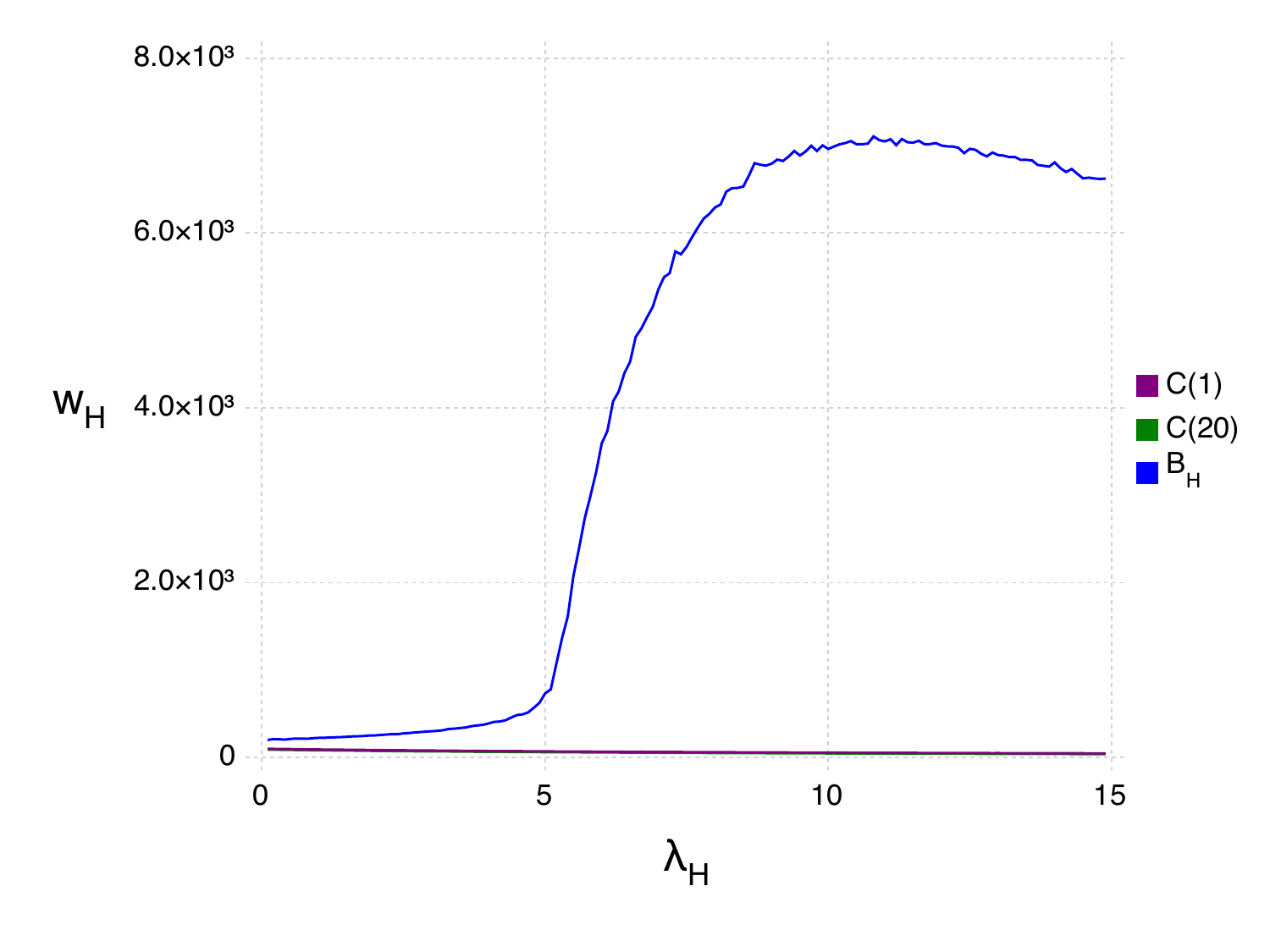}
\caption{Comparison of $w_H$,  for $\BH$, $\mathcal{C}(1)$, and $\mathcal{C}(20)$, as a function of $\lambda_H$, for a fixed $\lambda_E = 5$, $T = 2*10^6$, $p_E = 0.5$ and $p_H = 0.02$.}
\label{fig:chains_vs_BH}
\end{figure}

To further highlight the benefit of matching through chains, we consider the following scenario: Suppose market $1$ has rates $(\lambda_{H,1}, \lambda_{E,1})$ with $\lambda_{H,1} < \lambda_{E,1}$ and is endowed with $d$ altruistic agents and employs policy {\it ChainMatch(d)}. Now consider a second market with arrival rates $(\lambda_{H,2}, \lambda_{E,2})$ that does not have any altruistic agents and therefore employs  {\it BilateralMatch(H)}. Further suppose $\lambda_{H,1}= \lambda_{H,2} = \lambda_{H}$; how many more $E$ agents does market $2$ need to attract to be able to compete with market $1$ in term of average waiting times of $H$ agents? In the limit $p_H \rightarrow 0$, by  Theorems \ref{th:bilat_prioH} and \ref{th:chain}, for this to happens it is necessary that:

$$\frac{\ln \left( \frac{\lambda_H}{\lambda_{E,1} (1 - (1 - p_E)^d)} + 1 \right)}{\lambda_H} \geq \frac{\ln \left( \frac{\lambda_{E,2}}{\lambda_{E,2} -\lambda_H} \right)}{\lambda_H p_E},$$
which is equivalent to:
$$ \lambda_{E,2} \geq \frac{\lambda_H (\lambda_H + \lambda_{E,1} (1 - (1 - p_E)^d))^{p_E}}{ (\lambda_H + \lambda_{E,1} (1 - (1 - p_E)^d))^{p_E} -  (\lambda_{E,1} (1 - (1 - p_E)^d))^{p_E}}.$$
Note that the above condition is only a necessary condition, and valid in the limit $p_H \rightarrow 0$.
In the case where $p_E = 1$,  Proposition \ref{prop:chain:p1} makes this also a sufficient condition, and it simplifies to $\lambda_{E,2} \geq \lambda_H + \lambda_{E,1}$. 
In table \ref{tab}, we report the numerical values for $\lambda_{E,2}$ such that in simulations $w_{H,2}^\BH = w_{H,1}^\Cd$. 


\begin{table}[]
\centering
\caption{$\lambda_{E,2}$ as a function of $p_E$ and $d$, for $p_H = 0.02$, $\lambda_H = 1$ and $\lambda_{E,1} = 2$, $T=10^6$.}
\label{tab}
\begin{tabular}{|l||l|l|l|l|l|}
\hline
 $ p_E $ & 0.1 & 0.3 & 0.5 & 0.9 & 1.0 \\ \hline \hline
d = 1 & 20.75 & 8.45 & 5.4 & 3.3 & 3.0 \\ \hline
d = 10 & 27.15 & 10.25 & 6.55 & 3.9 & 3.6 \\ \hline
d = 50 & 66.05 & 24.8 & 15.1 & 9.0 & 8.15 \\ \hline
\end{tabular}
\end{table}
\subsection{Theoretical bounds vs heuristics vs. simulation}
\label{sec:app:tightness}


In two cases, our theoretical results yield bounds which are are not tight. However, in each of these cases we generate a heuristic guess for the exact behaviour. We plot here the simulation results, our heuristically generated guess (described later in Section \ref{subsec:BE} and Appendix \ref{app:chains_heuristic}) and the theoretical bounds for a variety of parameters.
The first case is under the policy \textit{{BilateralMatch(E)}} when $\lambda_H <\lambda _E$. Figure \ref{wtimeHeuristicBE} shows that our heuristic analysis (described in Section \ref{subsec:BE}) results in a guess of $\frac{\ln \left(\frac{\lambda_E + \lambda_H}{\lambda_E - \lambda_H}\right)}{p_E p_H}$ that coincides with the simulation results. 
 The figure further illustrates the behavior of our theoretical bounds for different parameters.
\begin{figure}[h!]
  \centering
 \includegraphics[width=8cm]{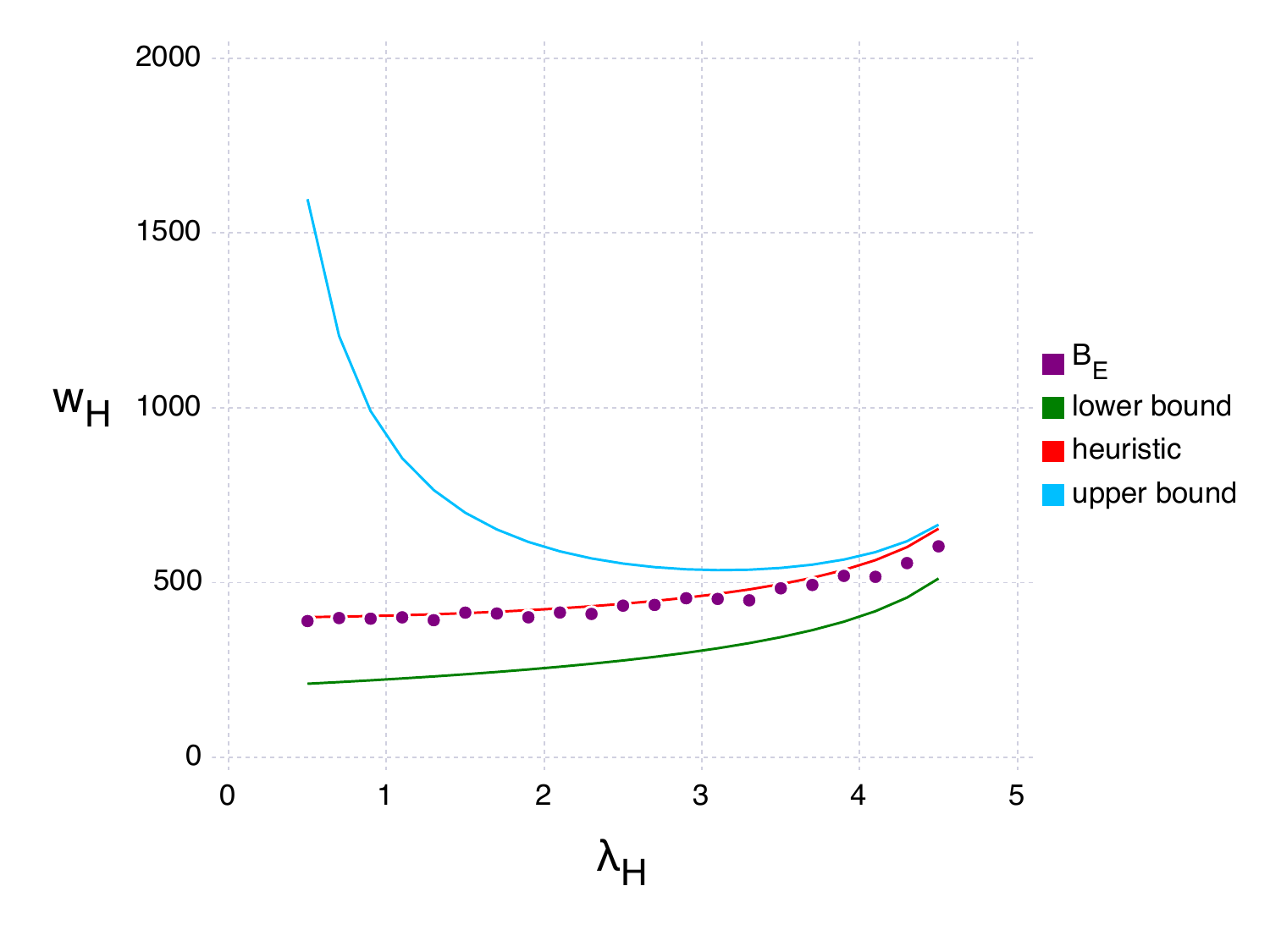}
  \includegraphics[width=8cm]{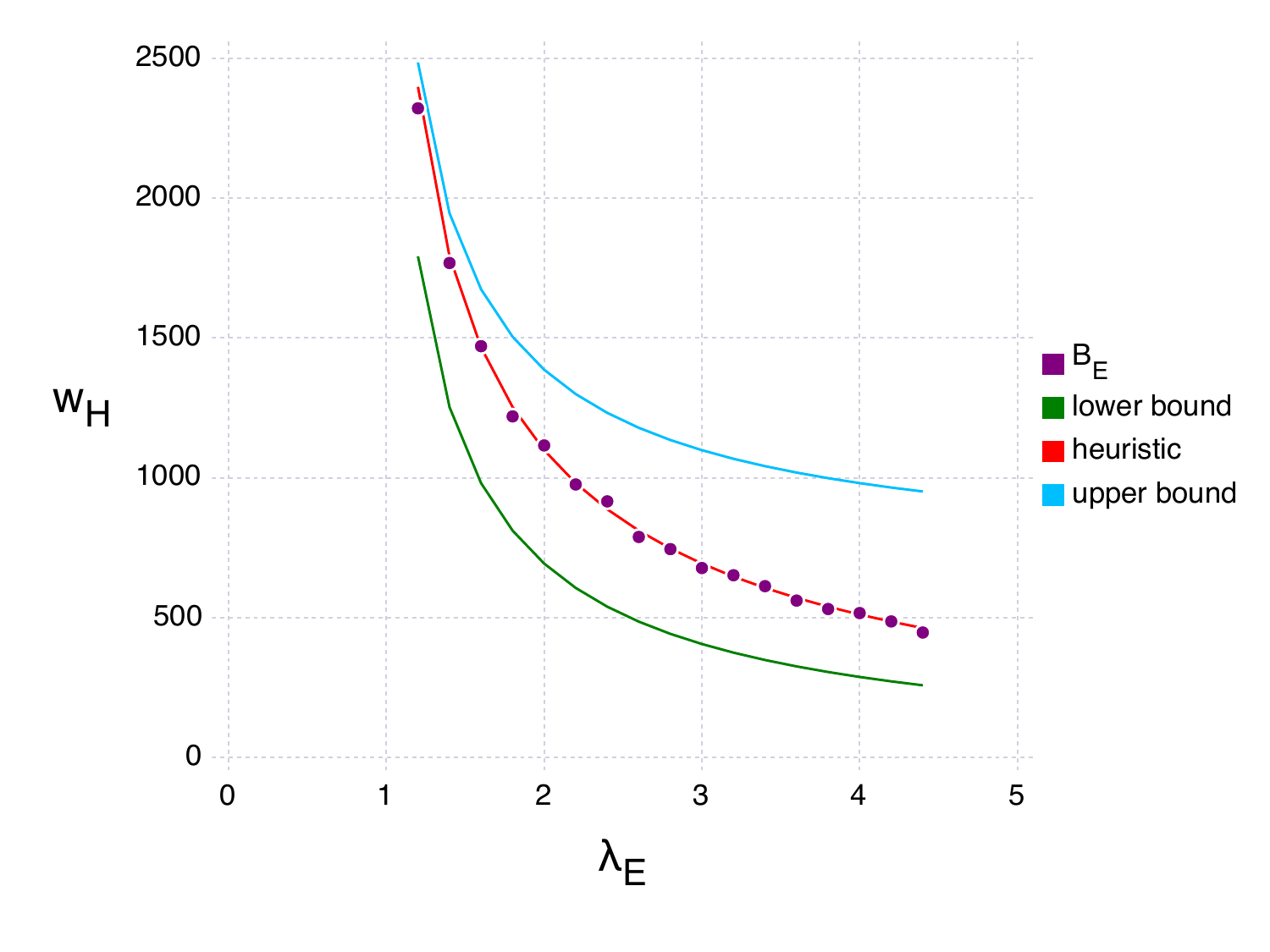}
\caption{Left: $w_H^{\BE}$ as a function of $\lambda_H$, for $\lambda_E = 5$, $T = 10^5$, $p_E = 0.5$, $p_H = 0.002$.
Right: $w_H^{\BE}$ as a function of $\lambda_E$, for $\lambda_H = 1$, $T = 10^5$, $p_E = 0.5$, $p_H = 0.002$. }
\label{wtimeHeuristicBE}
\end{figure}


\begin{figure}[h!]
  \centering
 \includegraphics[width=8cm]{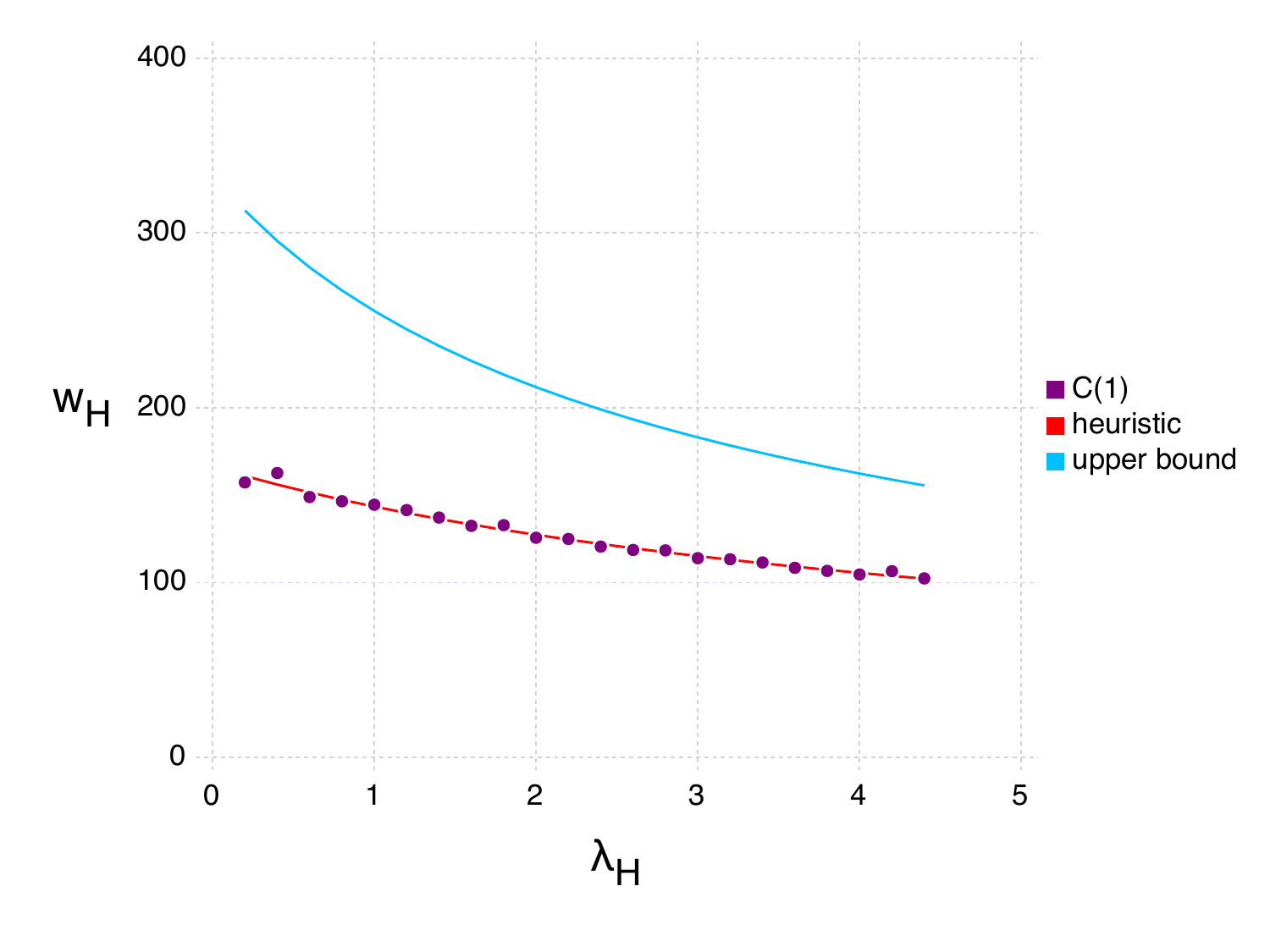}
 \includegraphics[width=8cm]{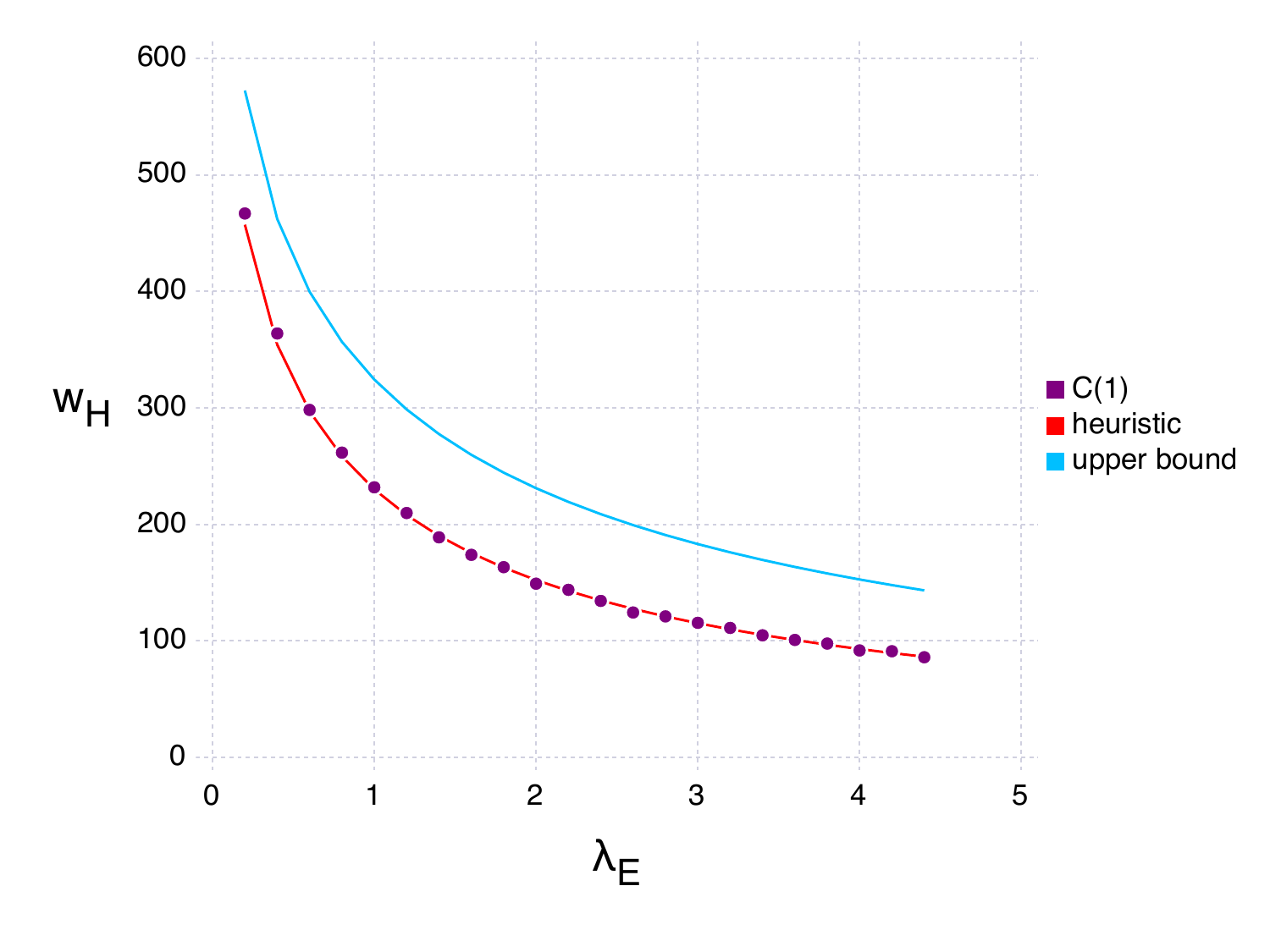}
\caption{Left: $w_H^{\mathcal{C}(1)}$ as a function of $\lambda_H$, for $\lambda_E = 3$, $T = 10^5$, $p_E = 0.5$, $p_H = 0.002$.
Right: $w_H^{\mathcal{C}(1)}$   as a function of $\lambda_E$, for $\lambda_H = 3$, $T = 10^5$, $p_E = 0.5$, $p_H = 0.002$.}
\label{wtimeHeuristicChains}
\end{figure}


The second case is under the policy \textit{{ChainMatch(d)}} when $p_E<1$.  Here too, Figure \ref{wtimeHeuristicChains} shows that  that our heuristic guess $\ln \left( \frac{\lambda_H + \lambda_E}{\lambda_H (1 - (1 - p_H)^d) + \lambda_E} \right)/p_H$ (described in Appendix \ref{app:chains_heuristic})  coincides with the numerical simulations.


\section{Proof ideas and outline of analysis}
\label{sec:proofs}


The analysis of each policy follows a similar pattern, although technically analyzing the bilateral setting and the chain setting are very different. 
For bilateral policies, we first offer a heuristic that will help guessing the value of  $\E[H^{\cP}]$ (which is proportional to the average waiting time) and then proceed to rigourously analyze $\E[H^{\cP}]$.
For the chain policy, we first couple the underlying Markov chain with a $1$-dimensional chain whose number of $H$ agents serve as an upper-bound on the number of $H$ agents under  \textit{{ChainMatch(d)}} policy. We then proceed to analyze the expected number of $H$ agents in the coupled chain.
In all three settings, the main idea is to prove that $H^{\cP}$ is concentrated around $\E[H^{\cP}]$ without directly computing the steady-state distribution, and based on the exponential decay of the tail distribution when moving away from the expected value.

We  often use the following  notations to avoid terms that vanish in the limit $p_H \rightarrow 0$.
Let $f,g: [0,1] \rightarrow \mathbb{R}$; We write that  $f=o(g)$ if $\lim_{p_H \rightarrow 0}\frac{f(p_H)}{g(p_H)}=0$ and write that  $f=O(g)$ if
$\limsup_{p_H \rightarrow 0}\frac{f(p_H)}{g(p_H)}< \infty$.


\subsection{The \textit{BilateralMatch(H)} policy}
\label{subsec:BH}

In this section we analyze the  policy $\BH$, which forms myopically bilateral exchanges while prioritizing $H$ agents. Under this policy, the evolution of the number of $H$ and $E$ agents in the market can be modeled by a  CTMC $[H_t,  E_t] \in \N^2$ with the following transition rates.

\begin{subequations}
\begin{align}
Q^{\BH}([h,e],[h+1,e]) &= \lambda_H (1 - p_H^2)^{h} (1 - p_E p_H)^{e} \label{eq:transtion:1}\\
Q^{\BH}([h,e],[h-1,e]) &=  \lambda_H (1 - (1 - p_H^2)^{h}) + \lambda_E (1 - (1 - p_E p_H)^{h}) \label{eq:transtion:2}\\
Q^{\BH}([h,e],[h,e + 1])&=\lambda_E (1 - p_E p_H)^{h} (1 - p_E^2)^{e} \label{eq:transtion:3}\\
Q^{\BH}([h,e],[h,e-1])&= \lambda_H (1 - p_H^2)^{h} (1 - (1 - p_E p_H)^{e}) + \lambda_E (1 - p_E p_H)^{h} (1 - (1 - p_E^2)^{e})) \label{eq:transtion:4}
\end{align}
\end{subequations}

The rates are computed based on the Poisson thinning property, simple counting arguments, and our assumption that edges are formed independently.

\begin{itemize}
\item[-] Rightward rate \eqref{eq:transtion:1}: moving from  $[h,e]$ to $[h+1,e]$ happens when an $H$ agent arrives and cannot form a cycle with any of the existing $H$ agents (with probability $(1 - p_H^2)^{h}$) nor with any of the existing $E$ agents (with probability $(1 - p_E p_H)^{e}$).
\item[-] Leftward rate \eqref{eq:transtion:2}: moving from $[h,e]$ to $[h-1,e]$ happens when an $H$ agent arrives and forms a cycle with at least one of the existing $H$ agents  (probability $(1 - (1 - p_H^2)^{h})$) or an $E$ agent arrives and forms a cycle with at least one of the existing $H$ agents (probability $(1 - (1 - p_E p_H)^{h}) $).
\item[-] Upward rate \eqref{eq:transtion:3}: moving  from $[h,e]$ to $[h,e+1]$ happens when an $E$ agent arrives and cannot form a cycle with any of the existing $H$ agents (probability $(1 - p_E p_H)^{h}$) nor with with any of the existing $E$ agents (probability $(1 - p_E^2 )^{e}$).
\item[-] Downward rate \eqref{eq:transtion:4}: moving from $[h,e]$ to $[h,e-1]$ happens when an $H$ agent arrives and cannot form a cycle with any of the existing $H$ agent (probability $(1 - p_H^2)^{h}$) but can form a cycle with an existing $E$ agents (probability $(1 - (1 - p_E p_H)^{e})$), or an $E$ agent arrives that cannot form a cycle with any of the existing $H$ agents (probability $(1 - p_Ep_H)^{h}$) but can form a cycle with an existing $E$ agent (probability $(1 - (1 - p_E^2)^{e})$).
\end{itemize}


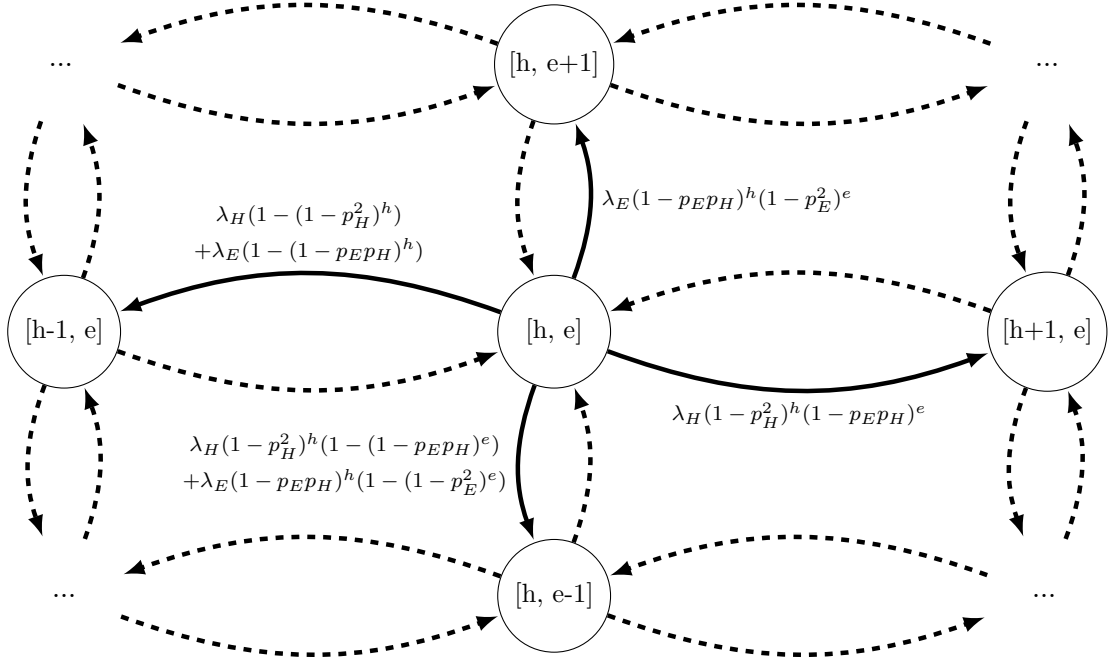
\begin{figure}[htbp]
\centering
\begin{tikzpicture}
[align=center,node distance=1.7cm]
\node[state, minimum size=1.5cm] (s) {\small [h, e]};
\node[state, minimum size=1.5cm, node distance=5cm, right=of s] (r) {\small [h+1, e]};
\node[state, minimum size=1.5cm, node distance=5cm, left=of s] (l) {\small [h-1, e]};
\node[state, minimum size=1.5cm, node distance=2cm, above=of s] (t) {\small [h, e+1]};
\node[state, minimum size=1.5cm, node distance=2cm, below=of s] (b) {\small [h, e-1]};
\node[draw=none, minimum size=1.5cm, node distance=5cm, right=of b] (br) {...};
\node[draw=none, minimum size=1.5cm, node distance=5cm, left=of b] (bl) {...};
\node[draw=none, minimum size=1.5cm, node distance=5cm, right=of t] (tr) {...};
\node[draw=none, minimum size=1.5cm, node distance=5cm, left=of t] (tl) {...};
        \draw[every loop, line width=0.6mm, >=latex]
		(s) edge[bend right = 20,auto=left] node[below] {\scriptsize $\lambda_H (1 - p_H^2)^{h} (1 - p_E p_H)^{e}$} (r)
		(r) edge[bend right = 20, dashed] node {} (s)
		(s) edge[bend right = 20, auto=right] node {{\scriptsize $\lambda_H (1 - (1 - p_H^2)^{h})$} \\ {\scriptsize $+ \lambda_E (1 - (1 - p_E p_H)^{h})$}} (l)
		(l) edge[bend right = 20, dashed] node {} (s)
		(s) edge[bend right = 20, auto=right] node {\scriptsize $\lambda_E (1 - p_E p_H)^{h} (1 - p_E^2)^{e}$} (t)
		(t) edge[bend right = 20, dashed] node {} (s)
		(s) edge[bend right = 20, auto=right] node {{\scriptsize $\lambda_H (1 - p_H^2)^{h} (1 - (1 - p_E p_H)^{e}) $} \\ {\scriptsize $+ \lambda_E (1 - p_E p_H)^{h} (1 - (1 - p_E^2)^{e})$}} (b)
		(b) edge[bend right = 20, dashed] node {} (s)
		(b) edge[bend right = 20, dashed] node {} (br)
		(br) edge[bend right = 20, dashed] node {} (b)
		(b) edge[bend right = 20, dashed] node {} (bl)
		(bl) edge[bend right = 20, dashed] node {} (b)
		(t) edge[bend right = 20, dashed] node {} (tr)
		(tr) edge[bend right = 20, dashed] node {} (t)
		(t) edge[bend right = 20, dashed] node {} (tl)
		(tl) edge[bend right = 20, dashed] node {} (t)
		(r) edge[bend right = 20, dashed] node {} (br)
		(br) edge[bend right = 20, dashed] node {} (r)
		(l) edge[bend right = 20, dashed] node {} (bl)
		(bl) edge[bend right = 20, dashed] node {} (l)
		(r) edge[bend right = 20, dashed] node {} (tr)
		(tr) edge[bend right = 20, dashed] node {} (r)
		(l) edge[bend right = 20, dashed] node {} (tl)
		(tl) edge[bend right = 20, dashed] node {} (l);

\end{tikzpicture}
\caption{Transition rates (solid arrows) under the CTMC induced by the $\BH$ policy.} 
\label{fig:2d:walk}
\end{figure}

Note that the process is a $2$-dimensional continuous-time spatially non-homogeneous random walk. Figure \ref{fig:2d:walk} illustrates this random walk along with its transition rates. Also observe that the
leftward and downward rates \eqref{eq:transtion:2} and \eqref{eq:transtion:4} depend on the priority assigned to  $H$ agents, and these rates will change when prioritizing $E$ agents, as we will see in the Subsection \ref{subsec:BE}. 
However, fixing the priority,  changing the tie-breaking rule between agents of the same type (for example, favoring agents with longer waiting times instead of selecting one at random) does not change the transition rates. 


In Appendix \ref{app:existence_proofs}, we prove that the above (irreducible) CTMC is positive recurrent, and therefore reaches steady-state. This is  intuitive given the above transition rates and the  ``self-regulating'' behavior of the process. The larger the market, the larger the probability that an arriving agent can form a cycle. 
Note that in steady-state, the expected drift for this CTMC in both horizontal and vertical dimension is zero. 
The drifts are given in \eqref{eq:transtion:1}-\eqref{eq:transtion:4}, and therefore,

\begin{subequations}
\begin{align}
& \E\left[\lambda_H (1 - p_H^2)^{H^{\BH}} (1 - p_E p_H)^{E^{\BH}}  - \lambda_H (1 - (1 - p_H^2)^{H^{\BH}}) - \lambda_E (1 - (1 - p_E p_H)^{H^{\BH}})\right] = 0  \label{eq:approx1}\\
& \E[\lambda_E (1 - p_E p_H)^{H^{\BH}} (1 - p_E^2)^{E^{\BH}}  -  \lambda_H(1 - p_H^2)^{H^{\BH}} (1 - (1 - p_E p_H)^{E^{\BH}}) \nonumber \\
&\text{     }  - \lambda_E (1 - p_E p_H)^{H^{\BH}} (1 - (1 - p_E^2)^{E^{\BH}})] = 0 \label{eq:approx2}
\end{align}
\end{subequations}

Assuming that the random variables $H^{\BH}$ and $E^{\BH}$ are very concentrated around their expectations, a reasonable approximation is to move the expectation inside the functions and solve the above system of nonlinear equations, and thus obtain  approximations for $\E[H^{\BH}]$ and $\E[E^{\BH}]$.

\begin{itemize}
\item[-]  For $\lambda_H < \lambda_E$, if we plug $[H^{\BH},E^{\BH}] = \left[\frac{\ln \left(\frac{\lambda_E}{\lambda_E - \lambda_H}\right)}{p_E p_H}, \frac{-\ln (2)}{\ln (1-p_E^2 )} \right]$ into \eqref{eq:approx1}-\eqref{eq:approx2},  the right-hand sides will be $O(p_H)$ terms.
\item[-] For $\lambda_H > \lambda_E$, if we plug $[H^{\BH},E^{\BH}] = \left[\frac{\ln \left( \frac{2 \lambda_H}{\lambda_H + \lambda_E} \right)}{p_H^2}, 0 \right]$  into \eqref{eq:approx1}-\eqref{eq:approx2},  the right-hand sides will be $O(p_H^2)$ terms.
\end{itemize}

This heuristic exercise provides us the correct value of $\E[H^{\BH}]$ in both cases. To establish this value rigorously and prove Theorem \ref{th:bilat_prioH} we show, in  the following two propositions, that  $H^{\BH}$ is  highly concentrated around its mean.


\begin{restatable}{proposition}{propLowerHprioH}\emph{[Lower-bound]}
\label{prop:lower_H_prio_H}
Under $\BH$ and in steady-state,
\begin{itemize}
\item[-] If $\lambda_H < \lambda_E$,  there exists a constant $c_1$ such that:
$$ \P\left[H^\BH \leq \frac{1}{p_E p_H} \left( \ln\left(\frac{\lambda_E}{\lambda_E - \lambda_H}\right) - c_1 p_H^{1/4}\right) \right]  \leq  o(p_H).$$
\item[-] If $\lambda_H > \lambda_E$,   there exists a constant $c_2$ such that:
$$ \P\left[H^\BH \leq \frac{1}{p_H^2} \left( \ln\left(\frac{2 \lambda_H}{\lambda_E + \lambda_H}\right) - c_2\sqrt{p_H}  \right) \right]  \leq  o(p_H).\footnote{While we do have closed form formulas for {$c_1, c_2$ (similarly for $c_3, c_4$ defined in the next proposition)}, these values are not informative. We refer the reader to the proofs for more details.}
$$
\end{itemize}
\end{restatable}

\begin{restatable}{proposition}{propUpperHprioH}\emph{[Upper-bound]}
\label{prop:upper_H_prio_H}
Under $\BH$ and in steady-state, for any $k \geq 0$,
\begin{itemize}
\item[-] If $\lambda_H < \lambda_E$, there exists a function $\gamma(p_H) = 1 - \sqrt{p_H} + o(\sqrt{p_H})$ and a constant $c_3$ such that:
$$ \P\left[H^\BH \geq \frac{1}{p_E p_H} \left( \ln\left(\frac{\lambda_E}{\lambda_E - \lambda_H}\right) + c_3\sqrt{p_H} \right) + k\right]  \leq \frac{\gamma(p_H)^k}{1 - \gamma(p_H)}.$$
\item[-] 
If $\lambda_H > \lambda_E$, there exists a function $\gamma'(p_H) = 1 - \sqrt{p_H} + o(\sqrt{p_H})$ and a constant $c_4$ such that:
$$ \P\left[H^\BH \geq \frac{1}{p_H^2} \left( \ln\left(\frac{2 \lambda_H}{\lambda_E + \lambda_H}\right) + c_4 \sqrt{p_H} \right) +  k \right]  \leq \frac{\gamma'(p_H)^k}{1 - \gamma'(p_H)}.$$

\end{itemize}
\end{restatable}

Note that in both cases in Proposition \ref{prop:upper_H_prio_H}, if $k = p_H^{-3/4}$ then the right-hand sides  become $o(p_H^2)$.

\noindent The proof of Theorem \ref{th:bilat_prioH} is a straightforward application of these propositions and the details are presented in Appendix \ref{subsec:proof:thm1}. To prove these propositions we derive exponentially decaying bounds on tails of the steady-state distribution of
$H^{\BH}$ and $E^{\BH}$.
In the next subsection we present two auxiliary lemmas that establish such bounds for a general class of $2$-dimensional continuous-time  random walks that includes the random walk defined above. The proof of Propositions \ref{prop:lower_H_prio_H} and \ref{prop:upper_H_prio_H} amount to applying these lemmas with appropriately defined parameters. The proofs are presented in Appendix \ref{subsec:proof:pro2} and \ref{subsec:proof:pro1}, respectively. 

\subsection{Concentration bounds for a general class of $2$-dimensional random walks}
\label{subsec:lemmas}


In the analysis of both \textit{BilateralMatch(H)} and  \textit{BilateralMatch(E)} policies, we repeatedly bound the left-tail or the right-tail of the steady-state distribution of the number of $H$ agents in the market.
These bounds rely on certain properties of the corresponding
$2$-dimensional continuous-time  random walks, which allow us to establish exponential decay on each tail of the steady-state distribution.
To avoid repeating these concentration results for each particular setting, we take a unifying approach and state the following two  auxiliary lemmas that establish concentration results  for a general class of $2$-dimensional random walks under certain conditions. These lemmas maybe useful in other applications that give rise to  similar random walks.



\begin{restatable}{lemma}{concentrationlower}\emph{[Lower-bound]}
\label{lem:lower_bound}
Let  $[X_t, Y_t] \in \mathbb{N}^2$ be a positive recurrent continuous time random walk with transition rate matrix $Q$ and $[X,Y]$ be a corresponding random vector following its steady-state distribution. Suppose the following exist:
\begin{enumerate}[{Condition} 1.]
\item A set $S \subset \N$ and a constant $\epsilon > 0$ such that $\P[Y \not \in S] \leq \epsilon$. \label{cond:1}
\item A non-increasing function $f :  \N \mapsto (0, \infty)$ such that $\forall y \in S$, $Q( [x, y], [x+1, y] ) \geq f(x)$.\label{cond:2}
\item A non-decreasing function $g :  \N \mapsto (0, \infty)$ such that  $\forall y \in S$, $Q( [x, y], [x-1, y] ) \leq g(x)$.\label{cond:3}
\end{enumerate}
Then for all $\rho < 1$ and $\eta \in \N$ such that $\frac{g(\eta+1)}{f(\eta)} < \rho$, and any $k>0$ we have:
$$ \P[X \leq \eta - k] \leq \eta \epsilon \left( 1 + \frac{1}{f(\eta) - g(\eta+1)}\right)  + \frac{\rho^k}{1 - \rho}.$$
%
\end{restatable}

\begin{proof}[Proof of Lemma \ref{lem:lower_bound}]
Let $\pi(x,y)$ be the joint distribution of $[X,Y]$, and let $\pi_X(x) = \sum_{y \geq 0} \pi(x,y)$ be the marginal distribution of $X$.
In steady-state, conservation of flow implies:

\begin{equation*}
\begin{split}
 \sum_{y \in S} \pi(x+ 1,y) Q( [x+1, y], [x, y] ) &+ \sum_{y \not \in S} \pi(x + 1,y) Q( [x+1, y], [x, y] ) \\ &= \sum_{y \in S} \pi(x,y) Q( [x, y], [x+1, y] ) + \sum_{y \not \in S} \pi(x,y) Q( [x, y], [x+1, y]).
\end{split}
\end{equation*}
Using Conditions \ref{cond:2} and \ref{cond:3}, we upper-bound the left hand side and lower-bound the right hand side which results in having:
\begin{equation*}
\begin{split}
 g(x+1) \P[X = x + 1, Y \in S] + \P[X = x+1, Y \not \in S] &\geq f(x) \P[X = x, Y \in S].
\end{split}
\end{equation*}

\noindent Let $\pi_S(x) = \P[X = x, Y \in S] = \sum_{y \in S} \pi(x,y)$. Observe that by Condition \ref{cond:1} we have: $\pi_X(x) \leq \pi_S(x) + \epsilon$.
Using the fact that $g$ is non-decreasing and $f$ is non-increasing, we get for $x \leq \eta $:
$$ \pi_S(x) \leq \frac{g(x+1)}{f(x)} \pi_S(x+1) + \frac{\P[Y \not \in S]}{f(x)} \leq \rho \pi_S(x+1)  + \frac{\epsilon}{f(\eta)}.$$
We can subtract $\frac{\epsilon/f(\eta)}{1 - \rho}$ from both sides and iterate: for all $j \geq 0$,
$$\pi_S(\eta - j) - \frac{\epsilon/f(\eta)}{1 - \rho}\leq \rho^j \left(\pi_S(\eta) - \frac{\epsilon / f(\eta)}{1 - \rho}\right) \leq \rho^j.$$
This allows us to conclude that for any $k>0$:
\begin{equation*}
\begin{split}
 \P[X \leq \eta - k]  &= \sum_{i = 0}^{ \eta - k} \pi_X(i)
 \leq (\eta - k)\epsilon + \sum_{j = k}^{ \eta} \pi_S(\eta - j)\\
 &\leq  (\eta - k)\epsilon  \left(1 +  \frac{1/f(\eta)}{1 - \rho}\right) + \sum_{j = k}^{ \eta}\rho^j
 \leq  \eta \epsilon \left(1 + \frac{1}{f(\eta)(1 - \rho)}\right) + \frac{\rho^k}{1 - \rho}.
 \end{split}
\end{equation*}
\end{proof}
%


\begin{restatable}{lemma}{concentrationupper}\emph{[Upper-bound]}
\label{lem:upper_bound}
Let  $[X_t, Y_t] \in \mathbb{N}^2$ be a positive recurrent continuous time random walk with transition rate matrix $Q$ and let $[X,Y]$ be a corresponding random vector following its steady-state distribution. Suppose the following exist:
\begin{enumerate}[{Condition} 1.]
\item A mapping $S : \N \mapsto 2^\N$ and two constants $c \in \R^{+}, \delta \in (0,1)$ such that \\ $\P[Y \not \in S(x)] \leq c \delta^x$.\label{cond:2_1}
\item Two functions $f, g :  \N \mapsto (0,\infty)$ such that  $\forall y \in S$, $Q( [x, y], [x+1, y] )   \leq f(x)$ and $Q( [x, y], [x-1, y] )  \geq g(x)$.\label{cond:2_2}
\end{enumerate}
Then for all $\eta > 0$ and $\rho \in [\delta, 1)$ such that  $\forall x \geq \eta$, $\frac{f(x)}{g(x+1)} \leq \rho$, and $\frac{\delta^x}{g(x+1)} \leq \frac{\rho^x}{g(\eta +1)}$, and for any $k>0$ we have:
\footnote{Note that the above conditions are weaker than that of Lemma \ref{lem:lower_bound} (where $f$ is non-increasing, $g$ is non-decreasing and $\frac{f(\eta)}{g(\eta + 1)} \leq \rho$). We will need this for the proofs of Propositions \ref{prop:lower_H_prio_H} and \ref{prop:lower_H_prio_E} where the corresponding function $g$ is not monotone.}
$$ \P[X \geq \eta + k] \leq \frac{\rho^k}{1 - \rho} \left(1 + c + \frac{ c(k+1)}{g(\eta + 1) - f(\eta)} \right).$$
\end{restatable}

%
\noindent{The proof of Lemma \ref{lem:upper_bound} follows similar arguments to that of Lemma \ref{lem:lower_bound} and is deferred to Appendix \ref{app:concentration_upper}. }

\subsection{The \textit{BilateralMatch(E)} policy}
\label{subsec:BE}

The policy $\BE$ forms myopically bilateral exchanges while prioritizing $E$ agents. The transition rate of the underlying CTMC are as follows.


\begin{subequations}
\begin{align}
Q^{\BE}([h,e],[h+1,e]) &= \lambda_H (1 - p_H^2)^{h} (1 - p_E p_H)^{e}\label{eq:transtionE:1} \\
Q^{\BE}([h,e],[h-1,e]) &=  \lambda_H (1 - p_E p_H)^{e}(1 - (1 - p_H^2)^{h}) + \lambda_E (1 - p_E^2)^{e}(1 - (1 - p_E p_H)^{h}) \label{eq:transtionE:2}\\
Q^{\BE}([h,e],[h,e + 1])&=\lambda_E (1 - p_E p_H)^{h} (1 - p_E^2)^{e} \label{eq:transtionE:3}\\
Q^{\BE}([h,e],[h,e-1])&= \lambda_H (1 - (1 - p_E p_H)^{e}) + \lambda_E (1 - (1 - p_E^2)^{e}) \label{eq:transtionE:4}
\end{align}
\end{subequations}
%

The rates are computed similarly to those under the \textit{BilateralMatch(H)}.
Observe that prioritizing $E$ results in different leftward and downward rates \eqref{eq:transtionE:2} and \eqref{eq:transtionE:4} than the corresponding rates under  \textit{BilateralMatch(H)}.
In particular note that in the leftward rate (moving from $[h,e]$ to $[h-1,e]$), the probability that an arriving $E$ agent matches an existing $H$ agent depends now on the current number of $E$ agents.
This dependency does not exist in \textit{BilateralMatch(H)}.
This makes the analysis of \textit{BilateralMatch(E)} more difficult since we need to compute tight bounds also on the number of $E$ agents in the market.
While we are able to prove such bounds in the case $\lambda_H > \lambda_E$, we are not able to do so in the case $\lambda_H < \lambda_E$. 

As before, we set the  expected drifts at steady-state in both dimensions to zero, resulting in the following system of equations.

%

\begin{subequations}
\begin{align}
& \E[ \lambda_H (1 - p_H^2)^{H^{\BE}} (1 - p_E p_H)^{E^{\BE}} -  \lambda_H (1 - p_E p_H)^{E^{\BE}}(1 - (1 - p_H^2)^{H^{\BE}}) \nonumber \\
& - \lambda_E (1 - p_E^2)^{E^{\BE}}(1 - (1 - p_E p_H)^{H^{\BE}}) ] = 0  \label{eq:approx3} \\
& \E\left[ \lambda_E (1 - p_E p_H)^{H^{\BE}} (1 - p_E^2)^{E^{\BE}} -  \lambda_H (1 - (1 - p_E p_H)^{E^{\BE}}) - \lambda_E (1 - (1 - p_E^2)^{E^{\BE}}) \right] = 0. \label{eq:approx4}
\end{align}
\end{subequations}


Similar to the heuristic analysis for \textit{BilateralMatch(H)}, we can obtain 
the following approximations for $\E[H^{\BE}]$ and $\E[E^{\BE}]$.


\begin{itemize}
\item[-] For the case $\lambda_H < \lambda_E$, if we plug $\left[\frac{\ln \left(\frac{\lambda_E + \lambda_H}{\lambda_E - \lambda_H}\right)}{p_E p_H}, \frac{\ln \left( \frac{\lambda_E + \lambda_H}{2 \lambda_E}\right)}{\ln (1-p_E^2 )} \right]$ into \eqref{eq:approx3}-\eqref{eq:approx4},  the right-hand sides will be $O(p_H)$ terms.
\item[-] For the case, $\lambda_H > \lambda_E$, if we plug  $\left[\frac{\ln \left( \frac{2 \lambda_H}{\lambda_H + \lambda_E} \right)}{p_H^2}, 0 \right]$  into \eqref{eq:approx3}-\eqref{eq:approx4},  the right-hand sides will be $O(p_H^2)$ terms.
\end{itemize}

As  stated in Theorem \ref{th:bilat_prioE}, for the case $\lambda_H > \lambda_E$ the constant for the limit of $\frac{\E[H^{\BE}]}{\lambda_H}$ coincides with the solution given by the above heuristic. For the case, $\lambda_H < \lambda_E$, the constant resulting from the above heuristic argument lies in between the constants of the lower and upper bounds we can prove (in Theorem \ref{th:bilat_prioE}), i.e.,

$$
\frac{\ln\left(\frac{\lambda_E}{\lambda_E - \lambda_H}\right)}{p_E } \leq \frac{\ln \left(\frac{\lambda_E + \lambda_H}{\lambda_E - \lambda_H}\right)}{p_E} \leq \frac{ \ln\left(\frac{2 \lambda_E}{\lambda_E - \lambda_H}\right)}{p_E}. 
$$
In Figure \ref{wtimeHeuristicBE} (Section \ref{sec:app:tightness}), we numerically show that  $\frac{\ln \left(\frac{\lambda_E + \lambda_H}{\lambda_E - \lambda_H}\right)}{p_E} $ is indeed the right constant.

The proof of the case $\lambda_H > \lambda_E$, and the lower bound when $\lambda_H < \lambda_E$ in Theorem \ref{th:bilat_prioE} follows similar steps as that of Theorem \ref{th:bilat_prioH}, and it uses the concentration results of the lemmas stated in the previous subsection.
The difficulty in closing the gap between our lower and upper bounds for the case $\lambda_H < \lambda_E$ comes from the dependency of the leftward rate on the current number of $E$ agents (i.e., the second term in \eqref{eq:transtionE:2}).
Our bounds on the right-tail of the distribution of number of $E$ agent are not tight enough to result in a matching lower and upper bounds.
Closing this gap remains an open question.
A notable difference is that in \eqref{eq:transtion:1} and \eqref{eq:transtion:2}, knowing that $E$ is bounded above by a constant (independent of $p_H$) is enough to get matching upper and lower bounds (up to a vanishing term).
This, however, is not the case in \eqref{eq:transtionE:2}. To prove the upper bound in the case $\lambda_H < \lambda_E$ we couple the Markov process underlying policy $\BE$ with another process in which an $E$ agent that cannot form a match upon arrival turns into an $H$ agent. See subsection \ref{app:prioE_coupling}.

\subsection{The \textit{{ChainMatch(d)}} policy}
\label{sec:proof_chains}

This section proves Theorem \ref{th:chain} and Proposition \ref{prop:chain:p1}. As we could establish only an upper bound for the average waiting time when $p_E<1$, we refer the reader to Appendix \ref{app:chains_heuristic} for a heuristic analysis that leads us to guess the constant that we can numerically verify to be the correct one (See Figure \ref{wtimeHeuristicChains} of Subsection \ref{sec:app:tightness}).


Instead of directly analyzing the \textit{{ChainMatch(d)}} policy under our setting, we consider a modified setting, in which  an $E$ agent that does not match immediately upon arrival is removed from the system.
We refer to this new setting under the policy \textit{{ChainMatch(d)}} by $\tCd$. Observe that $H^{\tCd}_t$ is a $1$-dimensional CTMC with the following  transition rates:

\begin{subequations}
\begin{align}
Q^{\tCd}(h,h+1) &= \lambda_H (1 - p_H)^{d} \label{eq:trans1},\\
Q^{\tCd}(h,h-i) &= \left( \lambda_H (1 - (1 - p_H)^{d}) + \lambda_E (1 - (1 - p_E)^{d}) \right) (1-p_H)^{h-i}   \prod_{j=0}^{i-1}(1-(1-p_H)^{h - j}), \nonumber\\
& \quad \quad \quad  \quad  \quad  \quad  \quad  \quad \quad \quad  \quad  \quad  \quad  \quad  \quad  \quad  \quad  \quad  \quad \quad \quad  \quad  \quad  \quad  \quad i \in \{1,2,\ldots, h\}.\label{eq:trans2} \\ \nonumber
\end{align}
\end{subequations}
The first expression, \eqref{eq:trans1}, corresponds to rate, at which an $H$ agent arrives, but cannot be matched by a bridge agent.  The second expression, \eqref{eq:trans2}, corresponds to the rate, at which  an agent arrives,  is matched by a bridge agent and forms a chain-segment of length $i$.\footnote{Observe that the case $i = 0$ is possible, and corresponds to an arriving agent that can receive from a bridge agent but cannot continue the chain further. In that case the CTMC does not transition and we consider the chain-segment to have length $1$.}  In Appendix \ref{app:existence_proofs}, we show that the above CTMC reaches steady-state.

We introduce some notation to simplify \eqref{eq:trans2}.
Set $\Lambda = \lambda_H (1 - (1 - p_H)^{d}) + \lambda_E(1 -  (1 - p_E)^{d})$, which is the rate at which a new chain-segment (possibly of length $1$) starts, regardless of the current state, and let $S_{h}$ be the random number of agents removed from the system, starting from state $h$.\footnote{Note that using the notation from Section \ref{sec:chain-segment}, $S_{h} + 1$ corresponds to the length of the chain-segment $L_k$ for the 1-dimensional Markov chain.} For any $i\leq h$ we can write

\begin{align}
\label{eq:def:Q}
Q^{\tCd}(h,h-i) = \Lambda \P[S_{h} = i] = \Lambda (1-p_H)^{h-i}   \prod_{j=0}^{i-1}(1-(1-p_H)^{h - j}).
\end{align}

Observe that we have

\begin{align}
\label{eq:def:Sh}
\P[S_{h} \geq k] = \prod_{j=0}^{k-1}(1-(1-p_H)^{h-j}) \\ \nonumber
\end{align}



The proof proceeds by showing that  $\E[H^{\tCd}] $ serves as an upper bound for $\E[H^{\Cd}]$ (Lemma \ref{cl:chains_coupling}) and then computing the limit of $\E[H^{\tCd}]$ (Proposition \ref{prop:upper_chains}). Before that, we make the following  crucial observation: the process of chain-segment formation under $\tCd$ exhibits a memoryless property.
That is, for any state $h$ and any $\tilde{i} \leq i  \leq h$:
\begin{align}
\label{eq:memoryless}
\P[S_{h} = i] &=   (1-p_H)^{h-i}   \prod_{j=0}^{i - 1}(1-(1-p_H)^{h - j})\\ \nonumber
&=  \prod_{j=0}^{\tilde{i} - 1} \left(1-(1-p_H)^{h - j} \right) \left[ (1-p_H)^{(h - \tilde{i}) -(i - \tilde{i})}   \prod_{j=0}^{(i - \tilde{i}) - 1} \left(1-(1-p_H)^{(h - \tilde{i}) - j}\right) \right]\\ \nonumber
&=  \P[S_{h} \geq  \tilde{i}] \P[S_{h - \tilde{i}} = i - \tilde{i}], \nonumber
\end{align}

In other words, the event of forming a chain-segment of length $i$ can be decomposed into two independent events: forming a chain-segment of length at least $\tilde{i}$ and then forming a chain-segment of length $i-\tilde{i}$ starting with $h- \tilde{i}$ agents in the market. 
This heavily relies on the fact that chain-segments proceed in a local search (one by one) fashion and the independence assumption. Indeed, the chain-segment formation in the original $2$-dimensional chain $\Cd$ has a similar property.

We now show that  $\E[H^{\tCd}]$ is an upper bound for  $\E[H^{\Cd}]$.

\begin{restatable}{lemma}{claimChainsCoupling}
\label{cl:chains_coupling}
The expected number of $H$ agents in steady-state under $\tCd$ satisfies:
$$ \E[H^{\Cd}] \leq \E[H^{\tCd}].$$
\end{restatable}

\begin{proof}
The proof is based on a coupling argument.
Consider two copies of the arrival process, one under the setting of $\Cd$ and one under $\tCd$. Let $[H^{\Cd}_k, E^{\Cd}_k]$ and ${H}^{\tCd}_k$ denote the  embedded discrete-time Markov chain resulting from observing the two dynamic systems at arrival epochs. 
We  prove a stronger result: at any step $k$, $H_k^{\Cd} \leq H_k^{\tCd}$. We prove this using the following coupling:

\begin{enumerate}
\item Upon arrival of an $H$ agent we flip a biased coin with probability $(1 - p_H)^d$. If the coin flip is head, the $H$ agent cannot start a chain-segment, and both $H_{k+1}^{\Cd}$ and $H_{k+1}^{\tCd}$ increment by one.
    If the coin flip is tail, the $H$ agent starts a chain-segment in both systems. Suppose that $[H^{\Cd}_k, E^{\Cd}_k] = [h,e]$ and ${H}^{\tCd}_k = \tilde{h}$, and let $[L^H_{[h,e]}, L^E_{[h,e]}]$ denote the random number of $H$ and $E$ agents in the chain-segment formed under $\Cd$ at state $[h,e]$; similarly let $S_{\tilde{h}}$ be the length of chain-segment formed under $\tCd$ at state ($\tilde{h}$).
    We dinstinguish between three cases:
\begin{enumerate}[(a)]
\item $\tilde{h} \geq h$ and the event $\{S_{\tilde{h}} <  (\tilde{h}-h)\}$ occurs: we let $[L^H_{[h,e]}, L^E_{[h,e]}]$ be realized independently of $S_{\tilde{h}}$. \label{case:1}

\item $\tilde{h} \geq h$ and the event $\{S_{\tilde{h}} \geq  (\tilde{h}-h)\}$ occurs. In this case the memoryless property of $\tCd$ in \eqref{eq:memoryless} can be rewritten as: $\P[S_{\tilde{h}} = i \mid S_{\tilde{h}} \geq  (\tilde{h}-h)] = \P[S_{h} = i - (\tilde{h}-h) ]$. This divides the chain-segment formation into two independent events: a subchain-segment of length $ (\tilde{h}-h)$ is formed, and then a subchain-segment of length $\xi$, where $\xi$ is a random variable drawn from the distribution of $S_h$.

	Now we focus on the chain-segment formation under $\Cd$. Because $H$ agents get a higher  priority, the chain-segment can be computed in steps. Starting with $[h,e]$ agents, we first look for a subchain-segment $L_1^H$ consisting of only $H$ agents. When this chain-segment cannot be continued further with only $H$ agents, we look for an $E$ agent to continue the chain. If this happens (with probability $(1- (1 - p_E)^e$), we look for a second subchain-segment $L_2^H$ of only $H$ agents, etc.

	Note that the first subchain-segment $L_1^H$ also has the same distribution as $S_h$. We can therefore set $L_1^H = \xi$.
	All further subchain-segments $L_i^H$ are realised independently.\label{case:2}

\item $\tilde{h} < h$: we let $[L^H_{[h,e]}, L^E_{[h,e]}]$ and $S_{\tilde{h}}$ be realized independently.\footnote{This case is only defined here for the sake of completeness, the induction will ensure that this never happens.} \label{case:3}
\end{enumerate}
\item Upon arrival of an $E$ agent we flip a biased coin with probability $(1 - p_E)^d$. If the flip is head, the $E$ agent cannot start a chain-segment in either system and we have: $H_{k+1}^{\Cd} = H_k^{\Cd} $, and $H_{k+1}^{\tCd} = H_k^{\tCd} $; on the other hand, if the flip is tail, the $E$ agent starts a chain-segment in both systems. The chain-segment formation in this case is exactly the same as the one for an $H$ arrival.
\end{enumerate}
Having the above coupling, we finish the proof by induction: The base case $k = 0$ is trivial: $H_0^{\Cd} = H_0^{\tCd} = 0$. Suppose $H_k^{\Cd} \leq H_k^{\tCd}$ holds for $k$, we show that it also holds for $k+1$: if an $H$/$E$ arrival does not start a chain-segment then by coupling construction $H_{k+1}^{\Cd} \leq H_{k+1}^{\tCd}$. If an $H$ arrival does start a chain-segment then we are either in Case \eqref{case:1} or \eqref{case:2}. In the former the length of the chain-segment in $\tCd$ was not even long enough to bring the number of $H$ agents back to $H_k^{\Cd} = h$; therefore $H_{k+1}^{\Cd} \leq H_{k+1}^{\tCd}$ holds. In the latter case, again by coupling construction $L^H_{[h,e]} \geq S_{{h}} + (\tilde{h} - h)$, which implies that $H_{k+1}^{\Cd} \leq H_{k+1}^{\tCd}$ holds. A similar argument holds if an $E$ arrival starts a chain-segment.

\end{proof}

%
%

The next proposition computes $ \E[H^\tCd]$ in the limit.
Together with Lemma \ref{cl:chains_coupling}, this completes the proof of  Theorem \ref{th:chain}.

\begin{proposition}
\label{prop:upper_chains}
Under $\tCd$ and in steady-state, the expected number of $H$ agents satisfies:
$$ \lim_{p_H \rightarrow 0} p_H \E[H^\tCd]  = \ln \left(1 + \frac{\lambda_H}{\lambda_E (1 - (1 - p_E)^d)} \right).$$
\end{proposition}

\begin{proof}
Let $\pi $ be the steady-state probability distribution. By the conservation of flow from state $h$ to $h+1$, we have:
$$ \pi(h) \lambda_H (1 - p_H)^d = \sum_{k \geq 1} \pi(h+ k) \left( \sum_{i \leq h} Q^\tCd(h + k, i) \right).$$
Note that $\sum_{i \leq h} Q^\tCd(h + k, i)$ is the total leftward flow starting from state $h + k$ and ending at state $i \leq h$.
Using \eqref{eq:def:Q} and \eqref{eq:def:Sh}, we have:
$$
\sum_{i \leq h} Q^\tCd(h + k, i) = \Lambda \P[S_{h+k} \geq k],
$$
and therefore,
\begin{align}
\label{eq:pr:1}
\pi(h) \lambda_H (1 - p_H)^d =  \Lambda \sum_{k \geq 1} \pi(h+ k) \P[S_{h+k} \geq k].
\end{align}
  Observe that applying Definition \eqref{eq:def:Sh}, we have $\P[S_{h+k} \geq k] = \P[S_{h+k} \geq k-1] \P[S_{h+1} \geq 1]$.
Therefore we can rewrite \eqref{eq:pr:1} as:
\begin{align}
\label{eq:pr:2}
\pi(h) \lambda_H (1 - p_H)^d =  \Lambda \left( \pi(h+ 1) \P[S_{h+1} \geq 1] + \P[S_{h+1} \geq 1] \sum_{k \geq 2} \pi(h+ k) \P[S_{h+k} \geq k-1] \right).
\end{align}

\noindent{Similarly we write the conservation of flow from state $h+1$ to $h+2$:}
\begin{align}
\label{eq:pr:3}
\pi(h+1) \lambda_H (1 - p_H)^d =  \Lambda \sum_{k \geq 1} \pi(h+1+k) \P[S_{h+k+1} \geq k]  =  \Lambda \sum_{k' \geq 2} \pi(h+k') \P[S_{h+k'} \geq k'-1],
\end{align}
where the last step follows from a change of variable $k'= k+1$. Note that the summation in the RHS of \eqref{eq:pr:3} also appears in the second term of RHS of \eqref{eq:pr:2}.
Substituting $\sum_{k' \geq 2} \pi(h+k') \P[S_{h+k'} \geq k'-1]$ with $\pi(h+1) \lambda_H (1 - p_H)^d/\Lambda$ in \eqref{eq:pr:2} gives that

\begin{align}
\label{eq:pr:4}
\pi(h) \lambda_H (1 - p_H)^d =  \pi(h+ 1) \P[S_{h+1} \geq 1] \left(\Lambda + \lambda_H (1 - p_H)^d\right).
\end{align}

We can now compute $\E[H^\tCd]$ by proving an upper and lower bound separately. We  use the fact that for states far enough from the expectation, the distribution decays geometrically.
We start with the upper-bound.
Let   $\eta = \ln \left(1 + \frac{\lambda_H}{\lambda_E (1 - (1 - p_E)^d)} \right)/p_H  + 1/\sqrt{p_H}$.
We know from \eqref{eq:def:Sh} that $\P[S_{h+1} \geq 1] = 1 - (1 - p_H)^{h+1}$. This implies that for $h \geq \eta$,
\begin{align*}
\P[S_{h+1} \geq 1] &\geq 1 - (1 - p_H)^{\eta + 1} \\
&= 1- e^{(\eta + 1) \ln (1-p_H)} \\
& =  1 -   \frac{\lambda_E (1 - (1 - p_E)^d) }{\lambda_H+ \lambda_E (1 - (1 - p_E)^d) } (1 - \sqrt{p_H} ) + o(p_H),
\end{align*}
where we used the Taylor expansion $\ln(1- x) = - x - x^2/2 - x^3/3 - \ldots$.

Using  \eqref{eq:pr:4}  for $h \geq \eta$, we have:

\begin{equation}
 \begin{split}
\frac{ \pi(h+1)}{\pi(h)} &=   \frac{\lambda_H (1 - p_H)^d}{\P[S_{h+1} \geq 1] \left(  \Lambda  +  \lambda_H (1 - p_H)^d \right)}\\
&\leq  \frac{\lambda_H (1 - p_H)^d}{ \lambda_H +  \lambda_E (1 - (1 - p_E)^d)\sqrt{p_H}} + o(\sqrt{p_H})\\
&= 1 - c \sqrt{p_H} + o(\sqrt{p_H})
=: \delta,
\end{split}
\label{eq:pr:5}
\end{equation}
where $c = \frac{\lambda_E (1 - (1 - p_E)^d)  }{\lambda_H }$. Having \eqref{eq:pr:5}, we upper bound $\E[H^{\tCd}]$ as follows:

\begin{equation*}
\begin{split}
\E[H^{\tCd}] &= \sum_{h \leq \eta +  p_H^{-3/4}} \pi(h) + \sum_{h \geq \eta +  p_H^{-3/4}+ 1} \pi(h) \\
& \leq \eta + p_H^{-3/4} + \pi(\eta) \frac{ \delta^{ p_H^{-3/4}+1}}{(1 - \delta)} \\
& = \frac{\ln \left(1 + \frac{\lambda_H}{\lambda_E (1 - (1 - p_E)^d)} \right)}{p_H }+ o(1/p_H).
\end{split}
\end{equation*}

\noindent{Similarly we lower bound $\E[H^{\tCd}] $: let $\hat \eta = \ln \left(1 + \frac{\lambda_H}{\lambda_E (1 - (1 - p_E)^d)} \right)/p_H  - 1/\sqrt{p_H}$, we can find $\hat c$ such that for $h \leq \hat \eta$}:
$$\frac{ \pi(h)}{\pi(h+1)} \leq  1 - \hat c \sqrt{p_H} + o(\sqrt{p_H})$$
The above inequality combined with Markov inequality enables us to lower bound $\E[H^{\tCd}] $ as follows:
$$
\E[H^{\tCd}] \geq (\eta - p_H^{-3/4}) \left( 1 - \sum_{h = 0}^{\eta - p_H^{-3/4}} \pi(h) \right) =  \frac{\ln \left(1 + \frac{\lambda_H}{\lambda_E (1 - (1 - p_E)^d)} \right)}{p_H } + o(1/p_H).
$$
\end{proof}

Finally, note in the special case $p_E = 1$, an arriving $E$ agent is matched immediately by a bridge agent, implying that $E^{\Cd}_t = 0$ and $H^{\Cd}_t  = H^{\tCd}_t$; consequently  Proposition \ref{prop:upper_chains} implies the limit stated in  Proposition \ref{prop:chain:p1}.

\section{Final comments}
\label{sec:conclusion}

In matching markets where monetary transfers are not allowed, exogenous thickness increases exchange opportunities (\cite{RothHahn}).
Using a simple dynamic model with heterogeneous agents we find a tight connection between market thickness and the desired matching technology;  matching through chains is significantly more efficient than (simple) bilateral matching only when the market is sufficiently thin. Furthermore, increasing the arrival rate of hard-to-match agents may have, under bilateral matching, an adverse effect on such agents who will face a harsher competition for matching with easy-to-match agents.

An important dynamic matching market is kidney exchange, which enables  incompatible patient-donor pairs to exchange donors. While our stylized model abstracts away from many details in this market, our findings may provide some useful insights to  policy issues. When merging  markets, which is an ongoing effort in various countries (see Section \ref{sec:kidney}), or attracting different types of pairs, there may be  negative effects on some pairs. This effect is well known for  pairs with O patients and non-O donors who compete to match with scarce O donors in the pool (\cite{RothKidneyAER});  our findings suggest that this negative effect extends also to blood-type compatible pairs (like O-O), many of which have   very highly sensitized  patients.
Understanding these externalities are a key step  towards aligning incentives towards cooperation between the relevant  players (\cite{AshlagiRothAER}). 
Our findings further provide some insights about  tradeoffs from prioritizing different types of pairs.


Next we discuss some limitations and possible extensions. One  interesting challenge is to  quantify the exact loss from restricting attention to myopic policies that do not wait before matching, rather than finding the optimal Markovian policy that may make some agents wait in order to increase matching opportunities.
\footnote{Similar to \cite{AndersonDynamic} we can show our policies achieve the same scaling as the best anonymous Markovian policy (see Proposition \ref{prop:lowBound} in Appendix \ref{sec:anyAlg})
but characterising the best constants is an open question.}
Another interesting direction is to extend the model to allow departures.\footnote{For example, \cite{Akbarpour} allow agents to depart prior to being matched and consider the match rate as the measure for efficiency.}
Finally, our focus has been on marketplaces, in which  any pair of agents have a non-zero probability of forming a match. We found that the composition of the market crucially impacts the efficiency of the market. An interesting direction for future research would be to extend this study to two-sided marketplaces, in particular explore what features determine waiting times; for example, whether it is more beneficial to be on the short side or have a large ex ante   match probability. 

\bibliographystyle{chicago}
\bibliography{kidneysbib}


\appendix
\section{Proof of Lemma \ref{lem:upper_bound}}
\label{app:concentration_upper}

\begin{proof}[Proof of Lemma \ref{lem:upper_bound}]
Let $\pi(x,y)$ be the state  joint distribution of $[X,Y]$, and let $\pi_X(x)= \sum_{y \geq 0} \pi(x,y)$ be the marginal distribution of $X$. For a given $x \geq 0$, we have:
\begin{equation*}
\begin{split}
\sum_{y \in S(x)} \pi(x,y) &Q( [x, y], [x+1, y] ) + \sum_{y \not \in S(x)} \pi(x,y) Q( [x, y], [x+1, y] ) \\
&= \sum_{y \in S(x+1)} \pi(x+ 1,y) Q( [x+1, y], [x, y] ) + \sum_{y \not \in S(x+1)} \pi(x + 1,y) Q( [x+1, y], [x, y] ).
\end{split}
\end{equation*}
Using Condition \ref{cond:2_2}, we upper-bound the LHS and lower-bound the RHS, which results in having:
$$f(x) \P[X = x, Y \in S(x)] + \P[X = x, Y \not \in S(x)] \geq g(x+1) \P[X = x + 1, Y \in S(x+1)].$$
Let $\pi_S(x) = \P[X = x, Y \in S(x)] = \sum_{y \in S(x)} \pi(x,y)$. Observe that $\pi_{X}(x) \leq \pi_S(x) + c \delta^x \leq \pi_S(x) + c \rho^x$ from Condition \ref{cond:2_1} and for any $\rho \in [\delta, 1)$. Assuming that for $x \geq \eta$, $\frac{f(x)}{g(x + 1)} \leq \rho$ , we get:
$$ \pi_S(x+1) \leq \frac{f(x)}{g(x+1)} \pi_S(x) + \frac{\P[Y \not \in S(x)]}{g(x+1)} \leq  \rho \pi_S(x)  + \frac{c \delta^x}{g(x+1)} \leq  \rho \pi_S(x)  + \frac{c \rho^x}{g(\eta+1)},$$
where the last inequality results from the assumption $\frac{\delta^x}{g(x+1)} \leq \frac{\rho^x}{g(\eta +1)}$.
We can now prove by induction that:
$$ \pi_S(\eta+i) \leq \rho^i \left( \pi_S(\eta) + i\frac{c  \rho^{\eta -1} }{g(\eta+1)}\right)$$
This allows us to conclude:
\begin{equation*}
\begin{split}
 \P[X \geq \eta + k]  &= \sum_{i = \eta + k}^\infty \pi_X(i)  \\
 &\leq \sum_{i =  k}^\infty \pi_S(\eta + i) + \sum_{i = k}^\infty c \rho^{\eta + i} \\
  &\leq \sum_{i =  k}^\infty \rho^{i} \pi_S(\eta) + \frac{c  \rho^{\eta -1} }{g(\eta+1)}\sum_{i =  k}^\infty i \rho^{i} +  \sum_{i = k}^\infty c \rho^{\eta + i} \\
 &\leq \frac{\rho^k}{1 - \rho} \left(1 + c \rho^{\eta} + \frac{c  \rho^{\eta -1} }{g(\eta+1)}  \frac{k+1}{1 - \rho}\right)\\
 &\leq\frac{\rho^k}{1 - \rho} \left(1 + c +  \frac{ c(k+1)}{g(\eta + 1) - f(\eta)} \right)
 \end{split}
\end{equation*}

\end{proof}

\section{Missing proofs for \textit{ BilateralMatch(H)} Policy}
\label{app:bilat_prioH}

In the proofs of this section and of the next ones we will use the following facts: for any bounded, non-negative function $\xi : \R_+ \mapsto \R_+$ and any constant $u >0$, in the limit where $p_H \rightarrow 0$,

\begin{fact}
\label{fact:1}
For $\eta = \frac{(\ln u) + \xi(p_H)}{p_H^2}$, we have $ (1 - p_H^2)^\eta = \frac{e^{-\xi(p_H)}}{u} + O(p_H^2) $.
\end{fact}

\begin{fact}
\label{fact:2}
For any constants $p_E$, $r>0$ and for $\eta = \frac{(\ln u)+ \xi(p)}{p_H^2}$, we have $ (1 - p_E p_H)^\eta = o(p_H^r)$.
\end{fact}

\begin{fact}
\label{fact:3}
For any constant $p_E$, and for $\eta = \frac{(\ln u)+ \xi(p_H)}{p_E p_H}$, we have $ (1 - p_H^2)^\eta = 1 - O(p_H)$.
\end{fact}

\begin{fact}
\label{fact:4}
For any constant $p_E$ and for $\eta = \frac{(\ln u)+ \xi(p)}{p_E p_H}$, we have $ (1 - p_E p_H)^\eta = \frac{e^{-\xi(p)}}{u} + O(p_H)$.
\end{fact}

\subsection{Proof of Theorem \ref{th:bilat_prioH}}
\label{subsec:proof:thm1}

\begin{proof}[Proof of Theorem \ref{th:bilat_prioH}]
We first upper-bound  $\E[H^\BH] $. Let
\begin{equation*}
v(\lambda_H, \lambda_E, p_E, p_H) = \Bigg\{
\label{eq:2d-walk}
\begin{matrix}
&\frac{1}{p_E p_H} \left( \ln\left(\frac{\lambda_E}{\lambda_E - \lambda_H}\right) + c_3 \sqrt{p_H} \right) + p_H^{-3/4} &\text{ when } \lambda_H < \lambda_E,\\
& \frac{1}{p_H^2} \left( \ln\left(\frac{2 \lambda_H}{\lambda_E + \lambda_H}\right) + c_4 \sqrt{p_H} \right) + p_H^{-3/4} &\text{ when } \lambda_H > \lambda_E.
\end{matrix}
\end{equation*}
Where $c_3$ and $c_4$ are the constants from Proposition \ref{prop:upper_H_prio_H}.
Using the equality  $\E[H^{\BH}] = \sum_{i = 1}^{\infty} \P[H^{\BH} \geq i]$, we have:
\begin{equation*}
\begin{split}
\E[H^{\BH}] &= \sum_{i = 1}^{v(\lambda_H, \lambda_E, p_E, p_H) - 1} \P[H^{\BH} \geq i] +  \sum_{i = v(\lambda_H, \lambda_E, p_E, p_H)}^{\infty} \P[H^{\BH} \geq i]\\
&\leq v(\lambda_H, \lambda_E, p_E, p_H) + \sum_{j = p_H^{-3/4}}^{\infty} \frac{\gamma(p_H)^j}{1- \gamma(p_H)} \\
& \leq  v(\lambda_H, \lambda_E, p_E, p_H) + \frac{\gamma''(p_H)^{p_H^{-3/4}}}{(1- \gamma''(p_H))^2}.
\end{split}
\end{equation*}

Where we denote $\gamma'' = \max(\gamma, \gamma')$ and we used the result from Proposition \ref{prop:upper_H_prio_H} : $\P[H^\BH \geq v(\lambda_H, \lambda_E, p_E, p_H) + j] \leq \frac{\gamma''(p_H)^j}{1- \gamma''(p_H)}$.

Applying the fact that $\gamma''(p_H)^{p_H^{-3/4}} = \left(1 - \sqrt{p_H} + o(\sqrt{p_H}) \right)^{p_H^{-3/4}} = o(p_H^2)$, and some algebra we get the following upper-bound on $\E[H^\BH]$:
\begin{itemize}
\item[-] If $\lambda_H < \lambda_E$, then $\E[H^\BH] \leq \frac{\ln\left(\frac{\lambda_E}{\lambda_E - \lambda_H}\right)}{p_E p_H} + o\left( \frac{1}{p_H} \right)$.
\item[-] If $\lambda_H > \lambda_E$, then $\E[H^\BH] \leq \frac{\ln\left(\frac{2 \lambda_H}{\lambda_H + \lambda_E} \right)}{p_H^2} + o \left( \frac{1}{p_H^2} \right)$.
\end{itemize}
Now we proceed to lower-bound  $\E[H^\BH] $: Applying Markov inequality to $\E[H^\BH] $ and  using Proposition \ref{prop:lower_H_prio_H}, we get the following lower-bound on $\E[H^\BH]$:

\begin{itemize}
\item[-] If $\lambda_H < \lambda_E$, then $\E[H^\BH] \geq \frac{\ln\left(\frac{\lambda_E}{\lambda_E - \lambda_H}\right)}{p_E p_H} + o\left( \frac{1}{p_H} \right)$.
\item[-] If $\lambda_H > \lambda_E$, then $\E[H^\BH] \geq \frac{\ln\left(\frac{2 \lambda_H}{\lambda_H + \lambda_E} \right)}{p_H^2} + o \left( \frac{1}{p_H^2} \right)$.
\end{itemize}
This completes the proof.
\end{proof}

\subsection{Proof of Corollaries \ref{cor:monotone} and \ref{cor:nonmonotone}}
\label{subsec:corrs}
\begin{proof}[Proof of Corollary \ref{cor:monotone}]
Define $x = \lambda_H/\lambda_E$, and $f(x) = \frac{\ln \left(\frac{1}{1-x}\right)}{x}$.
Note that the constant $\frac{\ln\left(\frac{\lambda_E}{\lambda_E - \lambda_H}\right)}{p_E  \lambda_H } = \frac{f(x)}{\lambda_E p_E}$.
It is easy check that $f'(x) >0$ in $x \in (0,1)$, and therefore $f(x)$ is increasing in $x \in (0,1)$.
\end{proof}

\begin{proof}[Proof of Corollary \ref{cor:nonmonotone}]
Define $y = \lambda_H/\lambda_E$, and $g(y) = \frac{\ln \left(\frac{2y}{1+y}\right)}{y}$.
Note that the constant $\frac{\ln \left( \frac{2 \lambda_H}{\lambda_H + \lambda_E} \right)}{\lambda_H } = \frac{g(y)}{\lambda_E}$.
It is easy to check that $g'(y) > 0$ when $y \in (1,y^*)$, and $g'(y) < 0$ in $y> y^*$ where $y^*$ is the solution of $g'(y) = 0$:
$$
g'(y^*) = 0 \Leftrightarrow \frac{1}{y^*+1} = \ln \left(\frac{2y^*}{1+y^*}\right) \Leftrightarrow (y^*+1) \ln(2-2/(y^*+1)) = 1,
$$
\end{proof}

\subsection{Proof of Proposition \ref{prop:upper_H_prio_H}}
\label{subsec:proof:pro1}
Instead of proving Proposition \ref{prop:upper_H_prio_H}, we prove a stronger result. This will be useful later on to prove an upper bound for $E$ agents (see Lemma \ref{lem:upper_E} in \ref{subsec:proof:pro2}).


\begin{restatable}{proposition}{propUpperHprioH}
\label{prop:upper_H_prio_H_general}
Under $\BH$, if $\lambda_H < \lambda_E$, for any non-negative bounded function $\xi(p_H)$, for all $k \geq 0$:
$$ \P\left[H^\BH \geq \frac{1}{p_E p_H} \left( \ln\left(\frac{\lambda_E}{\lambda_E - \lambda_H}\right) + \xi(p_H) \right) + k \right]  \leq \frac{\gamma(p_H)^k}{1 - \gamma(p_H)},$$
where $\gamma(p_H) =  \frac{\lambda_H}{\lambda_E - (\lambda_E - \lambda_H)e^{- \xi(p_H)}} + O(p_H)$.
\par
If $\lambda_H > \lambda_E$, for any non-negative bounded function $\xi'(p_H)$, for all $k \geq 0$:
$$ \P\left[H^\BH \geq \frac{1}{p_H^2} \left( \ln\left(\frac{2 \lambda_H}{\lambda_E + \lambda_H}\right) + \xi'(p_H) \right) +k \right]  \leq \frac{\gamma'(p_H)^k}{1- \gamma'(p_H)},$$
where $\gamma'(p_H) = \frac{ e^{- \xi'(p_H)} }{ 2 -  e^{- \xi'(p_H)}} +O(p_H^2) $.
\end{restatable}


\begin{proof}
We wish to apply Lemma \ref{lem:upper_bound} with $[X(t), Y(t)] = [H^\BH(t), E^\BH(t)]$, in the special case where $S(h) = \N$ for all $h$, and $c = \delta = 0$.
This implies that we need to find functions $f(\cdot)$ and $g(\cdot)$ such that for all $e \in \N$, $f(h) \geq Q^\BH([h,e],[h+1,e])$, and $g(h) \leq Q^\BH([h,e],[h-1,e])$.

Let $f(h) = \lambda_H (1 - p_H^2)^h$ and  $g(h) = \lambda_H (1 - (1 - p_H^2)^h) + \lambda_E (1 - (1 - p_E p_H)^h)$. We have
$$ \frac{f(h)}{g(h+1)} = \frac{\lambda_H (1 - p_H^2)^h}{ \lambda_H (1 - (1 - p_H^2)^{h+1}) + \lambda_E (1 - (1 - p_E p_H)^{h+1})}.$$

 \noindent {\bf Case $\lambda_H <\lambda_E$:}
Take $\eta =  \frac{1}{p_E p_H} \left( \ln\left(\frac{\lambda_E}{\lambda_E - \lambda_H}\right) + \xi(p_H) \right)$.
Facts \ref{fact:3} and \ref{fact:4} imply respectively that  $(1 - p_H^2)^\eta = 1 + O(p_H)$ and $(1 - p_Ep_H)^\eta = \frac{\lambda_E - \lambda_H}{\lambda_E} e^{- \xi(p_H)} +O(p_H)$. This yields for all $h \geq \eta$:
$$\frac{f(h)}{g(h+1)} \leq   \frac{\lambda_H}{\lambda_E - (\lambda_E - \lambda_H)e^{- \xi(p_H)}} + O(p_H) := \gamma(p_H).$$

\noindent {\bf Case $\lambda_H > \lambda_E$:}
Taking $\eta =  \frac{1}{p_H^2} \left( \ln\left(\frac{2 \lambda_H}{\lambda_E + \lambda_H}\right) + \xi'(p_H) \right)$.
Facts \ref{fact:1} and \ref{fact:2} imply respectively $(1 - p_H^2)^\eta = \frac{\lambda_E + \lambda_H}{2\lambda_H} e^{- \xi'(p_H)} +O(p_H^2)$ and $(1 - p_Ep_H)^\eta = o(p_H^2)$.
This yields for all $h \geq \eta$:
$$ \frac{f(h)}{g(h+1)} \leq  \frac{\frac{1}{2} e^{- \xi'(p_H)} }{ 1 - \frac{1}{2 } e^{- \xi'(p_H)}} +O(p_H^2) := \gamma'(p_H).$$

In both cases, we conclude by applying Lemma \ref{lem:upper_bound}  with $\rho = \gamma(p_H)$ or $\rho = \gamma'(p_H)$ .

\end{proof}

\begin{proof}[Proof of Proposition \ref{prop:upper_H_prio_H}]
The proof of Proposition \ref{prop:upper_H_prio_H} is a consequence of Proposition \ref{prop:upper_H_prio_H_general}:
\begin{itemize}
  \item[-] In the case where $\lambda_H < \lambda_E$, take $\xi(p_H) = \sqrt{p_H}\frac{\lambda_H}{\lambda_E - \lambda_H}$. This implies that $\gamma(p_H) =  \frac{\lambda_H}{\lambda_E - (\lambda_E - \lambda_H)e^{- \xi(p_H)}} + O(p_H) = 1 - \sqrt{p_H} + O(\sqrt{p_H})$.
  \item[-] In the case where $\lambda_H > \lambda_E$, take $\xi'(p_H) = 2\sqrt{p_H}$. This implies that $\gamma'(p_H) = \frac{ e^{- \xi'(p_H)} }{ 2 -  e^{- \xi'(p_H)}} +O(p_H^2) = 1 - \sqrt{p_H} + O(\sqrt{p_H})$.
\end{itemize}
\end{proof}

\subsection{Proof of Proposition \ref{prop:lower_H_prio_H}}
\label{subsec:proof:pro2}
The proof of Proposition \ref{prop:lower_H_prio_H} requires a concentration bound on the number of $E$ agents, which we state in Lemma \ref{lem:upper_E}.

\begin{lemma}
\label{lem:upper_E}
Under $\BH$, and assuming $p_H \leq p_E$, there exist constants $C$ and $\zeta < 1$ (which only depend on $\lambda_H, \lambda_E,$ and $p_E$) such that for any $k \geq 0$, there exists $p$ such that for any $p_H < p$, we have:
$$\P \left[E^\BH \geq \frac{1}{\sqrt{p_H}} + k \right] \leq C \zeta^k.$$
\end{lemma}

\begin{proof}
The proof is based on a bound on the right-tail distribution of $E$ agents in the market. To do this, we will apply Lemma \ref{lem:upper_bound} with $[X(t), Y(t)] = [E^\BH(t), H^\BH(t)]$. Therefore, we find an upper-bound $f(e) = \lambda_E (1 - p_E^2)^e \geq Q^\BH([h,e],[h,e+1])$ on the upward transition \eqref{eq:transtion:3}. Similarly, we would like to find a lower-bound $g(e)$ on the downward transition \eqref{eq:transtion:4}, but we cannot find one for any $h \in \N$. Therefore we need to restrict our attention to some subset $S(e) \subset \N$.
Recall that $ Q^\BH([h,e],[h,e-1]) = \lambda_H (1 - p_H^2)^{h} (1 - (1 - p_E p_H)^{e}) + \lambda_E (1 - p_E p_H)^{h} (1 - (1 - p_E^2)^{e}))$.

\subsubsection*{Case $\lambda_H < \lambda_E$:}
$$S(e) = \left\{ h \in  \N \mid h \leq \frac{1}{p_E p_H} \left( \ln\left(\frac{\lambda_E}{\lambda_E - \lambda_H}\right) + \ln 2 \right) + e \right\}.$$

Applying Proposition \ref{prop:upper_H_prio_H_general} with $\xi(p_H) = \ln 2$, we get that $\P[H^\BH \not \in S(e)] \leq \frac{\gamma(p_H)^e}{1 - \gamma(p_H)} =c \delta^e$, where $c = \frac{\lambda_H + \lambda_E}{\lambda_E - \lambda_H}$, $\delta = \frac{2\lambda_H}{\lambda_E + \lambda_H}$.
Fact \ref{fact:4} implies that for all $h \in S(e)$, 
$(1 - p_E p_H)^h \geq \frac{\lambda_E - \lambda_H}{2 \lambda_E} (1 - p_E p_H)^e + O(p_H)$. Keeping only the second term in $Q^\BH([h,e], [h, e-1])$, we get the lower-bound:

\begin{equation*}
\label{eq:lower_bound}
Q^\BH([h,e], [h, e-1]) \geq  \frac{\lambda_E - \lambda_H}{2}(1 - p_E p_H)^e(1 - (1 - p_E^2)^e) + O(p_H) := g(e)
\end{equation*}
This yields:
$$\frac{f(e)}{g(e+1)}= \frac{2 \lambda_E }{(\lambda_E - \lambda_H)(1 - (1 - p_E^2)^{e+1})} \left( \frac{1 - p_E^2}{1 - p_E p_H}\right)^{e+1} + O(p_H).$$
We get for all $e \geq \frac{1}{\sqrt{p_H}}$, and $p_H$ small enough, $\frac{f(e)}{g(e+1)} \leq \delta$. Furthermore, for $\rho = \frac{1 + \delta}{2}$ and $p_H$ small enough,
$$ \frac{\delta^e}{g(e+1)} \leq \delta^e \frac{(1 - p_E p_H)^{ \frac{1}{\sqrt{p_H}}- e-1}}{g( \frac{1}{\sqrt{p_H}} +1)} \leq \left( \frac{\delta}{1 - p_E p_H}\right)^e \frac{1}{g( \frac{1}{\sqrt{p_H}} +1)}\leq \frac{\rho^e}{g( \frac{1}{\sqrt{p_H}} + 1)}$$

We can apply Lemma \ref{lem:upper_bound} which yields the desired bound:
\begin{equation*}
\begin{split}
\P \left[E^\BH \geq \frac{1}{\sqrt{p_H}}+ k \right] &\leq \frac{\rho^k}{1 - \rho} \left( 1 + c + \frac{c (k+1)}{g(\frac{1}{\sqrt{p_H}} + 1) - f(\frac{1}{\sqrt{p_H}})} \right) \\
&\leq  \frac{\rho ^k}{1 - \rho} \left(1 + c + \frac{c(k+1) (\lambda_H + \lambda_E)}{2\lambda_H g(1 / \sqrt{p_H} + 1)} \right). \\
&\leq  \frac{\rho ^k}{1 - \rho} \left(1 + c +  \frac{c(k+1) (\lambda_H + \lambda_E)}{\lambda_H(\lambda_E - \lambda_H)} \right). \\
 \end{split}
\end{equation*}
Where we used first the fact that $\frac{f(1/\sqrt{p_H})}{g(1/\sqrt{p_H}+ 1)} \leq \delta =  \frac{2\lambda_H}{\lambda_E + \lambda_H}$ and therefore $g(1/\sqrt{p_H}+ 1) - f(1/\sqrt{p_H}) \geq \frac{\lambda_E - \lambda_H}{\lambda_E + \lambda_H}$ and second the fact that $g(1/\sqrt{p_H} + 1) = \frac{\lambda_E - \lambda_H}{2} + O(p_H)$.

\subsubsection*{Case $\lambda_H > \lambda_E$:}
$$ S(e) = \left\{ h \in  \N \mid h \leq \frac{1}{p_H^2} \left( \ln\left(\frac{2 \lambda_H}{\lambda_E + \lambda_H}\right) + \ln 2 \right) + e \right\}.$$
Applying Proposition \ref{prop:upper_H_prio_H_general} with $\xi(p_H) = \ln(2)$, we have $\P[H^\BE \not \in S(e)] \leq \frac{\gamma(p_H)^e}{1 - \gamma(p_H)} = c \delta^e$, with $c = 3/2$ and $\delta = 1/3$.
Fact \ref{fact:1} implies that for all $h \in S(e)$, $(1 - p_H^2)^h \geq \frac{\lambda_H + \lambda_E}{4 \lambda_H}(1 - p_H^2)^e+ O(p_H^2) $. Keeping only the first term in $Q^\BH([h,e], [h, e-1])$, we get the lower-bound:
$$Q^\BH([h,e],[h, e-1]) \geq  \frac{\lambda_H + \lambda_E}{4\lambda_H}(1 - p_H^2)^e(1 - (1 - p_E p_H)^e) + O(p_H^2) = g(e).$$
This yields:
$$\frac{f(e)}{g(e+1)} \leq \frac{4 \lambda_H \lambda_E}{(\lambda_H + \lambda_E)(1 - (1 - p_E p_H)^e)}\left( \frac{1 - p_E^2}{1 - p_H^2}\right)^e + O(p_H^2).$$


Furthermore, with $\rho = 1/2$ we get that for $p_H$ small enough, $\frac{\delta}{1 - p_H^2} \leq \rho$. This leads to
$$ \frac{\delta^e}{g(e+1)} \leq \delta^e \frac{(1 - p_H^2)^{p_H^{-1/2} - e}}{g\left(p_H^{-1/2} +1\right)} \leq \left(\frac{\delta}{1 - p_H^2}\right)^e \frac{1}{g\left(p_H^{-1/2} +1\right)} \leq \frac{\rho^e}{g\left(p_H^{-1/2} + 1\right)}.$$

Therefore we can apply Lemma \ref{lem:upper_bound}:
\begin{equation*}
\begin{split}
\P \left[E^\BH \geq n_E+ k \right] &\leq \frac{\rho^k}{1 - \rho}\left( 1 + c +   \frac{c (k+1)}{g\left(p_H^{-1/2} + 1\right) - f(p_H^{-1/2})} \right)\\
&\leq \frac{\rho^k}{1 - \rho}\left( 1 + c +   \frac{c (k+1)3}{2g(p_H^{-1/2})} \right)\\
&\leq \frac{\rho^k}{1 - \rho}\left( 1 + c + \frac{6c(k+1)}{\lambda_H + \lambda_E} \right) \\
 \end{split}
\end{equation*}
Where we used the fact that for all $e \geq \frac{1}{\sqrt{p_H}}$, $\frac{f(e)}{g(e+1)} = o(p_H)$. Therefore for $p_H$ small enough, $\frac{f(e)}{g(e+1)} \leq \frac{1}{3}$ and therefore $g(e + 1) - f(e) \geq \frac{2}{3} g(e + 1) $.
Which concludes the proof.

\end{proof}

We can now prove Proposition \ref{prop:lower_H_prio_H}. The main idea is to apply Lemma \ref{lem:lower_bound}, in two cases separately, where in one case we have few ($\leq \frac{2}{\sqrt{p_H}}$) $E$ agents, and in the other case where we have many. Using Lemma \ref{lem:upper_E}, we can exponentially bound the second case.



\begin{proof}[Proof of Proposition \ref{prop:lower_H_prio_H}.]
We will apply Lemma \ref{lem:lower_bound}.
Consider $S = \left\{e \in \mathbb{R}, e \leq  p_H^{-1/2} + p_H^{-1/2}  \right\}.$
Using Lemma \ref{lem:upper_E} with $k = p_H^{-1/2}$, we get that $\P[E^\BH \not \in S] = o(p_H^4)$.

For $e \in S$, we have $(1 - p_E p_H)^e \geq 1 - p_E p_H \left( p_H^{-1/2} + p_H^{-1/2} +o(\sqrt{p_H}) \right)= 1 - 2p_E \sqrt{p_H} +o(\sqrt{p_H})$.
Taking $f(h) = \lambda_H (1 - p_H^2)^h \left(1 - 2p_E \sqrt{p_H}\right) + o(\sqrt{p_H})$, we have for $e \in S$, $Q^\BH([h,e], [h+1,e]) \geq f(h)$. Let $g(h) = Q^\BH([h,e], [h-1,e]) = \lambda_H (1 - (1 - p_H^2)^h) + \lambda_E (1 - (1 - p_E p_H)^h)$, and note that $f(h)$ is non-increasing and $g(h)$ is non-decreasing. We have:
$$\frac{g(h+1)}{f(h)} = \frac{\lambda_H + \lambda_E - \lambda_H (1 - p_H^2)^{h+1} - \lambda_E (1 - p_E p_H)^{h+1}}{\lambda_H (1 - p_H^2)^h \left(1 - 2p_E \sqrt{p_H} + o(\sqrt{p_H})\right)} $$

\noindent{\bf In the case $\lambda_H < \lambda_E$},
Take $\eta =  \frac{1}{p_E p_H} \left( \ln\left(\frac{\lambda_E}{\lambda_E - \lambda_H}\right) - c'_1\sqrt{p_H} \right)$. Using Facts \ref{fact:3} and \ref{fact:4}, we have $(1 - p_H^2)^\eta = 1 + O(p_H)$ and $(1 - p_Ep_H)^\eta = \frac{ \lambda_E - \lambda_H}{\lambda_E} e^{c'_1 \sqrt{p_H}} + O(p_H)= \frac{\lambda_E - \lambda_H}{\lambda_E} (1  + c'_1\sqrt{p_H} + o(\sqrt{p_H}))$, therefore:
\begin{equation*}
\begin{split}
\frac{g(\eta+1)}{f(\eta)}  &= \frac{\lambda_H + \lambda_E - \lambda_H - (\lambda_E - \lambda_H) (1 +c'_1\sqrt{p_H}) + o(\sqrt{p_H}) }{\lambda_H \left(1 - 2p_E \sqrt{p_H}\right) + O(p_H)} \\
&= 1 - \left( \frac{\lambda_E - \lambda_H}{\lambda_H} c'_1 - 2 p_E \right) \sqrt{p_H} + o(\sqrt{p_H}).\\
&= 1 - \sqrt{p_H} + o(\sqrt{p_H}).
 \end{split}
 \end{equation*}
Where we fixed $c'_1 = \frac{\lambda_H (1 + 2p_E)}{\lambda_E - \lambda_H}$. Using Lemma \ref{lem:lower_bound} with $[X(t), Y(t)] = [H^\BH(t), E^\BH(t)]$, $k = p_H^{-3/4}$ and $\rho = \frac{g(\eta+1)}{f(\eta)}$, we get $\rho^k = o(p_H^2)$ and:
\begin{equation*}
\begin{split}
 \P[H^\BH \leq \eta - k ] &\leq \eta \cdot o(p_H^4) \left(1 + \frac{1 }{f(\eta) - g(\eta+1)}\right)  + \frac{\rho^k}{1 - \rho}  \\
  &\leq \eta  \cdot o(p_H^4) \left(1 + \frac{1}{f(\eta) \sqrt{p_H} + o(\sqrt{p_H})}\right)  + \frac{o(p_H^2)}{\sqrt{p_H} + o(\sqrt{p_H})} \\
  &=  o(p_H).
 \end{split}
 \end{equation*}

 {Taking $c_1 = c'_1 + p_E$, this enables us to conclude that $\P[H^\BH \leq \frac{1}{p_E p_H} \left( \ln\left(\frac{\lambda_E}{\lambda_E - \lambda_H}\right) - c_1p_H^{1/4} \right) ] \leq o(p_H)$.}

\noindent{\bf In the case $\lambda_H > \lambda_E$},
Let $\eta =  \frac{1}{p_H^2} \left(\ln\left(\frac{2 \lambda_H}{\lambda_H + \lambda_E}\right) - c'_2 p_H^{1/4}\right)$. Using Facts \ref{fact:1} and \ref{fact:2}, we have $(1 - p_H^2)^\eta = \frac{\lambda_H + \lambda_E}{2 \lambda_H} e^{c'_2 p_H^{1/4}} + O(p_H^2)$ and $(1 - p_Ep_H)^\eta = O(p_H^2)$. This implies that:
\begin{equation*}
  \begin{split}
\rho = \frac{g(\eta+1)}{f(\eta)}  &= \frac{ (\lambda_H + \lambda_E) - \frac{(\lambda_H + \lambda_E)}{2} e^{c'_2 \sqrt{p_H}} + O(p_H^2) }{\frac{(\lambda_H + \lambda_E)}{2}e^{c'_2 \sqrt{p_H}} \left(1 - 2p_E \sqrt{p_H}\right) + O(p_H^2)} \\
&= \frac{2 - (1 + c'_2 \sqrt{p_H})}{(1 + c'_2 \sqrt{p_H})(1 - 2 p_E \sqrt{p_H})} + o(\sqrt{p_H}) \\
&= 1 - (2 c'_2 - 2 p_E) \sqrt{p_H} + o(\sqrt{p_H})\\
&=1 - \sqrt{p_H} + o(\sqrt{p_H}).
 \end{split}
\end{equation*}
Where we have chosen $c'_2 = \frac{1 + {2}p_E}{2}$.
Taking $k = p_H^{-3/4}$, we have $ \rho^k = o(p_H^2)$. Applying Lemma \ref{lem:lower_bound}, we get that for all :
\begin{equation*}
\begin{split}
 \P[H^\BH \leq \eta - k ]  &\leq \eta \cdot o(p_H^4) \left( 1 + \frac{1} {f(\eta) - g(\eta+1)} \right) + \frac{\rho^k}{1 - \rho} \\
  &= \eta \cdot o(p_H^3)  + \frac{o(p_H^2)}{f(\eta) \sqrt{p_H} + o(\sqrt{p_H})} \\
  &= o(p_H)
 \end{split}
 \end{equation*}

 {Taking $c_2 = c'_2 + 1$, this enables us to conclude that $\P[H^\BH \leq \frac{1}{p_H^2} \left(\ln\left(\frac{2 \lambda_H}{\lambda_H + \lambda_E}\right) - c_2\sqrt{p_H} \right) ] \leq o(p_H)$.}

\end{proof}

\section{Missing proofs for \textit{ BilateralMatch(E)} Policy}
\label{sec:bilatE:proof}
The proof of Theorem \ref{th:bilat_prioE} relies on a concentration result on the left tail of the distribution of $H^\BE$ (Proposition \ref{prop:lower_H_prio_E}), and a coupling argument to upper-bound $\E[H^\BE]$ (Proposition \ref{prop:upper_H_prio_E}).

\begin{proposition}\emph{[Lower Bound]}
\label{prop:lower_H_prio_E}
Under $\BE$ in steady-state,
\begin{itemize}
\item[-] If $\lambda_H < \lambda_E$, there exist a function $\gamma(p_H) = 1 - \sqrt{p_H} + o(\sqrt{p_H})$ and a constant $c_5 \geq 0$ such that for all $k \geq 0$:
$$ \P\left[H^\BE \leq   \frac{1}{p_E p_H} \left( \ln\left(\frac{\lambda_E}{\lambda_E - \lambda_H}\right) - c_5 p_H^{1/4}\right) -k \right]  \leq \frac{\gamma(p_H)^k}{1- \gamma(p_H)}.$$
\item[-] If $\lambda_H > \lambda_E$, there exist a function $\gamma'(p_H) = 1 - \sqrt{p_H} + o(\sqrt{p_H})$  and a constant $c_6 \geq 0$ such that for all $k \geq 0$:
$$ \P\left[H^\BE \leq  \frac{1}{p_H^2} \left( \ln\left(\frac{2 \lambda_H}{\lambda_H + \lambda_E}\right)- c_6 \sqrt{p_H} \right) - k \right]  \leq \frac{\gamma'(p_H)^k}{1 - \gamma'(p_H)}.$$
\end{itemize}
\end{proposition}

\begin{restatable}{proposition}{propUpperHprioE}\emph{[Upper-bound]}
\label{prop:upper_H_prio_E}
Under $\BE$ and in steady-state,
\begin{itemize}
\item[-] If $\lambda_H < \lambda_E$, then $\E[H^\BE]   \leq \frac{\ln\left(\frac{2\lambda_E}{\lambda_E - \lambda_H}\right)}{p_E p_H} + o\left( \frac{1}{p_H}\right)$
\item[-] If $\lambda_H > \lambda_E$, then $\E[H^\BE]   \leq \frac{1}{p_H^2} \ln \left( \frac{2 \lambda_H}{\lambda_H + \lambda_E} \right) + o\left( \frac{1}{p_H^2}\right)$.
\end{itemize}
\end{restatable}

%


\begin{proof}[Proof of Theorem \ref{th:bilat_prioE}]
We will first compute a lower bound for $\E[H^\BE] $: Applying Markov inequality to $\E[H^\BE] $ and  using Proposition \ref{prop:lower_H_prio_E} with $k = p_H^{{-}3/4}$, we get the following lower-bound on $\E[H^\BE]$:

\begin{itemize}
\item[-] If $\lambda_H < \lambda_E$, then $\E[H^\BE] \geq \frac{\ln\left(\frac{\lambda_E}{\lambda_E - \lambda_H}\right)}{p_E p_H} + o\left( \frac{1}{p_H} \right)$.
\item[-] If $\lambda_H > \lambda_E$, then $\E[H^\BE] \geq \frac{\ln\left(\frac{2 \lambda_H}{\lambda_H + \lambda_E} \right)}{p_H^2} + o \left( \frac{1}{p_H^2} \right)$.
\end{itemize}

Using Proposition \ref{prop:upper_H_prio_E}, we can get the desired upper-bounds for  $\E[H^\BE] $. We can conclude using Little's law: $w_H = \frac{\E[H^\BE]}{\lambda_H}$.

\end{proof}

\subsection{Proof of Proposition \ref{prop:lower_H_prio_E}.}
In order to prove  Proposition \ref{prop:lower_H_prio_E}, we first need a concentration result on $E^\BE$ agents. This is stated in Lemma \ref{lem:upper_E_prioE}.

\begin{lemma}
\label{lem:upper_E_prioE}
Under $\BE$ in steady-state, for any $\alpha \in [0,1]$, let $n_E(\alpha) = \frac{\ln(\alpha/(1 + \alpha(1 - p_E^2))}{\ln(1 - p_E^2))}$. For any $k \geq 0$, we have:
$$\P \left[E^\BE \geq n_E(\alpha) + k \right] \leq  \frac{\alpha^k}{1 - \alpha}.$$
Also, $(1 - p_E^2)^{n_E(\alpha)} = \frac{\alpha}{1 + \alpha(1 - p_E^2)}$.
\end{lemma}

\begin{proof}
In the same way that upper-bounding $H$ agents when they get priority is easy because waiting $E$ agents can be ignored, upper-bounding the number of $E$ agents is easy when they get the priority because $H$ agents can be ignored. 
We can get the following bounds on the transition probabilities: $Q^\BE([h,e], [h,e+1]) \leq \lambda_E(1 - p_E^2)^{e} = f(e)$, and $Q^\BE([h,e], [h,e-1]) \geq  \lambda_E (1 - (1 - p_E^2)^{e}) = g(e)$ which leads to:
$$\frac{f(e)}{g(e+1)} = \frac{ (1 - p_E^2)^{e}}{1 - (1 - p_E^2)^{e+1}}.$$
Setting $\eta = n_E(\alpha)$ to be the solution to $\frac{f(n)}{g(n+1)} = \alpha$, and applying Lemma \ref{lem:upper_bound} with $[X(t),Y(t)] = [E(t), H(t)]$, and $S(e) = \N$ for all $h$, $c = 0$ and $\delta = 0$, we get the desired result.
\end{proof}

We can now prove Proposition \ref{prop:lower_H_prio_E}.

%
%
\begin{proof}[Proof of Proposition \ref{prop:lower_H_prio_E}]
Notice that $Q^\BE([h,e], [h+1, e]) = Q^\BH([h,e], [h+1, e])$, therefore the function $f(h) =  \lambda_H (1 - p_H^2)^h \left(1 - 2p_E \sqrt{p_H}\right) + o(\sqrt{p_H})$ used in the proof of Proposition \ref{prop:lower_H_prio_H} remains a lower bound for $e \in S = \left\{e \leq {2} p_H^{-1/2} \right\}$. Furthermore, $Q^\BE([h,e], [h-1, e]) \leq Q^\BH([h,e], [h-1, e])$ therefore the function $g(h) = \lambda_H (1 - (1 - p_H^2)^h) + \lambda_E (1 - (1 - p_E p_H)^h)$ remains an upper bound. Therefore, the proof is exactly the same as the proof for the lower bounds for Proposition \ref{prop:lower_H_prio_H}, except that instead of Lemma \ref{lem:upper_E} we use  Lemma \ref{lem:upper_E_prioE} in the special case  $n_E(\alpha) = p_H^{-1/2}$ with $S = \left\{e \in \mathbb{R}, e \leq 2 p_H^{-1/2} \right\}$, and $k = p_H^{-1/2}$. We still have, $\P[E^\BH \not \in S] = o(p_H^4)$.
\end{proof}

\subsection{Proof of Propostion \ref{prop:upper_H_prio_E}}
\label{app:prioE_coupling}

This proof is based on a coupling argument: instead of analysing the CTMC resulted from $\BE$, we analyse the CTMC underlying another process that we call $\widetilde{\BE}$.
$\widetilde{\BE}$ works similarly to $\BE$ with one crucial difference: each time an $E$ agent arrives which cannot be matched immediately, it changes type and becomes an $H$ agent and then joins the market. First note that under $\widetilde{\BE}$, no $E$ agent joins the market, making its underlying CTMC a $1$-dimensional Markov chain. Proving that the stochasitc process underlying $\widetilde{\BE}$ is a positive recurrent CTMC, and therefore it reaches steady-state is similar to the proof of positive recurrence of $\BE$ (See Claim \ref{cl:existence:Prior_E}) and therefore omitted.

Using a coupling argument we show that in steady-state number of $H$ agents under $\BE$ can be upper-bounded by the number of $H$ agents under $\widetilde{\BE}$.
Our motivation behind defining  $\widetilde{\BE}$ is that in $\BE$, unmatched $E$ agents in the market are competing against $H$ agents over whom they have priority. The idea is that by turning unmatched $E$ agents into $H$ ones, we expect to have  more $H$ agents waiting.
Let $H^{\widetilde{\BE}}$ be the random number of $H$ agents under $\widetilde{\BE}$ in steady-state. First in the next lemma we show that $\E[H^{\BE}] \leq \E[H^{\widetilde{\BE}}] + 1$, then we compute an upper-bound on $\E[H^{\widetilde{\BE}}]$.


\begin{lemma}
\label{lem:coup}
$\E[H^{\BE}] \leq \E[H^{\widetilde{\BE}}] + 1$.
\end{lemma}

\begin{proof}
Consider two copies of the arrival process, one under the setting of $\BE$ and one under $\widetilde{\BE}$. Let $[H^{\BE}_k, E^{\BE}_k]$ and ${H}^{\widetilde{\BE}}_k$ denote the  embedded (discrete-time) Markov chain resulting from observing the two dynamic systems at arrival epochs.
We will prove the following (stronger) result: there exists a coupling such that at any step $k$, $H_k^{\BE} + E_k^{\BE} \leq H_k^{\widetilde{\BE}} + 1$.

We start by coupling all arrivals: with probability $\frac{\lambda_H}{\lambda_E + \lambda_H}$ ($\frac{\lambda_E}{\lambda_E + \lambda_H}$), arrivals to $\BE$ and $\widetilde{\BE}$ at $k$ are both $H$ ($E$) agents. Three cases can happen:

\begin{enumerate}
\item $H_k^{\BE} + E_k^{\BE} < H_k^{\widetilde{\BE}}$; in this case, we let the two chains evolve independently. \label{it:case:1}
\item $H_k^{\widetilde{\BE}} \leq H_k^{\BE} + E_k^{\BE} \leq H_k^{\widetilde{\BE}} +1$.  Let $h = H_k^{\widetilde{\BE}}$. We consider two sub-cases: \label{it:case:2}
\begin{enumerate}[(a)]
\item The arrival at $k+1$ is an $H$ agent. We couple the events that $\widetilde{\BE}$  and $\BE$ cannot find a match as follows: we draw two independent Bernoulli random variables $\mathbb{B}_1$, $\mathbb{B}_2$ with  respective parameters of  $(1 - p_H^2)^h$ and $(1 - p_H^2)^{H_k^{\BE}}  (1 - p_H p_E)^{E_k^{\BE}}{(1 - p_H^2)^{-h}}$. $\mathbb{B}_1 = 1 $ corresponds to  the event that $\widetilde{\BE}$ cannot find a match; similarly  $ \mathbb{B}_1 \mathbb{B}_2 = 1$ corresponds to the event that $\BE$ cannot find a match. \label{it:case:2a}

%
%
\item The arrival at $k+1$ is an $E$ agent. Similarly we couple the events that $\widetilde{\BE}$  and $\BE$ cannot find a match as follows: we draw two independent Bernoulli random variables $\mathbb{B}_3$, $\mathbb{B}_4$ with  respective parameters of $(1 - p_E p_H)^h$ and ${(1 - p_E p_H)^{H_k^{\BE}}  (1 - p_E^2)^{E_k^{\BE}}}{(1 - p_E p_H)^{-h}}$. $\mathbb{B}_3 = 1$ corresponds to the event that $\widetilde{\BE}$ cannot find a match. Similarly $\mathbb{B}_3 \mathbb{B}_4 = 1$ corresponds to the event that  $\BE$ cannot find a match. \label{it:case:2b}

\end{enumerate}

\item  $H_k^{\BE} + E_k^{\BE} > H_k^{\widetilde{\BE}} +1$; in this case, we let the two chains evolve independently. \label{it:case:3}
\end{enumerate}
Having the above coupling, we finish the proof by induction. The base case $k = 0$ is trivial: $H_0^{\BE} + E_0^{\BE} = H_0^{\widetilde{\BE}} = 0$.
Suppose $H_k^{\BE} + E_k^{\BE} \leq H_k^{\widetilde{\BE}} + 1$, we show that $H_{k+1}^{\BE} + E_{k+1}^{\BE} \leq H_{k+1}^{\widetilde{\BE}} + 1$:
In Case (\ref{it:case:1}), note that the number of agents waiting can increase or decrease by at most $1$. For both cases, we have $H_{k+1}^{\BE} + E_{k+1}^{\BE} \leq H_{k+1}^{\widetilde{\BE}} + 1$. In Case (\ref{it:case:2a}), we have:
\begin{itemize}
\item[-] If $ \mathbb{B}_1 = 0$, then $H^{\widetilde{\BE}}_{k+1} = H^{\widetilde{\BE}}_k - 1$. Further if $ \mathbb{B}_1 = 0$ then $ \mathbb{B}_1 \mathbb{B}_2 = 0 $, and $H^\BE_{k+1} + E^\BE_{k+1} =  H^\BE_{k} + E^\BE_{k} - 1$.  Therefore $H_{k+1}^{\BE} + E_{k+1}^{\BE} \leq H_{k+1}^{\widetilde{\BE}} +1$
\item[-] If $ \mathbb{B}_1 = 1$, then $H^{\widetilde{\BE}}_{k+1} = H^{\widetilde{\BE}}_k + 1$. Note that number of agents under $\BE$ can increase by at most one, therefore  $H_{k+1}^{\BE} + E_{k+1}^{\BE} \leq H_{k+1}^{\widetilde{\BE}} +1$.
\end{itemize}
Similarly in Case (\ref{it:case:2b}), we have:
\begin{itemize}
\item[-] If $ \mathbb{B}_3 = 0$, then $H^{\widetilde{\BE}}_{k+1} = H^{\widetilde{\BE}}_k - 1$.  Further if $ \mathbb{B}_3 = 0$ then  $ \mathbb{B}_3 \mathbb{B}_4 = 0 $, and $H^\BE_{k+1} + E^\BE_{k+1} =  H^\BE_{k} + E^\BE_{k} - 1$; therfore $H_{k+1}^{\BE} + E_{k+1}^{\BE} \leq H_{k+1}^{\widetilde{\BE}} +1$
\item[-] If $ \mathbb{B}_3 = 1$, then $H^{\widetilde{\BE}}_{k+1} = H^{\widetilde{\BE}}_k + 1$. Note that number of agents under $\BE$ can increase by at most one, therefore and $H_{k+1}^{\BE} + E_{k+1}^{\BE} \leq H_{k+1}^{\widetilde{\BE}} +1$.
\end{itemize}

Thus, in all possible cases $H_{k+1}^{\BE} + E_{k+1}^{\BE} \leq H_{k+1}^{\widetilde{\BE}} + 1$ finishing the proof. Note that the above induction also implies that Case (\ref{it:case:3}) never occurs.

\end{proof}

\begin{proof}[Proof of Proposition \ref{prop:upper_H_prio_E}]
Observe that the CTMC $H^{\widetilde{\BE}}_t$ has the following rate matrix:
\begin{equation*}
\begin{split}
Q^\tBE(h,h+1) &= \lambda_H (1 - p_H^2)^h + \lambda_E (1 - p_E p_H)^h\\
Q^\tBE(h,h-1) &= \lambda_H (1 - (1 - p_H^2)^h) + \lambda_E (1 - (1 - p_E p_H)^h)
\end{split}
\end{equation*}

And let us define:
$$\rho(h) = \frac{Q^\tBE(h,h+1)}{Q^\tBE(h+1,h)}  =  \frac{\lambda_H (1 - p_H^2)^h + \lambda_E (1 - p_E p_H)^h}{\lambda_H (1 - (1 - p_H^2)^{h+1}) + \lambda_E (1 - (1 - p_E p_H)^{h+1})}.$$
Note that $\rho(h)$ is a decreasing function, and suppose that there exists $\eta \geq 0$ such that $\rho(\eta) < 1 - \sqrt{p_H}$ and let $\pi = \pi(h)_{h \geq 0}$ be the stationnary distribution of $H^{\widetilde{\BE}}$. Then we have for all $i \geq 0$, $\pi(\eta + i) \leq \rho(\eta)^i \pi(\eta) \leq \rho(\eta)^i $ and for all $k\geq 0$, $ \P[H^{\widetilde{\BE}} \geq \eta + k] \leq \pi(\eta) \sum_{i \geq k} \rho(\eta)^i \leq \frac{\rho(\eta)^k}{1 - \rho(\eta)}$. We  then have
\begin{equation*}
\begin{split}
\E[H^{\widetilde{\BE}}] &= \sum_{h \geq {1}} \P[H^{\widetilde{\BE}} \geq h] \\
\E[H^{\widetilde{\BE}}] &\leq (\eta + k) + \sum_{h \geq \eta + k} \P[H^{\widetilde{\BE}} \geq h] \\
\E[H^{\widetilde{\BE}}] &\leq (\eta + k) +  \frac{\rho(\eta)^k}{(1 - \rho(\eta))^2}\\
\end{split}
\end{equation*}

Let us consider the case when $\lambda_H < \lambda_E$. Let $\eta = \frac{\ln(1/u) + c \sqrt{p_H}}{p_E p_H}$. For $h \geq \eta$, we have $(1 - p_E p_H)^h \leq u \left(1 - c\sqrt{p_H} + o(\sqrt{p_H}) \right)$ and $(1 - p_H^2)^h \leq 1 - \ln(1/u) p_H^2 + o(p_H^2)$. We have $\rho(\eta) =  \frac{\lambda_H + \lambda_E u  (1 - c\sqrt{p_H}) }{\lambda_E (1 - u  (1 - c\sqrt{p_H}))} + o(\sqrt{p_H})$.  Taking $u = \frac{\lambda_E - \lambda_H}{2 \lambda_E}$ and $c = \frac{\lambda_E + \lambda_H}{2(\lambda_E - \lambda_H)}$, we get $\rho(\eta) = 1 - \sqrt{p_H} + o(\sqrt{p_H})$.
Taking $k = p_H^{-3/4}$ we get:

\begin{align}
\label{eq:priorE_case1}
\E[H^{\widetilde{\BE}}] \leq \eta + o(1/p_H) =  \frac{\ln(\frac{2 \lambda_E}{\lambda_E - \lambda_H})}{p_E p_H} + o(1/p_H), \quad \lambda_H < \lambda_E.
\end{align}
In the case when $\lambda_H > \lambda_E$, let $\eta = \frac{\ln(1/u) + c \sqrt{p_H}}{p_H^2}$. For $h \geq \eta$, we have $(1 - p_E p_H)^h = o(p_H)$ and $(1 - p_H^2)^h \leq u(1 - c \sqrt{p_H}) + o(\sqrt{p_H}).$ Taking $u = \frac{\lambda_H + \lambda_E}{2 \lambda_H}$ and $c = {1/2}$, we get $\rho(\eta) = 1 - \sqrt{p_H} + o(\sqrt{p_H}).$
Taking $k = p_H^{-3/4}$ we get:

\begin{align}
\label{eq:priorE_case2}
\E[H^{\widetilde{\BE}}] \leq \eta + o(1/p_H^2) =  \frac{\ln(\frac{2 \lambda_H}{\lambda_E + \lambda_H})}{p_H^2} + o(1/p_H^2), \quad \lambda_H > \lambda_E.
\end{align}

Putting Lemma \ref{lem:coup} and upper-bounds given in \eqref{eq:priorE_case1} and \eqref{eq:priorE_case2} together finishes the proof.

\end{proof}

\section{Missing proofs for \textit{{ChainMatch(d)} Policy}}
\label{app:chain_coupling}

\begin{proof}[Proof of Proposition \ref{prop:chain:p1}]
As pointed at the end of Subsection \ref{sec:proof_chains}, when $p_E = 1$, an arriving $E$ agent is matched immediately by a bridge agent, implying that $E^{\Cd}_t = 0$ and $H^{\Cd}_t  = H^{\tCd}_t$; consequently  Proposition \ref{prop:upper_chains} implies the limit stated in  the  Proposition \ref{prop:chain:p1}.

Further note that fixing $\lambda_H$, it is straightforward to check that function $f(\lambda_E) = \frac{\ln \left(\frac{\lambda_H}{\lambda_E} + 1 \right)}{\lambda_H}$ is decreasing for $\lambda_E > 0$. Similarly fixing $\lambda_E$, it is straightforward to check that function $g(\lambda_H) = \frac{\ln \left(\frac{\lambda_H}{\lambda_E} + 1 \right)}{\lambda_H}$ is decreasing for $\lambda_H > 0$.
\end{proof}

\begin{proof}[Proof of Corollary \ref{cor:comp}]
First we observe that $\ln \left(\frac{\lambda_H }{\lambda_E (1 - (1 - p_E)^d)} + 1\right)$ is decreasing in $d$. Therefore the worst upper-bound on $w_H^\Cd$ is for $d=1$. Next we show that:
\begin{align*}
\ln \left(\frac{\lambda_H }{\lambda_E  p_E} + 1\right)^{p_E} &  \leq \ln \left(\frac{\lambda_H }{\lambda_E  } + 1\right) \\
& \leq \ln \left(\frac{\lambda_H }{\lambda_E -\lambda_H } + 1\right),
\end{align*}
where the first inequality holds because function $f(x) := (\frac{\lambda_H/\lambda_E}{x} +1)^x$ is increasing in $x \in (0,1]$.
\end{proof}

\begin{proof}[Proof of Proposition \ref{chL}]
In steady-state, for any function $f(\cdot, \cdot)$, $0 = \mathbb{E}[f(H_k,E_{k}) - f(H_{k+1},E_{k+1})] $, and in particular for $f(h,e) = h+e$, we get
\begin{equation}
\begin{split}
0 &=(\lambda_H (1 - p_H)^d + \lambda_E (1 - p_E)^d ) \mathbb{E}[H_k + E_{k} - H_{k+1} - E_{k + 1} | H_k + E_{k} < H_{k+1} + E_{k + 1}] \\
&+  (\lambda_H (1 - (1 - p_H)^d) + \lambda_E (1 - (1 - p_E)^d) ) \mathbb{E}[H_k + E_{k} - H_{k+1} - E_{k + 1} | H_k + E_{k} \geq H_{k+1} + E_{k + 1}] \\
&= - (\lambda_H (1 - p_H)^d + \lambda_E (1 - p_E)^d)  + (\lambda_H (1 - (1 - p_H)^d) + \lambda_E (1 - (1 - p_E)^d) ) \E[L -1 \mid L \geq 1].
\end{split}
\end{equation}
This gives us:

$$\mathbb{E}[L \mid L \geq 1] = \frac{\lambda_H(1-p_H)^d + \lambda_E(1-p_E)^d }{(\lambda_H (1 - (1 - p_H)^d) + \lambda_E (1 - (1 - p_E)^d) )} + 1= \frac{\lambda_H + \lambda_E(1-p_E)^d }{\lambda_E (1 - (1 - p_E)^d) }  + 1 + o(1).$$
\end{proof}

\section{Heuristic argument for \textit{{ChainMatch(d)} Policy}}
\label{app:chains_heuristic}

In what follows we provide a heuristic analysis of the CTMC underlying the \textit{{ChainMatch(d)}}. This heuristic adds   intuition for the behavior of the policy and further  establishes what is supposedly the right constant in the case in which our theory only generates an upper bound ($p_E<1$).

We introduce an auxiliary  $3$-dimensional CTMC,  in which a chain-segment is not formed instantaneously; instead a chain-segment can only advance at certain ``tokens'' (or times) that  arrive according to a Poisson process with rate $\mu$\footnote{This process is independent from the Poisson processes guiding the arrivals of $H$ and $E$ agents.}.
We denote this auxiliary  CTMC by $\tCdm$ and its states by $[H_t^{\tCdm}, E_t^{\tCdm}, U_t] \in \N^2 \times \{0,1\}$. The first two dimensions represent the number of $H$ and $E$ agents.
The third dimension $U_t$ indicates whether a chain-segment is being conducted. Initially $U_0 =0$. Suppose at time $t_1$ the first agent $i$ arrives that is matched by an altruistic agent. At this time $U_{t_1}$ becomes $1$ indicating that a chain-segment formation is in process. The policy waits $x$ unit of time for a token to arrive to find  an agent who can be matched by $i$ (note that $x$ is exponentially distributed with rate $\mu$)\footnote{Note that it is possible that there is a new arrivals of $H$ and/or $E$ agents in the interval $(t_1, t_1+ x]$. In this heuristic, we assume that the arriving $H$ or $E$ agent is deleted. The intuition is that when $\mu$ is large, this event happens rarely.}.
If such an agent does not exist, then chain-segment is stopped; and $U_{t_1+ x}$ becomes $0$. Otherwise $U_{t_1+ x}$ remains $1$  and the same process repeats.
The transition rates of $\tCdm$ are:

\begin{subequations}
\begin{align}
&Q^\tCdm([h, e, 0 ] , [h+1, e, 0] ) = \lambda_H (1 - p_H)^d \label{eq:transtionC:1} \\
&Q^\tCdm([h, e, 0] , [h, e + 1, 0] ) = \lambda_E (1 - p_E)^d  \label{eq:transtionC:2}\\
&Q^\tCdm([h, e, 0] , [h, e, 1] ) = \lambda_H (1 - (1 -p_H)^d) + \lambda_E (1 - (1 - p_E)^d)  \label{eq:transtionC:3}\\
&Q^\tCdm([h, e, 1] , [h, e, 0] ) = \mu (1 - p_H)^{h} (1 - p_E)^{e} \label{eq:transtionC:4}\\
&Q^\tCdm([h, e, 1] , [h-1, e, 1] ) =\mu \left(1 - (1 - p_H)^{h}\right) \label{eq:transtionC:5}\\
&Q^\tCdm([h, e, 1] , [h, e-1, 1] ) =\mu (1 - p_H)^{h} \left(1 - (1 - p_E)^{e}\right) \label{eq:transtionC:6}
\end{align}
\end{subequations}

\begin{figure}[h!]
\centering
\begin{tikzpicture}
[align=center,node distance=2cm]
\node[state, minimum size=1.5cm,node distance=3cm,] (s) {\footnotesize  [h, e, 0]};
\node[state, minimum size=1.5cm, node distance=3cm, right=of s] (r) {\footnotesize [h+1, e, 0]};
\node[state, minimum size=1.5cm, node distance=3cm,left=of s] (l) {\footnotesize [h-1, e, 0]};
\node[state, minimum size=1.5cm, node distance=3cm,above=of s] (t) {\footnotesize [h, e, 1]};
\node[state, minimum size=1.5cm, node distance=3cm,above=of r] (tr) {\footnotesize [h+1, e, 1]};
\node[state, minimum size=1.5cm, node distance=3cm,above=of l] (tl) {\footnotesize [h-1, e, 1]};
\node[draw=none, minimum size=1.7cm, node distance=1.5cm, right=of r] (rr) {...};
\node[draw=none, minimum size=1.7cm, node distance=1.5cm, left=of l] (ll) {...};
\node[draw=none, minimum size=1.7cm, node distance=1.5cm, right=of tr] (trr) {...};
\node[draw=none, minimum size=1.7cm, node distance=1.5cm, left=of tl] (tll) {...};

        \draw[every loop, line width=0.6mm, >=latex]
		(s) edge[bend right, auto=right] node {{\scriptsize  $\lambda_H (1 - p_H)^d$}} (r)
		(l) edge[bend right, auto=right] node {{\scriptsize  $\lambda_H (1 - p_H)^d$}} (s)
		(s) edge[bend right, auto=right] node {{\scriptsize  $\lambda_H (1 - (1 -p_H)^d)$} \\{\scriptsize  $ + \lambda_E (1 - (1 - p_E)^d)$}} (t)
		(t) edge[bend right, auto = right] node {{\scriptsize  $\mu (1 - p_H)^{h}(1 - p_e)^{e}$}} (s)
		(tr) edge[bend right, auto = right] node[above] {{\scriptsize  $\mu \left(1 - (1 - p_H)^{h+1}\right)$}} (t)
		(l) edge[bend right, dashed, auto=left] node {} (tl)
		(tl) edge[bend right, dashed] node {} (l)
		(t) edge[bend right, auto=left] node[above] {{\scriptsize  $\mu \left(1 - (1 - p_H)^{h}\right)$}} (tl)
		(r) edge[bend right, dashed, auto=right] node {} (tr)
		(tr) edge[bend right, dashed] node {} (r)
		(trr) edge[bend right, dashed] node {} (tr)
		(tl) edge[bend right, dashed] node {} (tll)
		(r) edge[bend right, dashed] node {} (rr)
		(ll) edge[bend right, dashed] node {} (l);
\end{tikzpicture}
\caption{Transitions for the CTMC underlying $\tCdm$ in the first and third dimensions, i.e, transitions to states with a different number of $E$ agents are not shown.  Transition rates are given only for solid arrows.}
\label{fig:chains}
\end{figure}
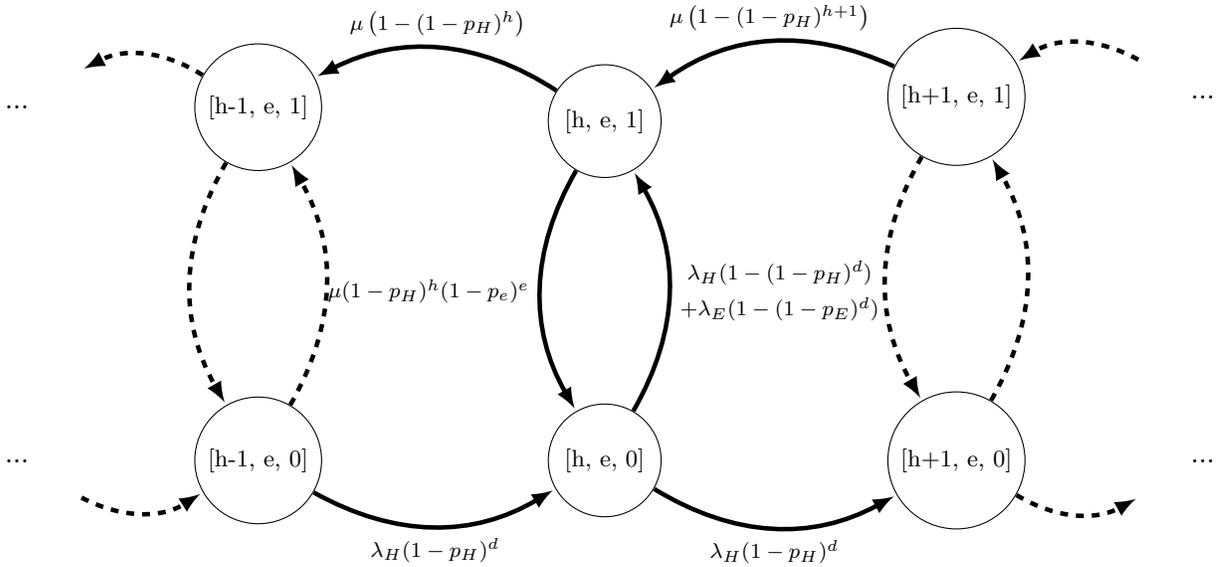

Figure \ref{fig:chains} illustrates the transitions of $\tCdm$  in $H_t^{\tCdm}$ and $U_t$ dimensions  along with the transition rates.
The rate \eqref{eq:transtionC:1} (\eqref{eq:transtionC:2}) corresponds to the event that an $H$ ($E$) agent arrives but cannot be matched by any bridge agent which happens with probability $(1 - p_H)^d$ ($(1 - p_E)^d$). The rate \eqref{eq:transtionC:3} corresponds to the case  an $H$  arrives and starts a chain-segment (occurs with probability $(1 - (1 -p_H)^d)$) or an $E$ agent arrives and starts a chain-segment (occurs  with probability $(1 - (1 - p_E)^d$);
similarly \eqref{eq:transtionC:4} represents the case where the chain-segment cannot advance any further (happens with probability $(1 - p_H)^{h} (1 - p_E)^{e}$). The last two rates correspond to adding one more $H$ and $E$  agent to the chain-segment respectively (with probability $\left(1 - (1 - p_H)^{h}\right)$ and $(1 - p_H)^{h} \left(1 - (1 - p_E)^{e}\right)$ respectively).

Similar to the heuristic argument for \textit{BilateralMatch(H)} and \textit{BilateralMatch(E)}, we can try to solve the following system of nonlinear equations (that result from setting expected drifts in all three dimensions to be zero):
\begin{subequations}
\begin{align*}
&\E \left[ \left( \lambda_H (1 - (1 - p_H)^d) + \lambda_E (1 - (1 - p_E)^d) \right)\I({U = 0}) - \mu (1 - p_H)^{H^{\tCdm}} (1 - p_E)^{E^{\tCdm}}  \I({U = 1}) \right] = 0\\
& \E\left[ \lambda_E (1 - p_E)^d  \I({U = 0}) - \mu (1 - p_H)^{H^{\tCdm}}(1 - (1 - p_E)^{E^{\tCdm}})  \I({U = 1}) \right] = 0 \\
&\E \left[ \lambda_H (1 - p_H)^d  \I({U = 0}) - \mu (1 - (1 - p_H)^{H^{\tCdm}})  \I({U = 1}) \right] = 0,
\end{align*}
\end{subequations}
where $\I(\cdot)$ is the indicator function.  We find that $\ln \left( \frac{\lambda_H + \lambda_E}{\lambda_H (1 - (1 - p_H)^d) + \lambda_E} \right)/p_H$ is an approximate solution for $\E[H^{\tCdm}]$. 
Further, note that in $\tCdm$ if an $H$/$E$ agent arrives while a chain-segment is being formed then the agent will not join the market. However, probability of such an event vanishes as $\mu \rightarrow \infty$ (i.e, forming chain-segments becomes instantaneous). Therefore, it is reasonable to approximate $\E[H^{\Cd}]$ with $\lim_{\mu \rightarrow \infty}\E[H^{\tCdm}]$, and thus with $\ln \left( \frac{\lambda_H + \lambda_E}{\lambda_H (1 - (1 - p_H)^d) + \lambda_E} \right)/p_H$.
Finally note that:
$$\lim_{p_H \rightarrow 0} \ln \left( \frac{\lambda_H + \lambda_E}{\lambda_H (1 - (1 - p_H)^d) + \lambda_E} \right) =  \ln \left( \frac{\lambda_H + \lambda_E}{ \lambda_E}\right),$$
which is
the constant in Proposition \ref{prop:chain:p1}.

We emphasize that the analysis above is only a heuristic for guessing the right constant when $p_E<1$.  We refer the reader to Section \ref{sec:app:tightness} for numerical simulations that  show the tightness of this constant.

\section{Positive recurrence proofs}
\label{app:existence_proofs}

To prove existence of a stationary distribution, we use a special case of a result from \cite{meyn1993stability}, as stated in \cite{Prieto-Rumeau2016}.

\begin{theorem}[\cite{meyn1993stability}]
Suppose that $X_t$ is an irreducible continuous time Markov chain, and suppose that there exist a nonnegative function $V$ on $S$, a function $w \geq 1$ on $S$, a finite set $C \subset S$, and constants $c > 0$ and $b \in \R$ such that

$$ \sum_{j \in S} q_{ij} V(j) \leq - c w(i) + b \cdot \mathbb{I}_C(i) \text{ for all } i \in S,$$
where $\mathbb{I}_C$ denotes the indicator function of the set $C$. Then the Markov chain $X$ is ergodic.

\end{theorem}
It is clear that all four markov chains are irreducible, so our proofs will focus on finding a suitable set $C$ and function $V$ for each case.
\subsection{Positive recurrence of $\BH$ and $\BE$.}
\begin{restatable}{claim}{claimBilatHSS}
\label{cl:existence_bilat_H}
The CTMC $[H^\BH_t, E^\BH_t]$ defined in \eqref{eq:transtion:1}, \eqref{eq:transtion:2}, \eqref{eq:transtion:3}, \eqref{eq:transtion:4}, is positive recurrent.
\end{restatable}

\begin{proof}
Let $V([h,e]) = h + e$. Observe that for a continuous time random walk, we can write for any state $i = [h,e]$:
\begin{equation*}
\begin{split}
 \sum_{j \in \N^2} q_{i,j} V(j) &=  \sum_{j \neq i} q_{i,j} (V(j) - V(i) )\\
 &= Q([h,e],[h+1,e]) - Q([h,e],[h-1,e]) + Q([h,e],[h,e + 1]) - Q([h,e],[h,e-1]). \\
 &= \lambda_H (1 - p_H^2)^{h} (1 - p_E p_H)^{e} -  \lambda_H (1 - (1 - p_H^2)^{h}) - \lambda_E (1 - (1 - p_E p_H)^{h})  + \\
 & \lambda_E (1 - p_E p_H)^{h} (1 - p_E^2)^{e} - \lambda_H (1 - p_H^2)^{h} (1 - (1 - p_E p_H)^{e}) - \lambda_E (1 - p_E p_H)^{h} (1 - (1 - p_E^2)^{e})) \\
 & = 2 \lambda_H (1 - p_H^2)^{h} (1 - p_E p_H)^{e} - \lambda_H + 2 \lambda_E (1 - p_E p_H)^{h} (1 - p_E^2)^{e} - \lambda_E.
 \end{split}
 \end{equation*}

Take $M(p_H)$ such that $(1 - p_H^2)^{M(p_H)} = \frac{1}{3}$, and  let $C = \left \{ [x,y] \text{ s.t. } x + y \leq 2 M(p_H) \right \}$, note that $C$ is finite. For any $[h,e] \not \in C$, either $h \geq M(p_H)$ or $e \geq M(p_H)$. In both cases, we have
$ \sum_{j \in \N^2} q_{i,j} V(j) \leq - \frac{\lambda_H + \lambda_E}{3}$.
\end{proof}

\begin{restatable}{claim}{claimBilatESS}
\label{cl:existence:Prior_E}
The CTMC $[H^\BE_t, E^\BE_t]$ defined in \eqref{eq:transtionE:1}, \eqref{eq:transtionE:2}, \eqref{eq:transtionE:3}, \eqref{eq:transtionE:4} is irreducible and positive recurrent.
\end{restatable}
The proof follows similar ideas as Claim \ref{cl:existence_bilat_H}.

\subsection{Positive recurrence of $\Cd$ and $\tCd$}
\begin{restatable}{claim}{claimChainSS}
\label{cl:existence_C_tilda}
Under $\Cd$, the number of $H$ and $E$ agents $[H^\Cd_t, E^\Cd_t]$ is a positive recurrent CTMC.
\end{restatable}

\begin{proof}
Similarly to the proof of Claim \ref{cl:existence_bilat_H}, let $V([h,e]) = h + e$, and $C = \left \{ [x,y] \text{ s.t. } x + y \leq 2 M \right \}$ for an appropriately chosen $M$. Consider a state $[h,e] \not \in C$, and assume first that $h \geq M$. Denoting $\Lambda = \lambda_H (1 - p_H)^d + \lambda_E (1 - p_E)^d$, we have:
\begin{equation*}
\begin{split}
 \sum_{j \in \N^2} q_{i,j} V(j) &=  \Lambda - (\lambda_H + \lambda_E - \Lambda) \sum_{[k_H, k_E] \leq [h,e]} (k_H + k_E) \P[k_H (k_E) \text{ H (E)  agents get matched}]\\
 &\leq  \Lambda - (\lambda_H + \lambda_E - \Lambda) \sum_{[k_H] \leq [h,e]} k_H \P[k_H \text{ H agents get matched}]\\
 &\leq  \Lambda - (\lambda_H + \lambda_E - \Lambda) \sum_{k \leq h} \prod_{i = h}^{h-k} (1 - (1 - p_H)^i ) \\
 &\leq  \Lambda - (\lambda_H + \lambda_E - \Lambda) \sum_{k \leq M/2} (1 - (1 - p_H)^{M/2} )^k \\
 \end{split}
 \end{equation*}
Going from the second to the third inequality bounds the number of agents matched with the number of agents matched in the first sub-chain-segment (before matching an $E$ agent). Because the function $M \rightarrow \sum_{k \leq M/2} (1 - (1 - p_H)^{M/2} )^k$ tends to infinity as $M$ grows large, we can find $M$ such that $ \sum_{j \in \N^2} q_{i,j} V(j)  \leq -1$, which concludes the proof.

The case where $h < M$ and $e \geq M$ can be treated similarly using the following observation: every-time the chain cannot match to any $H$ agent, it tries to match an $E$, and it succeeds with probability $1 - (1 - p_E)^e$ irrespective of $h$. Put together, these events constitute a set of $E$ agents that are matched sequentially, and the total length of the chain is more than the number of $E$ agents selected in that way.
\end{proof}

\begin{restatable}{claim}{claimChainsMuD}
 $H^{\tCd}_t$ is a positive recurrent CTMC.
%
\end{restatable}

The proof follows similar arguments as that of Claim \ref{cl:existence_C_tilda}.
%

\section{Lower Bound}
\label{sec:anyAlg}

\begin{definition}
We call a matching policy  \emph{anonymous Markovian} if matching decisions are made at arrival epochs, and only depend on the current compatibility graph $\mathcal{G}_t=(\mathcal{V}_t, \mathcal{E}_t)$, and are anonymous among agents {of the same type}. In that case, the market $\mathcal{G}_t$ is a continuous-time Markov chain.

We say that an anonymous Markovian policy $\mathcal{P}$ is \emph{stable} if the resulting Markov chain is ergodic and has a  stationary distribution.
In line with our previous notation, we denote $w_H^{\mathcal{P}}$ ($w_E^{\mathcal{P}}$) the expected stationary waiting times for $H$ and $E$ agents under policy $\mathcal{P}$.
\end{definition}

\begin{remark}
  Observe that all the matching policies described in this paper are anonymous Markovian.
\end{remark}

\begin{proposition}
\label{prop:lowBound}
For any $p_E \in [0,1]$, $\lambda_H > 0, \lambda_E \geq 0$,  there exists a constant $c$ such that for any $p_H > 0$, under any stable anonymous Markovian matching policy $\mathcal{P}$, $w_H^{\mathcal{P}} + w_E^{\mathcal{P}} \geq \frac{c}{p_H}$.
\end{proposition}

The proof follows ideas used in \cite{AndersonDynamic}. The main intuition is the following: Suppose the market size is too small, then an arriving agent has to wait a \emph{long} time to obtain at least one incoming edge. This long waiting time contradicts the small market size (with Little's law).
\begin{proof}
In this proof, we fix a Markovian policy $\mathcal{P}$, and we will omit the superscript notations.
Observe that under $\mathcal{P}$, the market  $\mathcal{G}_t$ only evolves at arrival epochs, and we can analyze the embedded discrete-time Markov chain resulting from observing the system at arrival epochs which we denote by $\mathcal{G}_i$. Let $n = \E[|\mathcal{V}|] = \E[H] + \E[E]$ be the expected size of the market in steady-state. Let us denote $\theta = \frac{\lambda_H}{\lambda_H + \lambda_E}$, and let $\hat{w}_H$ be the expected number of (discrete) time steps that an $H$ agent arriving in steady-state has to wait before getting matched. Little's law for discrete Markov chains implies that $\hat{w}_H = \E[H]/\theta$.

Note that it is enough to prove that there exists a constant $k$ such that $n \geq k/p_H$ (we then choose $c = \frac{k}{\max(\lambda_H, \lambda_E)}$).
Let $k$ be a constant to be defined later. Assume for contradiction that there exists $p_H$ such that $n < k/p_H$.
Let $a$ be an $H$ agent entering the market in steady-state.
Let $\mathcal{V}_a$ be the set of agents in the market when agent $a$ arrives.
Note that we assumed $\mathbb{E}[|\mathcal{V}_a|] = n < k/p_H$. Define the event $ E_1 = \{ |\mathcal{V}_a| \leq 3 n / \theta \}$. By Markov's inequality, $\mathbb{P}[E_1] \geq 1 - \frac{\mathbb{E}[|\mathcal{V}_a|] \theta}{3 n} \geq 1 - \theta/3$.

Let $\mathcal{A}_a$ be the first $3n / \theta $ arrivals after $a$, and let $E_2$ be the event that at least one agent from $\mathcal{V}_a \cup \mathcal{A}_a$ has an outgoing edge towards $a$. We have
$$\mathbb{P}[E_2] = \mathbb{P}[{\bf Bin}(|\mathcal{V}_a| + |\mathcal{A}_a|, p_H) \geq 1].\footnote{Note that we abuse the notation of {\bf Bin(u,p)} by allowing its parameters to be random variables. In this case, conditional on the event $|\mathcal{V}_a| + |\mathcal{A}_a| = u$, the random variable ${\bf Bin}(|\mathcal{V}_a| + |\mathcal{A}_a|,p)$ has a binomial distribution with parameters {$u$} and $p$.}$$
Therefore we get:
$$\mathbb{P}[E_2|E_1] \leq \mathbb{P}[{\bf Bin}(6n / \theta ,p_H)\geq 1] \leq \mathbb{P}[\text{Bin}(6k/\theta p_H,p_H) \geq 1] \leq 6k/ \theta.$$
Where the first inequality derives from the definition of $E_1$, the second uses the fact that $n \leq k/p_H$  and the third is Markov's inequality.

We now use the fact that if $a$ does not have any edge from either $\mathcal{V}_a$ or $\mathcal{A}_a$, then she must wait longer than $3n/\theta$ time steps.

$$ \hat{w}_H \geq \frac{3n}{\theta} \mathbb{P}[E_2^c] \geq  \frac{3n}{\theta} \mathbb{P}[E_2^c |E_1]\mathbb{P}[E_1]  \geq  \frac{3n}{\theta} (1-6k/ \theta)(1 - \theta/3 ) \geq \frac{3n}{\theta} (1- 6k / \theta) (2/3).$$
Thus we get:
$$n \geq \E[H] = \hat{w}_H \theta \geq 2n(1 - \frac{6k}{\theta}).$$
Therefore for $k = \frac{\theta}{24}$, we obtain a contradiction.

\end{proof}

Observe that similar to \cite{AndersonDynamic}, the above reasoning could be generalized to the case of {\it periodic} Markovian policies (see Definition 2 of \cite{AndersonDynamic}) which includes batching polices with constant batch size.


\end{document}